\documentclass[11pt]{article}
\usepackage{amsmath,amsfonts,amsthm,amssymb,xcolor}
\usepackage{thm-restate}
\usepackage{fancybox}
\usepackage{mathtools}
\usepackage{pdfsync}
\usepackage{hyperref}
\usepackage{nicefrac}
\usepackage{fullpage}
\usepackage{xspace}
\usepackage{algorithm} 
\usepackage[noend]{algpseudocode}

\usepackage{footnote}
\makesavenoteenv{tabular}
\makesavenoteenv{table}
\usepackage{threeparttable}

\newtheorem{theorem}{Theorem}[section]
\newtheorem{corollary}[theorem]{Corollary}
\newtheorem{lemma}[theorem]{Lemma}

\newtheorem{fact}[theorem]{Fact}

\theoremstyle{definition}
\newtheorem{definition}[theorem]{Definition}
\newtheorem{remark}[theorem]{Remark}

\newenvironment{fminipage}%
  {\begin{Sbox}\begin{minipage}}%
  {\end{minipage}\end{Sbox}\fbox{\TheSbox}}

\let\originalleft\left
\let\originalright\right
\renewcommand{\left}{\mathopen{}\mathclose\bgroup\originalleft}
  \renewcommand{\right}{\aftergroup\egroup\originalright}

\newcommand{\E}{\mbox{{\bf E}}}

\def\defeq{\stackrel{\mathrm{def}}{=}}

\def\union{\cup}

\def\abs#1{\left|#1  \right|}

\def\norm#1{\left\| #1 \right\|}

\def\calC{\mathcal{C}}
\def\calE{\mathcal{E}}
\def\calG{\mathcal{G}}
\def\calP{\mathcal{P}}
\def\calH{\mathcal H}

\newcommand\ppsi{\boldsymbol{\mathit{\psi}}}

\newcommand\ddelta{\boldsymbol{\delta}}
\newcommand\ddeltatil{\boldsymbol{\widetilde{\mathit{\delta}}}}
\newcommand\ddeltahat{\boldsymbol{\widehat{\mathit{\delta}}}}

\newcommand\bb{\boldsymbol{\mathit{b}}}

\newcommand\dd{\boldsymbol{\mathit{d}}}

\newcommand\ff{\boldsymbol{\mathit{f}}}

\renewcommand\gg{\boldsymbol{\mathit{g}}}

\newcommand\rr{\boldsymbol{\mathit{r}}}
\renewcommand\ss{\boldsymbol{\mathit{s}}}

\newcommand\yy{\boldsymbol{\mathit{y}}}
\newcommand\zz{\boldsymbol{\mathit{z}}}
\newcommand\xx{\boldsymbol{\mathit{x}}}

\newcommand\ddbar{\overline{\boldsymbol{\mathit{d}}}}

\renewcommand\AA{\boldsymbol{\mathit{A}}}
\newcommand\BB{\boldsymbol{\mathit{B}}}

\newcommand\DD{\boldsymbol{\mathit{D}}}

\newcommand\II{\boldsymbol{\mathit{I}}}

\newcommand\MM{\boldsymbol{\mathit{M}}}
\newcommand\LL{\boldsymbol{\mathit{L}}}
\newcommand\PP{\boldsymbol{\mathit{P}}}
\newcommand\QQ{\boldsymbol{\mathit{Q}}}

\newcommand\XX{\boldsymbol{\mathit{X}}}

\newcommand\fftil{\boldsymbol{\widetilde{\mathit{f}}}}
\newcommand\ffhat{\boldsymbol{\widehat{\mathit{f}}}}
\newcommand\gghat{\boldsymbol{\widehat{\mathit{g}}}}
\newcommand\ggtil{\boldsymbol{\widetilde{\mathit{g}}}}

\newcommand\ppsitil{\boldsymbol{\widetilde{\mathit{\psi}}}}

\newcommand\nhat{{\widehat{{n}}}}\newcommand\mhat{{\widehat{{m}}}}

\newcommand\uhat{{\widehat{{u}}}}
\newcommand\vhat{{\widehat{{v}}}}

\newcommand\Gbar{{\overline{{G}}}}
\newcommand\Ghat{{\widehat{{G}}}}
\newcommand\Hbar{{\overline{{H}}}}
\newcommand\Hhat{{\widehat{{H}}}}

\newcommand\Vhat{{\widehat{{V}}}}

\newcommand\ehat{{\widehat{{e}}}}

\newcommand\Pcal{{\mathcal{{P}}}}

\newcommand\Otil{\widetilde{O}}

\newenvironment{tight_enumerate}{
\begin{enumerate}
 \setlength{\itemsep}{2pt}
 \setlength{\parskip}{1pt}
}{\end{enumerate}}
\newenvironment{tight_itemize}{
\begin{itemize}
 \setlength{\itemsep}{2pt}
 \setlength{\parskip}{1pt}
}{\end{itemize}}

\newcommand\rea{\mathbb R}
\newcommand\nfrac[2]{\nicefrac{#1}{#2}}

\newcommand{\vone}{\boldsymbol{\mathbf{1}}}
\newcommand{\vzero}{\boldsymbol{\mathbf{0}}}

\newcommand\DDelta{\boldsymbol{\mathit{\Delta}}}

\newcommand{\str}{\mathsf{Str}}

\newcommand{\etal}{\emph{et al.}}

\newcommand\map[2]{\mathcal{M}_{#1 \rightarrow #2}}
\newcommand\obj{\calE\xspace}
\newcommand{\circapprox}{\preceq^{\text{cycle}}}

\newcommand{\fftilde}{\boldsymbol{\widetilde{f}}}
\newcommand{\Eliminate}{\textsc{Eliminate}}
\newcommand\poly{{\textrm{poly}}}  

\newcommand{\RemoveLoops}{\textsc{RemoveLoops}}
\newcommand{\SolveLoops}{\textsc{SolveLoops}}
\newcommand\kap[2]{\kappa_{#1 \to #2}}

\newcommand\nbucket{\log\frac{\norm{\rr^{\calG}}_{\infty}}{\delta}}
\newcommand\Ehat{{\widehat{{E}}}}
\newcommand\Etil{{\widetilde{{E}}}}
\newcommand\ebar{{\overline{{e}}}}
\newcommand\Ebar{{\overline{{E}}}}

\begin{document}

\title{Flows in Almost Linear Time via Adaptive Preconditioning}

\author{
  Rasmus Kyng\footnote{
Emails:
\texttt{\{rjkyng,richard.peng\}@gmail.com},
\texttt{sachdeva@cs.toronto.edu},
\texttt{di.wang@cc.gatech.edu}
}
\footnote{Supported by ONR grant N00014-18-1-2562.}
\\
  Harvard
  \and
  Richard Peng\footnotemark[1]
  \footnote{Supported in part by the National Science Foundation under Grant No. 1718533.}\\
  Georgia Tech / \\MSR Redmond
  \and
  Sushant Sachdeva\footnotemark[1] \footnote{Supported by the Natural Sciences and Engineering Research Council of Canada (NSERC), and a Connaught New Researcher	award.}\\
  UToronto
  \and
  Di Wang\footnotemark[1] \footnotemark[3]\\
  Georgia Tech
}

\maketitle
\thispagestyle{empty}

\begin{abstract}
  We present algorithms for solving a large class of flow and regression
  problems on unit weighted graphs to $(1 + 1 / poly(n))$ accuracy in
  almost-linear time.
  These problems include $\ell_p$-norm minimizing flow for $p$ large
  ($p \in [\omega(1), o(\log^{2/3} n) ]$), and their duals,
  $\ell_p$-norm semi-supervised learning for $p$ close to $1$. 

As $p$ tends to infinity, $\ell_p$-norm flow and its dual tend to max-flow
and min-cut respectively. Using this connection and our algorithms,
we give an alternate approach for approximating undirected max-flow,
and the first almost-linear time approximations of discretizations of
total variation minimization objectives.

This algorithm demonstrates that many tools previous viewed as limited
to linear systems are in fact applicable to a much wider range of convex
objectives.
It is based on the the routing-based solver for Laplacian
linear systems by Spielman and Teng (STOC '04, SIMAX  '14),
but require several new tools: adaptive non-linear preconditioning,
tree-routing based ultra-sparsification for mixed
$\ell_2$ and $\ell_p$ norm objectives, and decomposing graphs
into uniform expanders.

\end{abstract}

\newpage
\thispagestyle{empty}
\tableofcontents

\newpage
\setcounter{page}{1}


\section{Introduction}

Graphs are among the most ubiquitous representations of data,
and efficiently computing on graphs is a task central to operations
research, machine learning, and network science.
Among graph algorithms, network flows have been extensively studied
~\cite{EdmondsK72,Karzanov73,EvenT75,GoldbergT88,GoldbergR98,Schrijver02,
Hochbaum08,ChristianoKMST10,HochbaumO13,Orlin13,GoldbergT14},
and have wide ranges of applications~\cite{KolmogorovBR07,LiSBG13,PengZZ13}.
Over the past decade, the `Laplacian paradigm' of designing graph algorithms
spurred a revolution in the best run-time upper bounds for many
fundamental graph optimization problems.
Many of these new graph algorithms incorporated numerical primitives:
even for the $s$-$t$ shortest path problem in graphs with negative edge
weights, the current best running times~\cite{CohenMTV17}
are from invoking linear system solvers.

This incorporation of numerical routines~\cite{DaitchS08,ChristianoKMST10}
in turn led to a dependence on $\epsilon$, the approximation accuracy.
While maximum flow and transshipment problems on undirected graphs can now be
approximated in nearly-linear time
\cite{KelnerLOS14,Sherman13,BeckerKKL17,Peng16,
  Sherman17a,Sherman17b}
(and the distributed setting has also been studied
\cite{GhaffariKKLP15, BeckerKKL17}),
these algorithms are \emph{low accuracy} in that their running times have
factors of $1 / \epsilon$ or higher.
This is in contrast to \emph{high accuracy} solvers for linear systems and convex programs,
which with $polylog(n)$ overhead give $1 / poly(n)$-approximate solutions.
Prior to our result, such \emph{high accuracy} runtime bounds for problems beyond linear systems
all utilize second order methods~\cite{DaitchS08,Madry13,LeeS14,CohenMTV17,AllenZhuLOW17,BubeckCLL18}
from convex optimization.

The main contribution of this paper is giving almost-linear time, \emph{high accuracy}
solutions to a significantly wider range of graph optimization problems that can
be viewed as interpolations between maximum flow, shortest paths, and graph-structured
linear systems.
Our unified formulation of these problems is based on the following unified formulation
of flow/path problems as norm minimization over a demand vector
$\bb \in \rea_{\ge 0}^{V}$:
\begin{equation}
\min_{\text{flow $\ff$ with residue $\bb$}} \norm{\ff}_{\odot}.
\label{eq:Flow}
\end{equation}
In particular, when $\norm{\cdot}_{\odot}$ is the $\ell_{\infty}$-norm,
this formulation is equivalent finding the flow of minimum congestion,
which is in turn equivalent to computing maximum flows and bipartite matchings
in unit capacitated graphs~\cite{Madry11:thesis}.
Our main result is that for any $p \geq 2,$ given weights
$\rr \in \rea_{\ge 0}^E,$ a ``gradient'' $\gg \in \rea^E$, and a
demand vector $\bb \in \rea^V$ (with $\bb^{\top} \vone = 0$), we can
solve
\begin{align}
  \label{eq:intro:problem}
  \min_{\text{flow $\ff$ with residue $\bb$}} \sum_e \gg_e \ff_e + \rr_e \ff_e^2 +
  \abs{\ff_e}^p,
\end{align}
to $\nfrac{1}{\poly(n)}$ additive error in time
$2^{O(p^{3/2})} m^{1+O(\frac{1}{\sqrt{p}})}$.
We will formally state this result as Theorem~\ref{cor:smoothedpflows}
at the start of Section~\ref{subsec:apps},
and discuss several of its applications in flows, semi-supervised learning,
and total variation minimization.

We believe that our algorithm represents a new approach to designing
\emph{high accuracy} solvers for graph-structured optimization problems.
A brief survey of relevant works is in Section~\ref{subsec:Related}:
previous \emph{high accuracy} algorithms treat linear systems
as the separation between graph theoretic and numerical components:
the outer loop adjusts the numerics,
while the inner loop quickly solves the resulting linear systems using the
underlying graph structures.
Our result, in contrast, directly invoke analogs of linear system solving
primitives to the non-linear (but still convex) objective functions,
and no longer has this clear separation between graph theoretic and
numerical components.

We will overview key components of our approach, as well as how they are combined,
in Section~\ref{subsec:Overview}.
Discussions of possible avenues for addressing shortcomings of our result,
namely the exponential dependence on $p$, the restriction to unweighted graphs,
and gap between $\ell_{\sqrt{\log{n}}}$-norm flow and $\ell_{\infty}$ are
in Section~\ref{subsec:open}.

\subsection{Main Results and Applications}
\label{subsec:apps}

The formal formulation of our problem relies on the following objects
defined on a graph $G = (V, E)$ with $n$ vertices and $m$ edges:
\begin{tight_enumerate}
\item edge-vertex incidence matrix $\BB$,
\item a vector $\bb$ indicating the required residues on vertices
(satisfying $\vone^T\bb=0$), and
\item a norm $p$ as finding a flow $\ff$
with demands $\bb$ that minimize a specified norm $\|\cdot\|$.
\end{tight_enumerate}
The normed flow problem that we solve can then be formulated as:
\begin{align}
  \min_{\BB^{\top} \ff = \bb} \sum_e \gg_e \ff_e + \rr_e \ff_e^2 +
  \abs{\ff_e}^p,
\tag{\ref{eq:intro:problem}}
\end{align}
 Using
$\norm{\ff}_{2,\rr} = \sqrt{ \sum_e \rr_e \ff_e^2 }$ to denote the
$\rr$-weighted 2-norm, the objective can also be viewed as
$\gg^{\top} \ff + \norm{\ff}_{2,\rr}^2 + \norm{\ff}_{p}^p$.
Let $\text{val}(\ff)$ denote value of a flow $\ff$
according to the above objective,
and let 
$\text{OPT}$ denote value of the optimal solution to Problem~\eqref{eq:intro:problem}.
Our main technical result is the following statement which we prove as corollary
of our main technical theorem in
Section~\ref{subsec:recursive}.
\begin{restatable}[Smoothed $\ell_p$-norm flows]{theorem}{smoothedpflows}
  \label{cor:smoothedpflows}
  For any $p \geq 2,$ given weights $\rr \in \rea_{\ge 0}^E,$ a
  ``gradient'' $\gg \in \rea^E$, a demand vector $\bb \in \rea^V$
  (with $\bb^{\top} \vone = 0$), and an initial solution $\ff^{(0)}$
  such that all parameters are bounded by $2^{\poly(\log n)},$
  we can compute a flow $\fftilde$ satisfying demands $\bb,$
  \emph{i.e.}, $\BB^{G\top} \fftilde = \bb,$ such that
  \[
  \text{val}(\fftil) - \text{OPT} \leq \frac{1}{\poly(m)} \left(
    \text{val}(\ff^{(0)}) - \text{OPT}\right) + \frac{1}{\poly(m)}
\]
in  $2^{O(p^{\nfrac{3}{2}})} m^{1+O(\frac{1}{\sqrt{p}})}$ time, where
  $m$ denotes the number of edges in $G.$
\end{restatable}

\subsubsection{$\ell_{p}$-Norm Flows}
\label{subsubsec:pFlows}

From this, we also get a (slightly simpler) statement
about $\ell_p$-norm flows. 
\begin{restatable}[$\ell_p$-norm flows]{theorem}{pflows}
\label{cor:pflows}
  For any $p \ge 2,$ given an unweighted graph $G(V, E)$ 
  and demands $\bb$, using the routine  $\textsc{pFlows}(\calG, \bb)$ (Algorithm~\ref{alg:pFlows})
  we can compute a flow $\fftilde$ satisfying
  satisfying $\bb,$ \emph{i.e.}, $\BB^{G\top} \fftilde = \bb,$ such that
  \[\norm{\fftilde}^{p}_p \le \left( 1+ \frac{1}{\poly(m)} \right)
    \min_{\ff : \BB^{G\top} \ff = \bb}
    \norm{\ff}^{p}_{p}.\]
 in  $2^{O(p^{\nfrac{3}{2}})} m^{1+O(\frac{1}{\sqrt{p}})},$ time, where
  $m$ denotes the number of edges in $G.$
\end{restatable}
This corollary is also proven in Section~\ref{subsec:recursive}.

Picking $\gg,\rr = \vzero$ gives us an
$2^{O(p^{3/2})} m^{1+O(\frac{1}{\sqrt{p}})}$ time high-accuracy
algorithm for $\ell_{p}$-norm minimizing flows on unit weighted undirected
graphs ($p \geq 2$).  For large $p$, e.g. $p = \sqrt{\log n}$ this is
an $m^{1+o(1)}$ time algorithm, and to our knowledge the first almost
linear time high-accuracy algorithm for a flow problem other than
Laplacian solvers ($\ell_2$) or shortest-paths ($\ell_1$).

\subsubsection{Semi-Supervised Learning on Graphs.}
\label{subsubsec:pLaplacians}

Semi-supervised learning on graphs in machine learning is often based
on solving an optimization problem where voltages (labels) are fixed
at some vertices in a graph the voltages at remaning nodes are chosen
so that some overall objective is minimized (e.g. $\ell_{p}$-norm of the
vector of voltage differences across
edges)~\cite{AlamgirL11,KyngRSS15,ElalaouiCRWJ16}.
Formally, given a graph $G=(V,E)$ and a labelled subset of the nodes
$T \subset V$ with labels $\ss_{T} \in \rea^T$, we can write the
problem as 
\begin{equation}
\min_{\xx \in \Re^{V} \mid \xx_{T} = \ss_{T}}
\sum_{u \sim v} \abs{\xx_u - \xx_{v}}^{p}.
\label{eq:semisupervised}
\end{equation}
By converting this problem to its dual, we get an almost linear time algorithm for
solving it to high accuracy, provided the initial voltage problem uses $p$ close to 1:
In this case, voltage solutions are ``cut-like''.
Given $p < 2$, we get a solver that computes a $(1+1/\poly(m))$
multiplicative accuracy solution
in time $2^{O((\frac{1}{p-1})^{\nfrac{3}{2}})}
m^{1+O(\sqrt{p-1})}$.
For $p = 1+\frac{1}{\sqrt{\log n}}$, this is time is bounded by $m^{1+o(1)}$.

Converting the dual of Problem~\eqref{eq:semisupervised} into a form
solvable by our algorithms requires a small transformation, which we
describe in Appendix~\ref{sec:semisupervised}.





\subsubsection{Use as Oracle in Conjunction with Multiplicative Weight Updates}
\label{subsubsec:Oracle}

The mixed $\ell_{2}^2$ and $\ell_p^p$ objective in our Problem
\eqref{eq:intro:problem} is useful for building oracles to use in
multiplicative weight update algorithms based on flows, as they appear
in \cite{ChristianoKMST10, AdilKPS19}.
Assume we are looking to solve some problem to $(1+\epsilon)$-accuracy
measured as multiplicative error, and let us assume
$\frac{1}{\poly(m)} < \epsilon < 0.5$.
Specifically we can solve for the following objective subject to
certain linear constraints.
\begin{equation}
\label{eq:mwuoracle}
\sum_{e} \rr_e \ff_e^2 + \frac{\epsilon \norm{\rr}_1 }{m}
\abs{\ff_e}^{p}.
\end{equation}
This gives an oracle for several problems.
Algorithms based on oracle solutions to this type of objective work by
noting that any $\ff$ with $\norm{\ff}_{\infty} \leq 1$ gives an
objective value at most
\[
\sum_{e} \rr_e \ff_e^2 + \frac{\epsilon \norm{\rr}_1 }{m} \abs{\ff_e}^{p}
\leq
\left( 1 + \epsilon \right) \norm{\rr}_1.
\]
Since such a flow must exist in the context where the oracle is
applied, the optimum flow must also meet this bound.
Now, if we compute a $( 1 + 0.01\epsilon)$ approximately optimal solution to
this problem, it must satisfy
\[
\sum_{e} \rr_e \ff_e^2 + \frac{\epsilon \norm{\rr}_1 }{m} \abs{\ff_e}^{p}
\leq
\left( 1 + 1.1\epsilon \right) \norm{\rr}_1.
\]
By Cauchy-Schwarz, we get
$\sum_{e} \rr_e \abs{\ff_e}  \leq 
\sqrt{ \norm{\rr}_1 \sum_{e} \rr_e \ff_e^2 } \leq \left( 1 +
  1.1\epsilon \right) \norm{\rr}_1$, which tells us the oracle is
``good-on-average'' according to the weights $\rr$.
The objective value also implies for every edge that
\[
\frac{\epsilon \norm{\rr}_1 }{m} \abs{\ff_e}^{p}
\leq
\left( 1 + 1.1\epsilon \right) \norm{\rr}_1
\leq
2 \norm{\rr}_1,
\]
which simplifies to:
\begin{equation}
\label{eq:mwuwidth}
\abs{\ff_e}
\leq
\left( m / \epsilon \right)^{1/p}
\leq
m^{o\left(1\right)}
\end{equation}
when we set $p = \log^{0.1}n$.
This is the width of the oracle, and together these conditions
demonstrate that the oracle suffices for a multiplicative weights algorithm and bounds
the number of calls to the oracle by 
$m^{o\left(1\right)} \poly(1/\epsilon)$.

This oracle has multiple uses:
\begin{description}
\item[Approximate undirected maximum flow.] Using the oracle, we can approximate maximum flow using~\cite{ChristianoKMST10},
giving an algorithm for undirected maximum flow that is not 
based on oblivious routings unlike other fast algorithms for
approximate maximum flow~\cite{Sherman13,KelnerLOS14,Peng16}.
Our algorithm obtains almost-linear time, albeit
only for unit weighted graphs.
\item[Isotropic total variation denoising.] Using our algorithm, we can give the first almost linear time,
  low accuracy algorithm for total variation denoising on unit
  weighted graphs~\cite{RudinOF92,ZhuWC10}.
While there has been significant advances in image processing
since the introduction of this objective, it still remains
a representative objective in pixel vision tasks.
The total variation objectives can be viewed as variants
of semi-supervised learning on graphs: Given a ``signal''
vector $\ss$ which corresponds to noisy observations of pixels of an image, we want
to find a denoised version of $\ss$, which we refer to as $\xx$.
The denoised output $\xx$ should minimize an objective that measures
both the between pixels in $\xx$ that are close
to each other in the image (which should be small), and the difference
between $\xx$ and $\ss$ (which should also be small).
The most popular version of this problem, known as isotropic total
variation denoising, allows the input to specify a collection of
groups of pixels with connections inside each group $i$ given by a set
of edges $E_i$, and 
asks that 1) the denoised pixels are close in an $\ell_2$ sense to
the measured signal, 2) in each group, the standard deviation between
denoised pixels is not too high.
These goals are expressed in the objective
\[
\sum_{u} \left( \xx_u - \ss_{u}\right)^2
+ \sum_{i} \sqrt{\sum_{e \in E_{i}} \left(\xx_{u} - \xx_{v}\right)^2}.
\]
The dual of this problem is grouped flows, which is finding $\ff$
such that $\BB^{\top} \ff = \dd$ and for edge sets $E_i$,
\[
\norm{\ff_{E_i}}_2^2 \leq 1.
\]
Our oracle gives the first routine for approximate isotropic $TV$ denoising
that runs in almost linear time.
The previous best running time was
about $m^{4/3}$~\cite{ChinMMP13}.
\end{description}




\subsection{Related Work}
\label{subsec:Related}

Network flow problems have a long history of motivating broader
developments in algorithms, including the introduction of
strongly polynomial time as a benchmark of algorithmic
efficiency~\cite{Edmonds65,EdmondsK72},
the development of tree data
structures~\cite{GalilN79,SleatorT83,SleatorT85}, and
randomized graph algorithms and graph approximations~\cite{KargerS96,BenczurK96}.
For general capacities, the best strongly polynomial time algorithms
run in about quadratic time due to the flow decomposition barrier~\cite{EdmondsK72,GalilN79,GoldbergT88,HochbaumO13,Orlin13},
which says that the there exists graphs where the
path decomposition of an optimum flow must have quadratic size.

The flow decomposition barrier suggest that sub-quadratic time algorithms
for network flows should decompose solutions numericallly, and this
has indeed been the case in the development of
such algorithms~\cite{GoldbergR98,GoldbergT14}.
These numerical approaches recently culminated in nearly-linear time algorithms
for undirected maximum flow and transshipment
(the $\ell_1$ case of Problem~(\ref{eq:Flow})),
yielding nearly-linear time
algorithms~\cite{ChristianoKMST10,Sherman13,KelnerLOS14,Peng16,Sherman17a,Sherman17b}. 
Much of these progress were motivated by the development of
nearly-linear time high-accuracy solvers for Laplacian linear
systems~\cite{SpielmanTengSolver:journal, KoutisMP12,
	KelnerOSZ13, LeeS13, KyngS16}, whose duals, electrical flows
are the $\ell_{\infty}$ case of Problem~(\ref{eq:Flow}).
Such solvers can in turn be used to give the current best
high accuracy flow algorithms.
For graphs with polynomially bounded capacities,
the current best running time is $\Otil(m \sqrt{n})$
due to Lee and Sidford~\cite{LeeS14}.
On sparse graphs, this bound still does not break the long-standing
$O(n^{1.5})$ barrier dating back to the early 70s~\cite{HopcroftK73,Karzanov73,EvenT75}.
Recently Madry~\cite{Madry13,Madry16} broke this barrier on unit
capacitated graphs, obtaining $\Otil(m^{10/7})$ running time.

Our result has in common with all previous results on almost-linear time
optimization problems on graphs~\cite{KelnerLOS14,Sherman13,BeckerKKL17,Peng16,
Sherman17a,Sherman17b} in that it is based on white-box modifications
of a linear system solver.
In particular, our high level algorithmic strategy in creating edge
and vertex reductions is identical to the first nearly-linear time
solver by Spielman and Teng~\cite{SpielmanTengSolver:journal}.
Much of this similarity is due to the lack of understanding of more
general versions of key components: some possibilities for simplifying
the result will be discussed in Section~\ref{subsec:open}.
On the other hand, our algorithms differ from previous adaptations
of solvers in that it obtains high accuracy
\footnote{
The nearly-linear time matrix scaling algorithm~\cite{CohenMTV17}
has a  linear dependence on the condition number $\kappa$,
while convex optimization methods for matrix scaling have dependencies
of $\log \kappa$ instead.
}.
This requires us to tailor the scheme to the residual problems
from the $p$-norm iterative methods, and results in us taking a more
numerical approach, instead of the more routing and path embedding-based
approaches utilized in similar adaptations of  Spielman and
Teng~\cite{SpielmanTengSolver:journal}
to cuts~\cite{Madry10},
flows~\cite{Sherman13,KelnerLOS14}, and shortest paths~\cite{BeckerKKL17}.

The development of high-accuracy algorithms for $p$-norm minimization
that are faster than interior point methods (IPMs)~\cite{NesterovN94}
was pioneered by the recent work of Bubeck~\etal~\cite{BubeckCLL18}
which introduced the $\gamma$-functions that were also used in
\cite{AdilKPS19}.  However, the methods in~\cite{BubeckCLL18} are conceptually
similar to interior point methods (IPMs)~\cite{NesterovN94}
(as in they are \emph{homotopy} methods).
Their runtime for large $p$ behaves essentially like IPMs,
requiring about $m^{\nfrac{3}{2}-o(1)}$ time for solving $p$-norm flow
problems, whereas the limiting behavior of our result is about
$m^{1 + o(1)}$.

\subsection{Overview}
\label{subsec:Overview}

At a high level, our approach can be viewed as solving a graph optimization
problem as a linear system.
This is done by combining the numerical methods for
$\ell_p$-norms by Adil~\etal~\cite{AdilKPS19} with the recursive
preconditioning of graph structured linear systems by Spielman and
Teng~\cite{SpielmanTengSolver:journal}.
Many conceptual obstacles arise in trying to realize this vision,
preventing us from adopting later Laplacian linear solvers that have
greatly simplified the result of Spielman and Teng.
The main one is the lack of concentration theory for the smoothed $p$-norm
objectives integral to our algorithms: these concentration arguments are
at the core of all subsequent improvements to Laplacian
solver algorithms~\cite{KoutisMP11, KelnerOSZ13,LeePengSpielman, KyngS16}.

Our starting point is a recent result involving a subset of the
authors~\cite{AdilKPS19} that significantly generalized the phenomenon
of \emph{high-accuracy} numerical methods.
In particular, this method is applicable to general $\ell_{p}$-norm optimization
problems, for all $p$ that are bounded away from $1$ and $\infty$.
It also opens up a major question: can we develop an appropriate notion
of \emph{preconditioning}, the other central ingredient of fast
solvers for linear systems, applicable to $\ell_{p}$-norms?
We resolve this question in the affirmative, and develop a theory of
preconditioning that works for a wide class of non-linear
problems in Section~\ref{sec:numerical}.
In particular, we show that the second and p\textsuperscript{th} order
terms from the main formulation in Equation~\ref{eq:intro:problem}
form a class of functions that's closed under taking residual problems.
We will formally define these as smoothed $\ell_{p}$-norms in
Section~\ref{subsec:SmoothNorm}.

The crux of our problem then becomes designing preconditioners that
interact well with these smoothed $\ell_{p}$-norms.
Here it's worth noting that earlier works on preconditioning for non-linear
(maximum) flow problems all relied on \emph{oblivious routing} which gives
rise to linear preconditioners.
Such an approach encounters a significant obstacle with $\ell_{p}$ norms:
consider changing a single coordinate from, say $1$, to $(1+\delta)$:
\begin{tight_itemize}
\item If the update $\delta$ is much smaller than $1$ in absolute value, the change in the
objective from $1^p$ to $(1+\delta)^p$ is dominated by terms that
scale as $\delta$ and $\delta^2$.
\item However, if the update is much
larger than $1$, the change is dominated by a $\delta^p$ term.
\end{tight_itemize}
This means that good preconditioning across small and large updates is
inherently highly dependent on the current coordinate value.

This example captures the core difficulties of our
preconditioned iterative methods for smoothed $\ell_p$-norm problems,
which heavily rely on both the second and pth power terms the objective functions.
It means our graph theoretic components must simultaneously control terms of
different degrees (namely scaling as $\delta$, $\delta^2$, and $\delta^p$)
related to the flows on graphs.
Here our key idea is that unit-weighted graphs have
``multi-objective low-stretch trees'' that simultaneously preserve
the $\delta^{2}$ and $\delta^{p}$ terms, while the linear (gradient)
terms can be preserved exactly when routing along these trees.
Here a major difficulty is that the tree depends on the second order
derivatives of the current solution point, and thus must continuously
change as the algorithm proceeds.
Additionally, after rerouting graph edges along the tree, we need to
sparsify the graph according to weights defined by the same second derivatives at the
current solution, which makes the adaptive aspect of the algorithm
even more important.
We defer the construction of our adaptive preconditioner to
Section~\ref{sec:graph}, after first
formally defining our objective functions in Section~\ref{sec:prelims},
and introducing numerical methods based on them in Section~\ref{sec:numerical}.

%


\subsection{Open Questions}
\label{subsec:open}

We expect that our algorithm can be greatly simplified and adapted to
non-unit weight graphs in ways similar to the sampling based solvers
for Laplacian linear systems~\cite{KoutisMP10,KelnerOSZ13,KyngS16}.
The current understanding of concentration theory
for $\ell_{p}$ norms rely heavily on tools from functional analysis~\cite{CohenP15}:
generalizing these tools to smoothed $\ell_{p}$-norm objectives is
beyond the scope of this paper.

A major limitation of our result is the restriction to unit capacitated graphs.
We believe this limitation is inherent to our approach of constructing preconditioners
from trees: for general weights, there are cases where no tree can simultaneously
have small stretch w.r.t. $\ell_{2}$-norm and $\ell_{p}$-norm weights.
We believe that by developing a more complete theory of elimination and
sparsification for these objectives, it will be possible to sparsify
non-unit weight graphs, and develop solvers following the patterns of
non-tree based Laplacian solvers~\cite{PengS14, LeePengSpielman, KyngS16}.

We also believe that the overall algorithmic approach established here
is applicable far beyond the class of objective functions studied in this paper.
Here a direct question is whether the dependency on $p$ can be improved
to handling $\ell_{m}$ flows, which in unit weighted graphs imply
maximum flows.
The exponential dependence on $p$ has already been shown to be improvable
to about $\Otil(p^2)$~\cite{Sachdeva:communication}.
For even larger values of $p$, a natural approach is to use homotopy methods
that modify the $p$ values gradually.
Here it is also plausible that our techniques, or their possible generalizations
to weighted cases, can be used as algorithmic building blocks.


\section{Preliminaries}
\label{sec:prelims}

\subsection{Smoothed $\ell_{p}$-norm functions}
\label{subsec:SmoothNorm}
We consider $p$-norms smoothed by the addition of a
quadratic term. First we define such a smoothed $p\textsuperscript{th}$-power on $\rea.$
\begin{definition}[Smoothed $p\textsuperscript{th}$-power]
  Given $r, x \in \rea, r \ge 0$ define the $r$-smoothed $s$-weighted
  $p\textsuperscript{th}$-power of $x$ to be
  \[h_{p}(r, s, x) = r x^2 + s\abs{x}^p.
  \]
\end{definition}
\noindent This definition can be naturally extended to vectors to
obtained smoothed $\ell_{p}$-norms.
\begin{definition}[Smoothed $\ell_{p}$-norm]
  Given vectors $\xx \in \rea^m, \rr \in \rea^{m}_{\ge 0},$ and a
  positive scalar $s \in \rea_{\ge 0},$   define the
  $\rr$-smooth $s$-weighted $p$-norm of $\xx$ to be
  \[h_{p}(\rr, s, \xx) = \sum_{i=1}^m h_{p}(\rr_i, s, \xx_i) =
    \sum_{i=1}^{m} ( \rr_i \xx_i^2 + s\abs{\xx_i}^{p}).
  \]
\end{definition}

\subsection{Flow Problems and Approximation}
\label{sec:overview}
We will consider problems where we seek to find flows minimizing
smoothed $p$-norms. We first define these problem instances.
\begin{definition}[Smoothed $p$-norm instance]
  A \emph{smoothed $p$-norm instance} is a tuple
  $\calG,$
  \[\calG\defeq (V^{\calG}, E^{\calG}, \gg^{\calG}, \rr^{\calG}, s^{\calG}),\]
  where $V^{\calG}$ is a set of vertices, $E^{\calG}$ is a set of
  undirected edges on $V^{\calG},$ the edges are accompanied
  by a gradient, specified by $\gg^{\calG} \in \rea^{E^{\calG}},$ 
  the edges have $\ell_2^{2}$-resistances
  given by
  $\rr^{\calG}
  \in \rea^{E^{\calG}}_{\ge 0},$ and $s \in \rea_{\ge 0}$ gives the
  $p$-norm scaling.
\end{definition}
\begin{definition}[Flows, residues, and circulations]
  Given a smoothed $p$-norm instance $\calG,$ a vector
  $\ff \in \rea^{E^{\calG}}$ is said to be a flow on $\calG$. A flow
  vector $\ff$ satisfies residues $\bb \in \rea^{V^{\calG}}$ if
  $\left( \BB^{\calG} \right)^{\top} \ff = \bb,$ where
  $\BB^{\calG} \in \rea^{E^{\calG} \times V^{\calG} }$ is the
  edge-vertex incidence matrix of the graph $(V^{\calG},E^{\calG}),$
  \emph{i.e.},
  $\left( \BB^{\calG} \right)^{\top}_{(u,v)} = \vone_{u} - \vone_{v}.$

  A flow $\ff$ with residue $\vzero$ is called a circulation on $\calG$.
\end{definition}
Note that our underlying instance and the edges are
undirected. However, for every undirected edge $e=(u,v) \in E$, we
assign an arbitrary fixed direction to the edge, say $u \to v,$ and
interpret $\ff_e \ge 0$ as flow in the direction of the edge from $u$
to $v,$ and $\ff_e < 0$ as flow in the reverse direction. For
convenience, we assume that for any edge $(u,v) \in E,$ we have
$\ff_{(u,v)} = - \ff_{(v,u)}.$
\begin{definition}[Objective, $\obj^{\calG}$]
  Given a smoothed $p$-norm instance $\calG,$ and a flow $\ff$ on
  $\calG,$ the associated objective function, or the energy, of $\ff$
  is given by
  \[\obj^{\calG}(\ff) = \left( \gg^{\calG} \right)^{\top} \ff -
  h_p(\rr, s, \ff).\]
\end{definition}

\begin{definition}[Smoothed $p$-norm flow / circulation problem]
  Given a smoothed $p$-norm instance $\calG$ and a residue vector
  $\bb \in \rea^{E^{\calG}},$ the \emph{smoothed $p$-norm flow
    problem} $(\calG, \bb)$, finds a flow $\ff \in \rea^{E^{\calG}}$
  with residues $\bb$ that maximizes $\obj^{\calG} (\ff),$
  \emph{i.e.},
  \[
    \max_{\ff: (\BB^{\calG})^{\top} \ff = \bb} \obj^{\calG} \left( \ff \right).
  \]
  If $\bb = \vzero,$ we call it a \emph{smoothed $p$-norm circulation
    problem}.
\end{definition}
Note that the optimal objective of a smoothed $p$-norm circulation
problem is always non-negative, whereas for a smoothed $p$-norm flow
problem, it could be negative.
\subsection{Approximating Smoothed $p$-norm Instances}
Since we work with objective functions that are non-standard (and not
even homogeneous), we need to carefully define a new notion of
approximation for these instances.
\begin{definition}[$\calH \preceq_{\kappa} \calG$]
  For two smoothed $p$-norm instances, $\calG, \calH,$ we write
  $\calH \preceq_{\kappa} \calG$ if there is a linear map
  $\map{\calH}{\calG}: \rea^{E^{\calH}} \rightarrow \rea^{E^{\calG}}$
  such that for every flow $\ff^{\calH}$ on $\calH,$ we have that
  $\ff^{\calG} = \map{\calH}{\calG} (\ff^{\calH})$ is a flow on
  $\calG$ such that
  \begin{enumerate}
  \item $\ff^{\calG}$ has the same residues as $\ff^{\calH}$ \emph{i.e.},
    $(\BB^{\calG})^{\top} \ff^{\calG} = (\BB^{\calH})^{\top}
    \ff^{\calH},$ and
  \item has energy bounded by:
    \[
      \frac{1}{\kappa} \obj^{\calH} \left( \ff^{\calH} \right) \leq
      \obj^{\calG} \left( \frac{1}{\kappa} \ff^{\calG} \right).
    \]
  \end{enumerate}
\end{definition}
For some of our transformations on graphs, we will be able to prove
approximation guarantees only for circulations. Thus, we define the
following notion restricted to circulations.
\begin{definition}[$\calH \circapprox_{\kappa} \calG$]
  For two smoothed $p$-norm instances, $\calG, \calH,$ we write
  $\calH \circapprox_{\kappa} \calG$ if there is a linear map
  $\map{\calH}{\calG} : \rea^{E^{\calH}} \rightarrow \rea^{E^{\calG}}$
  such that for any circulation $\ff^{\calH}$ on $\calH$, \emph{i.e.},
  $(\BB^{\calH})^{\top}\ff^{\calH} = \vzero,$ the flow
  $\ff^{\calG} = \map{\calH}{\calG}(\ff^{\calH})$ is a circulation,
  \emph{i.e.},
  $(\BB^{\calG})^{\top} \ff^{\calG} = \vzero, $
  and satisfies
  \[
    \frac{1}{\kappa} \obj^{\calH} \left( \ff^{\calH} \right) \leq
    \obj^{\calG} \left( \frac{1}{\kappa} 
      \ff^{\calG} \right).
  \]
  Observe that $\calH \preceq_{\kappa} \calG$ implies
  $\calH \circapprox_{\kappa} \calG.$
\end{definition}
These definitions satisfy most properties that we want from
comparisons.
\begin{restatable}[Reflexivity]{lemma}{identity}
  \label{lem:approximations:identity}
  For every smoothed $p$-norm instance $\calG,$ and every
  $\kappa \ge 1$, $\calG \preceq_{\kappa} \calG$ and
  $\calG \circapprox_{\kappa} \calG$ with the identity map.
\end{restatable}
It behaves well under composition.
\begin{restatable}[Composition]{lemma}{composition}
  \label{lem:approximations:composition}
  Given two smoothed $p$-norm instances, $\calG_1, \calG_2,$ such that
  $\calG_1 \preceq_{\kappa_1} \calG_2$ with the map
  $\map{\calG_1}{\calG_2}$ and $\calG_2 \preceq_{\kappa_2} \calG_3$
  with the map $\map{\calG_2}{\calG_3}$, then
  $\calG_1 \preceq_{\kappa_1 \kappa_2} \calG_3$ with the map
  $ \map{\calG_1}{\calG_3} = \map{\calG_2}{\calG_3} \circ
  \map{\calG_1}{\calG_2}.$

  Similarly, for any $\calG_1, \calG_2,$ if
  $\calG_1 \circapprox_{\kappa_1} \calG_2$ with the map
  $\map{\calG_1}{\calG_2}$ and
  $\calG_2 \circapprox_{\kappa_2} \calG_3$ with the map
  $\map{\calG_2}{\calG_3}$, then
  $\calG_1 \circapprox_{\kappa_1 \kappa_2} \calG_3$ with the map
  $ \map{\calG_1}{\calG_3} = \map{\calG_2}{\calG_3} \circ
  \map{\calG_1}{\calG_2}.$
\end{restatable}

The most important property of this is that this notion of
approximation is also additive, \emph{i.e.}, it works well with graph
decompositions.
\begin{definition}[Union of two instances]
  Consider smoothed $p$-norm instances, $\calG_1, \calG_2,$ with the
  same set of vertices, \emph{i.e.}  $V^{\calG_1} = V^{\calG_2}.$
  Define $\calG = \calG_1 \cup \calG_2$ as the instance on the same
  set of vertices obtained by taking a disjoint union of the edges
  (potentially resulting in multi-edges). Formally,
  \[ \calG = (V^{\calG_1}, E^{\calG_1} \cup E^{\calG_2}, (\gg^{\calG_1},
  \gg^{\calG_2}), (\rr^{\calG_1}, \rr^{\calG_2}), (\ss^{\calG_1},
  \ss^{\calG_2})).\]
\end{definition}
\begin{restatable}[$\preceq_{\kappa}$ under union]{lemma}{union}
  \label{lem:Composition}
  Consider four smoothed $p$-norm instances,
  $\calG_1, \calG_2, \calH_1, \calH_2,$ on the same set of vertices,
  \emph{i.e.} $V^{\calG_1} = V^{\calG_2} = V^{\calH_1} = V^{\calH_2},$
  such that for $i=1,2,$ $\calH_i \preceq_{\kappa} \calG_i$ with the
  map $\map{\calH_i}{\calG_i}.$ Let
  $\calG \defeq \calG_1 \cup \calG_2,$ and
  $\calH \defeq \calH_1 \cup \calH_2.$
  Then, $\calH \preceq_{\kappa} \calG$ with the map
  \[
    \map{\calH}{\calG} \left( \ff^{\calH} =
      (\ff^{\calH_1},\ff^{\calH_2}) \right) \defeq \left(
      \map{\calH_1}{\calG_1} \left( \ff^{\calH_1} \right),
      \map{\calH_2}{\calG_2} \left( \ff^{\calH_2} \right) \right),
  \]
  where $(\ff^{H_1}, \ff^{H_2})$ is the decomposition of
  $\ff^{H}$ onto the supports of $H_1$ and $H_2$.
\end{restatable}

This notion of approximation also behaves nicely with scaling of
$\ell_2$ and $\ell_p$ resistances.
\begin{restatable}{lemma}{PerturbResistances}
  \label{lem:PerturbResistances}
  For all $\kappa \ge 1,$ and for all pairs of smoothed $p$-norm
  instances, $\calG, \calH,$ on the same underlying graphs,
  \emph{i.e.}, $(V^{\calG}, E^{\calG}) = (V^{\calH}, E^{\calH}),$ such
  that,
  \begin{tight_enumerate}
  \item the gradients are identical, $\gg^{\calG} = \gg^{\calH},$
  \item the $\ell_{2}^2$ resistances are off by at most $\kappa,$ \emph{i.e.},
    $\rr_e^{\calG} \leq \kappa \rr_{e}^{\calH}$ for all edges $e,$ and
  \item the $p$-norm scaling  is off by at most
    $\kappa^{p - 1},$ \emph{i.e.},
    $s^{\calG} \leq \kappa^{p - 1} s^{\calH},$
  \end{tight_enumerate}
  then $\calH \preceq_{\kappa} \calG$ with the identity map.
\end{restatable}


\subsection{Orthogonal Decompositions of Flows}
\label{sec:CyclePotential}

At the core of our graph decomposition and sparsification
procedures is a decomposition of the gradient $\gg$
of $\calG$ into its cycle space and potential flow space.
We denote such a splitting using
\begin{equation}
\gg^{\calG}
= 
\gghat^{\calG} + \BB^{\calG} \ppsi^{\calG},
\text{ s.t. }~{\BB^{\calG}}^{\top} \gghat^{\calG}=\vzero.
\end{equation}
Here $\gghat$ is a circulation, while
$\BB \ppsi$ gives a potential induced edge value.
We will omit the superscripts when the context is clear.

The following minimization based formulation of this
splitting of $\gg$ is critical to our method of bounding
the overall progress of our algorithm
\begin{fact}
\label{fact:Projection}
The projection of $\gg$ onto the cycle space is obtained
by minimizing the energy added to a potential flow to $\gg$.
Specifically,
\[
\norm{\gghat}_2^2
=
\min_{\xx} \norm{\gg + \BB \xx }_2^2.
\]
\end{fact}

\begin{lemma}
\label{lem:EnergyDecrease}
Given a graph/gradient instance $\calG$, consider
$\calH$ formed from a subset of its edges.
The projections of $\gg^{\calG}$ and $\gg^{\calH}$ onto
their respective cycle spaces, $\gghat^{\calG}$
and $\gghat^{\calH}$ satsify:
\[
\norm{\gghat^{\calH}}_2^2
\leq
\norm{\gghat^{\calG}}_2^2
\leq
\norm{\gg^{\calG}}_2^2.
\]
\end{lemma}


\section{Numerical Methods}
\label{sec:numerical}

The general idea of (preconditioned) numerical methods, which
are at the core of solvers for graph-structured linear
systems~\cite{SpielmanTengSolver:journal} is to repeatedly
update a current solution in ways that multiplicative reduce
the difference in objective value to optimum.
In the setting of flows, suppose we currently have some tentative
solution $\ff$ to the minimization problem
\begin{align}
\label{eq:pnormmin}
\min_{\BB \ff = \bb} \norm{\ff}_p^p
\end{align}
by performing the step
\[
\ff \leftarrow \ff+\ddelta,
\]
with the goal of improving the objective value substantially.

The work of Adil \etal~\cite{AdilKPS19} proved that $\ell_p$-norm
minimization problems could be iteratively refined.
While that result hinted at a much more general theory of
numerical iterative methods for minimizing convex objectives,
this topic is very much under development.
In this section, we will develop the tools necessary for
preconditioning $\ell_{p}$-norm based functions,
and formalize the requirements for preconditioners
necessary for recursive preconditioning algorithms.

\subsection{Iterative Refinement}
\label{subsec:IterativeRefinement}

The following key Lemma from~\cite{AdilKPS19} allows us to approximate
the change in the smoothed
$p$-norm of $\xx + \ddelta$ relative to the norm of $\xx,$ in
terms of another smoothed $p$-norm of $\ddelta.$

\begin{restatable}[\cite{AdilKPS19}]{lemma}{RestateIterativeRefinementApprox}
	\label{lem:iterative-refinement:approximation}
	For all $\rr, \xx, \ddelta \in \rea^{m},$ with
	$\rr \in \rea_{\ge 0}^{m},$ and $s \ge 0,$ we have
	\[2^{-p} \cdot h_p(\rr + |\xx|^{p-2}, s, \ddelta) \le h_p(\rr, s,
	\xx+ \ddelta) - h_p(\rr, s, \xx) - \ddelta^{\top} \nabla_{\xx}
	h_{p}(\rr, s, \xx) \le 2^{2p} \cdot h_p(\rr + |\xx|^{p-2}, s,
	\ddelta).\]
\end{restatable}
The above lemma gives us the following theorem about iteratively
refining smoothed $\ell_p$-norm minimization problems.
While the lemma was essentially proven in ~\cite{AdilKPS19}, they used
slightly different definitions, and for completeness we prove the
lemma in Appendix~\ref{sec:numerical-proofs}.

The following theorem also essentially appeared in \cite{AdilKPS19},
but again as slightly different definitions were used in that paper,
we prove the theorem in Appendix~\ref{sec:numerical-proofs} for completeness.
\begin{restatable}[\cite{AdilKPS19}]{theorem}{Residual}
	\label{thm:iterative-refinement:residual}
	Given the following optimization problem,
	\begin{align} \tag{P1}
	\label{eq:iterative-refinement:original}
	\begin{array}{rl}
	\max_{\xx} \quad & \obj_1(\xx) \defeq \gg^{\top}\xx - h_p(\rr,
	s, \xx) \\
	\text{s.t.} \quad & \AA\xx = \bb
	\end{array}
	\end{align}
	and an initial feasible solution $\xx_0,$ we can construct the
	following residual problem:
	\begin{align} \tag{R1}
	\label{eq:iterative-refinement:residual}
	\begin{array}{rl}
	\max_{\ddelta} \quad & \obj_2(\ddelta) \defeq
	(\gg')^{\top}\ddelta - h_p(\rr', s, \ddelta) \\
	\text{s.t.} \quad & \AA\ddelta = \vzero,
	\end{array}
	\end{align}
	where
	$\gg' = 2^{p}\left( \gg - \nabla_{\xx}h(\rr, s, \xx)|_{\xx = \xx_0}
	\right),$ and $\rr' = \rr + s \abs{\xx_0}^{p-2}.$
	
	There exists a feasible solution $\widetilde{\ddelta}$ to the
	residual problem \ref{eq:iterative-refinement:residual} that
	achieves an objective of
	$\obj_{2} \left( \widetilde{\ddelta} \right) \ge
	2^{p}(\obj_1(\xx^{\star}) - \obj_1(\xx_0)),$ where $\xx^{\star}$ is
	an optimal solution to
	problem~\ref{eq:iterative-refinement:original}.
	
	Moreover, given any feasible solution $\ddelta$ to Program
	\ref{eq:iterative-refinement:residual}, the vector
	$\xx_1 \defeq \xx_0 + {2^{-3p}}\ddelta$ is a feasible
	solution to the Program~\ref{eq:iterative-refinement:original} and
	obtains the objective
	\[\obj_1(\xx_1) \ge \obj_1(\xx_0) + {2^{-4p}} \obj_2(\ddelta).\]
\end{restatable}

Importantly, we can apply the above theorem to smoothed $p$-norm
flow/circulation problems.
\begin{restatable}[Iterative refinement for smoothed $p$-norm
	flow/circulation problems]{corollary}{IterativeRefinementFlow}
	\label{cor:iterative-refinement:flow}
	Given any smoothed $p$-norm flow problem $(\calG, \bb)$ with optimal
	objective $\obj^\star(\calG),$ and any initial circulation $\ff_0,$
	we can build, in $O(\abs{E^{\calG}})$ time, a smoothed $p$-norm
	circulation problem $\calH$ with the same underlying graph
	$(V^{\calH}, E^{\calH}) = (V^{\calG}, E^{\calG}),$ such that
	$\obj^{\star}(\calH) \ge 2^{p} (\obj^{\star}(\calG) -
	\obj^{\calG}(\ff_0))$ and for any circulation $\ff^{\calH}$ on
	$\calH,$ the flow $\ff_1 \defeq \ff_0 + 2^{-3p} \ff^{\calH}$
	satisfies residues $\bb$ on $\calG$ and has an objective
	\[\obj^{\calG}(\ff_1) \ge \obj^{\calG}(\ff_0) + 2^{-4p} \obj^{\calH}(\ff^{\calH}).\]
\end{restatable}

This means if we find even a crude approximate minimizer $\ddeltatil$
of this update problem, we can move to a new point $\ff' =\ff+\ddeltatil$, 
so that the gap to the optimum in the original optimization
problem~\eqref{eq:pnormmin} will decrease by a constant factor
(depending only on $p$) from $\norm{\ff}_p^p-\text{OPT}$
to  $\norm{\ff'}_p^p-\text{OPT} \leq (1-2^{-O(p)}) (
  \norm{\ff}_p^p-\text{OPT} )$.
In other words, we have a kind of iterative refinement: crude
solutions to an update problem directly give constant factor progress in the
original objective.

Note that $\norm{\ff}_p^p = \sum_i \ff_i^p$.
This will help us understand the
objective function of the update problem coordinate-wise.
Our update problem objective function is motivated by the following observations.
Our function differs slightly from the function used in \cite{AdilKPS19}, which in turn was based on functions from~\cite{BubeckCLL18}, but our function still uses a few special properties of the \cite{BubeckCLL18} functions. 
Suppose $p \geq 2$ is an even integer (only to avoid writing absolute values), then
\[
  \ff_i^p 
  + p \ff_i^{p-1} \ddelta_i
  +2^{-O(p)}\underbrace{ ( \ff_i^{p-2} \ddelta_i^2 + \ddelta_i^p)
  }_{\text{write as } h_p(\ff_i^{p-2}, \ddelta_i)}
  \leq 
   (\ff_i + \ddelta_i)^p
 \leq 
 \ff_i^p 
 + p \ff_i^{p-1} \ddelta_i
  +2^{O(p)}\underbrace{( \ff_i^{p-2} \ddelta_i^2 +
    \ddelta_i^p)}_{h_p(\ff_i^{p-2}, \ddelta_i)}
\]
Of course, the exact expansion gives
\begin{equation}
\label{eq:pexpand}
(\ff_i + \ddelta_i)^p =\ff_i^p + p \ff_i^{p-1} \ddelta_i +
\frac{p(p-1)}{2} \ff_i^{p-2} \ddelta_i^2 + \frac{p(p-1)(p-2)}{6}
\ff_i^{p-3} \ddelta_i^3 
+\ldots
+\ddelta_i^p
\end{equation}
So essentially we can approximate this expansion using
only the zeroth, first, second, and \emph{last} term in the expansion.
We use $\gg(\ff)$ to denote the vector with $\gg_i(\ff) = p \ff_i^{p-1}$
(i.e. the gradient), and let $\ff^{p-2}$ denote the entrywise powered
vector, and define $h_p(\ff^{p-2}, \ddelta) = \sum_i h_p(\ff_i^{p-2},
\ddelta_i)$.
Thus we have 
\[
\norm{\ff}_p^p + \gg(\ff)^\top \ddelta + 2^{-O(p)} h_p(\ff^{p-2}, \ddelta)
\leq
\norm{\ff+\ddelta}_p^p
\leq
\norm{\ff}_p^p + \gg(\ff)^\top \ddelta + 2^{O(p)} h_p(\ff^{p-2}, \ddelta)
\]
Note that for any scalar $0 < \lambda < 1$,
\[
\lambda^p h_p(\ff^{p-2}, \ddelta) \leq h_p(\ff^{p-2}, \lambda\ddelta) \leq \lambda^2 h_p(\ff^{p-2}, \ddelta)
\]
Together, these observations are enough to ensure that if we have
$\ddeltatil$ which is a constant factor approximate solution  to the
follow optimization problem, which we define as our \emph{update problem}
\begin{equation}
\label{eq:stepProblem}
\min_{\BB \ddelta = \vzero}  \gg(\ff)^\top \ddelta + h_p(\ff^{p-2},
\ddelta) 
\end{equation}
then we can find a $\lambda$ s.t. 
$\norm{\ff +\lambda \ddeltatil }_p^p-\text{OPT} \leq (1-2^{-O(p)}) (
  \norm{\ff}_p^p-\text{OPT} )$.

But what have we gained? Why is Problem~\eqref{eq:stepProblem} more tractable
than the one we started with?

A key reason is that unlike the exact expansion of an update as given
by Equation~\eqref{eq:pexpand}, all the higher order terms in the
objective function of \eqref{eq:stepProblem} are coordinate-wise even
functions, i.e. flipping the sign of a coordinate does not change the value of the function.
\cite{AdilKPS19} used a different but still even function instead of our
$h_p$.
This symmetrization allowed them to develop a multiplicative weights update algorithm motivated by~\cite{ChristianoKMST10} for their version of
Problem~\eqref{eq:stepProblem}, reducing the problem to solving a
sequence of linear equations.

For our choice of $h_p$, it is particularly simple to show another
very important property: 
Consider solving Problem~\eqref{eq:stepProblem} by \emph{again}
applying iterative refinement to this problem.
That is, at some intermediate step with $\delta$ being the current solution, we aim to find an update $\ddeltahat$ s.t. $\BB \ddeltahat =
\vzero$ and $\gg(\ff)^\top (\ddelta+\ddeltahat) +
h_p(\ff^{p-2},\ddelta+\ddeltahat) $ is smaller than $\gg(\ff)^\top (\ddelta) +
h_p(\ff^{p-2},\ddelta) $.
Then by expanding the two non-linear terms of $h_p(\ff_i^{p-2},\ddelta_i)$, 
i.e. $(\ddelta_i + \ddeltahat)^p$ \emph{and} $(\ddelta_i
+ \ddeltahat)^2$,  
similar to Equation~\eqref{eq:pexpand}, we get
a sequence of terms with powers of $\ddelta_i$ ranging from $2$ to $p$.
If we approximate this sequence again using only the $\ddeltahat_i^2$ and
$\ddeltahat_i^p$ terms, we get another update problem.
This update problem is an instance of 
Problem~\eqref{eq:intro:problem}. 
And in general, we can set up iterative refinement update problems for
instances of Problem~\eqref{eq:intro:problem}, and get back another
problem of the that class (after our approximation based on dropping intermediate terms).
Thus, the problem class~\eqref{eq:intro:problem} is closed under
repeatedly creating iterative update problems.
This observation is central because it allows us to develop recursive
algorithms.


\subsection{Vertex Elimination}
\label{subsec:elimination}

Following the template of the Spielman-Teng Laplacian solver, we
recursively solve a problem on $m$ edges by solving
about $\kappa$ problems on graphs with $n -1 + m / \kappa$
edges.
These ultra-sparse graphs allow us to eliminate degree 1 and 2
vertices and obtain a smaller problem.
Because our update problem~(Problem~\eqref{eq:stepProblem}) corresponds to a flow-circulation problem
with some objective, we are able to understand elimination on these
objectives in a relatively simple way: the flow on degree 1 and 2
vertices is easily related to flow in a smaller graph created by elimination.
Unlike Spielman-Teng, every recursive call must rely on a new
graph sparsifier, because the ``graph'' we sparsify  depends heavily on
the current solution that we are seeking to update: we have to
simultaneously preserve linear, quadratic and $p$-th order terms,
whose weights depend on the current solution.

A critical component of this schema is the mapping of flows
back and forth between the original graph and the new
graph so a good solution on a smaller graph can be transformed into a
good solution on the larger graph.
These mappings are direct analogs of eliminating degrees $1$
and $2$ vertices.
In Appendix~\ref{sec:Elimination}, we generalize these processes
to smoothed $\ell_{p}$-norm objectives, proving the following statements:
\begin{restatable}[Eliminating vertices with degree 1 and 2]{theorem}{elimination}
	\label{thm:elimination}
	Given a smoothed $p$-norm instance $\calG,$ the algorithm
	$\Eliminate(\calG)$ returns another smoothed $p$-norm instance
	$\calG',$ along with the map $\map{\calG'}{\calG}$ in
	$O(\abs{V^{\calG}} + \abs{E^{\calG}})$ time, such that the graph
	$G' = (V^{\calG'},E^{\calG'})$ is obtained from the graph
	$G = (V^{\calG},E^{\calG})$ by first repeatedly removing vertices
	with non-selfloop degree\footnote{By non-selfloop degree, we mean
		that self-loops do not count towards the degree of a vertex.} 1 in
	$G$, and then replacing every path $u \leadsto v$ in $G$ where all
	internal path vertices have non-selfloop degree exactly 2 in $G,$
	with a new edge $(u, v).$
	
	Moreover,
	\[\calG' \circapprox_{n^{\frac{1}{p-1}}} \calG \circapprox_1
	\calG',\] where $n = \abs{V^{\calG}},$ and the map
	$\map{\calG'}{\calG}$ can be applied in
	$O(\abs{V^{\calG}} + \abs{E^{\calG}})$ time.
\end{restatable}

\begin{restatable}[Eliminating Self-loops]{lemma}{selfLoops}
	\label{lem:elimination:selfLoops}
	There is an algorithm {\RemoveLoops} such that, given a smoothed
	$p$-norm instance $\calG$ with self-loops in $E^{\calG},$ in
	$O(\abs{V^{\calG}} + \abs{E^{\calG}})$ time, it returns instances
	$\calG_1, \calG_2,$ such that $\calG = \calG_1 \cup \calG_2,$ where
	$\calG_1$ is obtained from $\calG$ by eliminating all self-loops
	from $E^{\calG},$ and $\calG_2$ is an instance consisting of just
	the self-loops from $\calG.$ Thus, any flow $\ff^{\calG_2}$ on
	$\calG_2$ is a circulation.
	
	Moreover, there is an algorithm $\SolveLoops$ that, given $\calG_2,$
	for any $\delta \le \nfrac{1}{p},$ in time
	$O(|E^{\calG_2}| \log \nfrac{1}{\delta}),$ finds a circulation
	$\fftilde{}^{\calG_2}$ on $\calG_2$ such that
	\[
	\obj^{\calG_2}(\fftilde{}^{\calG_2}) \ge (1-\delta)
	\max_{\ff^{\calG} : (\BB)^{\calG} \ff^{\calG} = \vzero}
	\obj^{\calG_2}(\ff^{\calG_2}).
	\]
\end{restatable}

We remark that only the map from the
smaller graph to the larger has to be constructive.

\subsection{Recursive Preconditioning}
\label{subsec:recursive}
We can now present our main recursive preconditioning
algorithm, {\textsc{RecursivePreconditioning}}
 (Algorithm~\ref{alg:RecursivePreconditioning}).
 Earlier work on preconditioning for non-linear (maximum) flow
 problems all relied on \emph{oblivious routing} which gives rise to
 linear preconditioners.
 These inherently cannot work well for high-accuracy iterative
 refinement, and the issue is not merely linearity:
 Consider Problem~\eqref{eq:stepProblem}: the optimal $\ddelta$ is
 highly dependent on the current $\ff$, and when a coordinate
 $\ddelta_i$ is large compared to the current $\abs{\ff_i}$, the function depends on it as
 $\ddelta_i^p$, while when $\ddelta_i$ is small compared to
 $\abs{\ff_i}$, it behaves as $\ddelta_i^2$.
 Thus the behavior is highly dependent on the current solution.
 This necessitates adaptive and non-linear preconditioners.
 
 To develop adaptive preconditioners, we employ recursive chains of alternating calls to
 non-linear iterative refinement and a new type of
 (ultra-)sparsification that is more general and
 stronger, allowing us to simultanously
 preserve multiple different properties of our problem.
 And crucially, every time our solution is updated, our preconditioners
 change.
 The central theorem governing the combinatorial components of our algorithm,
 which is the main result proven in Section~\ref{sec:graph}, is:
 \begin{restatable}[Ultra-Sparsification]{theorem}{UltraSparsify}
 	\label{thm:ultrasparsify}
 	Given any instance $\calG =(V^{\calG}, E^{\calG}, \gg^{\calG}, \rr^{\calG},s^{\calG})$ with $n$ nodes, $m$ edges, and parameters $\kappa, \delta$ where 
 	$\log \frac{1}{\delta}$ and $\log \norm{\rr^{\calG}}_{\infty}$ are both $O(\log^c n)$ for some constant $c$ and $\kappa<m$
 	, \textsc{UltraSparsify} computes in $\Otil(m)$ running time another instance $\calH=(V^{\calH}, E^{\calH}, \gg^{\calH}, \rr^{\calH},s^{\calH}=s^{\calG})$ along with flow mapping functions $\map{\calH}{\calG},\map{\calG}{\calH}$ such that $V^{\calH}=V^{\calG}$, and with high probability we have 
 	\begin{enumerate}
 		\item $E^H$ consists of a spanning tree in the graph $(V^{\calG}, E^{\calG})$, up to $m-n+1$ self-loops and at most $\Otil(\frac{m}{\kappa})$ other non self-loop edges. 
 		\item With $\kap{\calG}{\calH}=\Otil(\kappa m^{3/(p-1)})$ for any flow $\ff^{\calG}$ of $\calG$ we have
 		\[
 		\obj_{\calH}(\frac{\map{\calG}{\calH}(\ff^{\calG})}{\kap{\calG}{\calH}})\geq \frac{1}{\kap{\calG}{\calH}}\obj_{\calG}(\ff^{\calG})-\delta\norm{\ff^{\calG}}_2\norm{\gg^{\calG}}_2,
 		\]
 		and with $\kap{\calH}{\calG}=\Otil(m^{2/(p-1)})$, for any flow solution $\ff^{\calH}$ of $\calH$ we have
 		\[
 		\obj_{\calG}(\frac{\map{\calH}{\calG}(\ff^{\calH})}{\kap{\calH}{\calG}})\geq \frac{1}{\kap{\calH}{\calG}}\obj_{\calH}(\ff^{\calH})-\delta(\norm{\ff^{\calH}}_2\norm{\gg^{\calG}}_2+\norm{\ff^{\calH}}_2^2).
 		\]
 		The flow mappings $\map{\calH}{\calG},\map{\calG}{\calH}$ preserve residue of flow, and can be applied in $\Otil(m)$ time. 
 	\end{enumerate}
 \end{restatable}


\begin{algorithm}
  \caption{Recursive Preconditioning Algorithm for $p$-smoothed
    flow/circulation problem}
  \label{alg:RecursivePreconditioning}
  \begin{algorithmic}[1]
    \Procedure{RecursivePreconditioning}{$\calG, \bb, \ff^{(0)}, \kappa, \delta$}

    \State $m \leftarrow \abs{E^{\calG}}.$ If $m \le \Otil(\kappa),$
    solve $\calG$ using the algorithm from~\cite{AdilKPS19}


    \State $T \leftarrow  \Otil(2^{3p} \kappa m^{\frac{6}{p-1}})$

    \For{$t = 0$ to $T$} 

    \State Construct the residual smoothed $p$-norm circulation
    problem $\calH_1$ for $(\calG, \bb)$ with the current solution
    $\ff^{(t)},$ given by
    Corollary~\ref{cor:iterative-refinement:flow}.

    \State $\delta' \leftarrow   \min\{1,
    \norm{\gg^{\calH_1}}^{-\frac{p}{p-1}}\} \cdot {\delta} / (4 T m).$
    
    \State 
    $\calH_2, \map{\calH_2}{\calH_1}, \kappa_{\calH_2 \to \calH_1}
    \leftarrow \textsc{UltraSparsify} (\calH_1, \kappa, \delta')$
    \Comment $\calH_2$ is an ultrasparsifier for $\calH_1$

    \State $\calH_3, \map{\calH_3}{\calH_2} \leftarrow \Eliminate(\calH_2)$ \Comment Gaussian elimination to remove degree 1,2 vertices
    
      \State $\calH_4, \calH_{\text{loops}} \leftarrow \RemoveLoops(\calH_3)$
      \Comment Remove self-loops
      
      \State
      $\DDelta^{\calH_\text{loops}} \leftarrow
      \SolveLoops(\calH_{\text{loops}}, \nfrac{1}{p})$
      \Comment Solve the self-loop instance

      \State
      $\DDelta^{\calH_4} \leftarrow \textsc{RecursivePreconditioning}(\calH_4,
      \vzero, \vzero, \kappa, \delta/T)$
    \Comment Recurse on smaller instance 

    \State
    $\DDelta^{\calH_3} \leftarrow \DDelta^{\calH_4} +
    \DDelta^{\calH_{\text{loops}}}$ \Comment Adding solution for
    $\calH^{\text{loops}}$ to obtain solution for $\calH_3$
    
    \State $\DDelta^{\calH_2} \leftarrow \abs{V^{\calH_2}}^{-\frac{1}{p-1}}
    \cdot \map{\calH_3}{\calH_2}(\DDelta^{\calH_3}).$
    \Comment Undo elimination to map solution back to $\calH_2$


    \State $\DDelta^{\calH_1} \leftarrow \kappa^{-1}_{\calH_2 \to
      \calH_1} \map{\calH_2}{\calH_1}(\DDelta^{\calH_2})$
    \Comment Map it back to the residual problem

    \State $\ff^{(t+1)} \leftarrow \ff^{(t)} + 2^{-3p}
    \DDelta^{\calH_1}$
    \Comment Update the current flow solution

    




    \EndFor

    \State \Return $\ff^{(T)}$

    \EndProcedure
 \end{algorithmic}
\end{algorithm}

Our key theorem about the performance of the algorithm is then:

\begin{theorem}[Recursive Preconditioning]
  \label{thm:RecPrecon}
  For all $p \ge 2,$ say we are given a smoothed $p$-norm instance
  $\calG,$ residues $\bb,$ initial solution $\ff^{(0)},$ and
  $\delta \le 1$ such that
  $\log \nfrac{1}{\delta}, \log \norm{\gg^{\calG}}, \log
  \norm{\rr^{\calG}}, \log \norm{\ff^{(0)}} \le \Otil(1).$ We can pick
  $\kappa = \widetilde{\Theta}(m^{\frac{1}{\sqrt{p-1}}})$ so that the
  procedure
  $\textsc{RecursivePreconditioning}(\calG, \bb, \ff^{(0)}, \kappa, \delta)$
  runs in time
  $2^{O(p^{\nfrac{3}{2}})} m^{1+O(\frac{1}{\sqrt{p-1}})},$ and returns
  a flow $\ff$ on $\calG$ such that $\ff$ satisfies residues $\bb,$
  and
  \[
    \obj^{\star}(\calG) - \obj^{\calG}(\ff^{(T)}) \le \frac{1}{2}
    (\obj^{\star}(\calG) - \obj^{\calG}(\ff^{(0)})) + 
    \delta s^{\calG}.
  \]
\end{theorem}
\begin{proof} (of Theorem~\ref{thm:RecPrecon})
  By scaling $\gg^{\calG}, \rr^{\calG},$ we can assume that $s^{\calG}
  = 1$ without loss of generality.
  
  Let us consider iteration $t$ of the for loop in
  $\textsc{RecursivePreconditioning}$. First, let us prove guarantees on the
  optimal solutions of all the relevant instances.  By the guarantees
  of Corollary~\ref{cor:iterative-refinement:flow}, we know that
  $\calH_1$ is a smoothed $p$-norm circulation problem with the same
  underlying graph
  $(V^{\calH_1}, E^{\calH_1}) = (V^{\calG}, E^{\calG}),$ such that
  $\obj^{\star}(\calH_1) \ge 2^{p} (\obj^{\star}(\calG) -
  \obj^{\calG}(\ff^{(t)})).$

  From Theorem~\ref{thm:ultrasparsify}, we know that {\textsc{UltraSparsify}}
  returns a smoothed $p$-norm circulation instance $\calH_2$ on the
  same set of vertices such that
  $ \obj^{\star}(\calH_2) \ge {\kappa^{-1}_{{\calH_1} \to \calH_2}}
  \obj^{\star}(\calH_1) - \delta' \norm{\gg^{\calH_1}}
  \norm{\ff^{\star}{}^{\calH_1}}.$
  
  From Theorem~\ref{thm:elimination}, we know that the instance
  $\calH_3$ returned by ${\Eliminate}(\calH_2)$ satisfies
  $\calH_2 \circapprox_1 \calH_3,$ and hence
  $\obj^{\star}(\calH_2) \le \obj^{\star}(\calH_3).$ From
  Lemma~\ref{lem:elimination:selfLoops}, we know that
  $\obj^{\star}(\calH_4) + \obj^{\star}(\calH_{\text{loops}}) =
  \obj^{\star}(\calH_3).$ Combining these guarantees, we obtain,
  \[ \obj^{\star}(\calH_3) = 
    \obj^{\star}(\calH_4) + \obj^{\star}(\calH_{\text{loops}}) \ge
    2^{p} {\kappa^{-1}_{{\calH_1} \to \calH_2}} (\obj^{\star}(\calG) -
    \obj^{\calG}(\ff^{(t)})) - \delta \norm{\gg^{\calH_1}}  \norm{\ff^{\star}{}^{\calH_1}}.
  \]

  Now, we analyze the approximation guarantee provided by the
  solutions to these instances.  From
  Lemma~\ref{lem:elimination:selfLoops},
  $\SolveLoops(\calH_{\text{loops}})$ returns a
  $\DDelta^{\calH_{\text{loops}}}$ that satisfies
  $\obj^{\calH_{\text{loops}}}(\DDelta^{\calH_{\text{loops}}}) \ge
  \frac{1}{2} \alpha^{\star}(\calH_{\text{loops}}).$ By induction,
  $\textsc{RecursivePreconditioning}(\calH_4, \vzero, \kappa, \delta T^{-1}),$
  upon starting with the initial solution $\vzero,$
  returns a $\DDelta^{\calH_4}$ that satisfies,
  $\obj^{\calH_4}(\DDelta^{\calH_4}) \ge \frac{1}{2}
  \obj^{\star}(\calH_4) - \delta T^{-1}$ .  Combining these
  guarantees, we have,
  \[\obj^{\calH_3}(\DDelta^{\calH_3}) =
    \obj^{\calH_4}(\DDelta^{\calH_4}) +
    \obj^{\calH_{\text{loops}}}(\DDelta^{\calH_{\text{loops}}}) \ge
    \frac{1}{2} \obj^{\star}(\calH_3) - \delta T^{-1}.
  \]
  From
  Theorem~\ref{thm:elimination}, we also have
  $\calH_3 \circapprox_{\kappa_{\text{elim}}} \calH_2,$ for
  $\kappa_{\text{elim}} = \abs{V^{\calH_2}}^\frac{1}{p-1},$ and hence
  \[ \kappa^{-1}_{\text{elim}} \obj^{\calH_3}(\DDelta^{\calH_3}) \le
    \obj^{\calH_2}( \kappa^{-1}_{\text{elim}}
    \map{\calH_3}{\calH_2}(\DDelta^{\calH_3})) = \obj^{\calH_2}(\DDelta^{\calH_2}).
  \]
  Finally, from Theorem~\ref{thm:ultrasparsify}, we have,
  \[\obj^{\calH_1}(\DDelta^{\calH_1})
    = \obj^{\calH_1}(\kappa^{-1}_{\calH_2 \to \calH_1}
    \map{\calH_2}{\calH_1}(\DDelta^{\calH_2}))
    \ge \kappa^{-1}_{\calH_2 \to \calH_1}
    \obj^{\calH_2}(\DDelta^{\calH_2}) - \delta' \norm{\gg^{\calH_1}}
    \norm{\DDelta^{\calH_2}} - \delta'  \norm{\DDelta^{\calH_2}}^{2} \]
  Combining these guarantees, we obtain,
  \begin{align*}
    \obj^{\calH_1}(\DDelta^{\calH_1})
    & \ge \kappa^{-1}_{\calH_2 \to \calH_1} \kappa^{-1}_{\text{elim}}
      \left( \frac{1}{2} 2^{p} {\kappa^{-1}_{{\calH_1} \to \calH_2}}
      (\obj^{\star}(\calG) - \obj^{\calG}(\ff^{(t)}))  
      - \frac{1}{2} \delta' \norm{\gg^{\calH_1}}
      \norm{\ff^{\star}{}^{\calH_1}} - \delta T^{-1} \right) \\
    & \qquad 
      - \delta' \norm{\gg^{\calH_1}}
      \norm{\DDelta^{\calH_2}} - \delta'  \norm{\DDelta^{\calH_2}}^{2} \\
    & \ge \widetilde{\Omega}( 2^{p} \kappa^{-1} m^{-\frac{6}{p-1}} )
      (\obj^{\star}(\calG) - \obj^{\calG}(\ff^{(t)})) \\
    & \qquad - \delta' \norm{\gg^{\calH_1}}
      \norm{\ff^{\star}{}^{\calH_1}} - \delta' \norm{\gg^{\calH_1}}
      \norm{\DDelta^{\calH_2}} - \delta'  \norm{\DDelta^{\calH_2}}^{2}
      - \delta T^{-1}. \displaybreak[3]\\
    & \ge \widetilde{\Omega}( 2^{p} \kappa^{-1} m^{-\frac{6}{p-1}} )
      (\obj^{\star}(\calG) - \obj^{\calG}(\ff^{(t)})) - 2 \delta T^{-1},
  \end{align*}
    $\norm{\ff^{\star}{}^{\calH_1}} \le \sqrt{m}
    \norm{\gg^{\calH_1}}^{\frac{1}{p-1}},$ since
    $\obj^{\calH_1}(\ff^{\star}{}^{\calH_1}) \ge 0$ implies that
    $m^{1-\nfrac{p}{2}}\norm{\ff^{\star}{}^{\calH_1}}_{2}^{p} \le
    \norm{\ff^{\star}{}^{\calH_1}}_{p}^{p} \le \gg^{\calH_1 \top}
    \ff^{\star}{}^{\calH_1} \le
    \norm{\gg^{\calH_1}}\norm{\ff^{\star}{}^{\calH_1}}.$ Similarly,
    $\norm{\DDelta^{\calH_1}} \le \sqrt{m}
    \norm{\gg^{\calH_1}}^{\frac{1}{p-1}},$ and
    $\norm{\DDelta^{\calH_2}} \le \sqrt{m} \norm{\DDelta^{\calH_1}},$
    by the reverse tree routing map.
  
  Thus, by Theorem~\ref{cor:iterative-refinement:flow}, $\ff^{(t+1)}$
  satisfies
  \begin{align*}
    \obj^{\star}(\calG) - \obj^{\calG}(\ff^{(t+1)})
    & \le  \obj^{\star}(\calG) - \obj^{\calG}(\ff^{(t)}) -  2^{-4p}
      \obj^{\calH_1}(\DDelta^{\calH_1}) \\
    & \le \left( 1- \widetilde{\Omega}( 2^{-3p} \kappa^{-1} m^{-\frac{6}{p-1}}
      ) \right)
      (\obj^{\star}(\calG) - \obj^{\calG}(\ff^{(t)})) + 2^{-4p} \cdot 2 \delta T^{-1}.
  \end{align*}
  Thus, repeating for loop $\Otil(2^{3p} \kappa m^{\frac{6}{p-1}})$
  times gives us a solution $\ff^{(T)}$ such that
  \[
    \obj^{\star}(\calG) - \obj^{\calG}(\ff^{(T)})
    \le \frac{1}{2}
    (\obj^{\star}(\calG) - \obj^{\calG}(\ff^{(0)}))
    + \delta.
  \]
  
  Now, we analyze the running time. In a single iteration, the total
  cost of all the operations other than the recursive call to
  $\textsc{RecursivePreconditioning}$ is $\Otil(m).$ Note that $\calH_1$ has
  $m$ edges.  Theorem~\ref{thm:ultrasparsify} tells us that $\calH_2$
  consists of a tree ($n-1$ edges), at most $\Otil(m/\kappa)$
  non-selfloop edges, plus many self-loops. After invoking
  $\Eliminate,$ and $\RemoveLoops,$ the instance $\calH_4$ has no
  self-loops left, and after dropping vertices with degree 0, only has
  vertices with degree at least 3. Every edge removed in $\Eliminate$
  decreases the number of edges and vertices with non-zero
  non-selfloop degree by 1. Suppose $n'$ is the number of vertices in
  $\calH_4$ with non-zero degree. Then, {\Eliminate} must have removed
  $n-n'$ edges from $\calH_2.$ Since $\calH_4$ must have at least
  $\frac{3}{2}n'$ vertices, we have
  $n-1 + \Otil(m/\kappa) - (n-n')\ge \frac{3}{2} n'.$ Hence,
  $n' \le \Otil(m/\kappa).$ Thus $\calH_4$ is an instance with at most
  $\Otil(m/\kappa)$ vertices and edges.

  Thus, the total running
  time recurrence is
  \[
    T(m) \le \Otil(2^{3p} \kappa m^{\frac{6}{p-1}}) \left( T(m/\kappa)
      + \Otil(m) \right).
  \]
  Note that $\kappa$ is fixed throughout the
  recursion. By picking $\kappa = \widetilde{\Theta}(m^{\frac{1}{\sqrt{p-1}}}),$ we
  can fix the depth of the recursion to be $O(\sqrt{p-1}).$ The total
  cost is dominated by the cost at the bottom level of the recursion,
  which adds up to a total running time of
  $2^{O(p^{\nfrac{3}{2}})} m^{1+O(\frac{1}{\sqrt{p-1}})}.$

  The above discussion does not take into account the reduction in
  $\delta$ as we go down the recursion. Observe that $\delta$ is lower
  bounded by
  \[\delta (mT \max \norm{\gg}^{\frac{p}{p-1}})^{-O(\sqrt{p-1})} =
  \delta (m2^{p} \max\{\{\norm{\ff_0}, \norm{\gg^{\calG}},
  \norm{\rr^{\calG}}\} ^{\frac{p}{p-1}}\})^{-O(\sqrt{p-1})} .\] Thus,
  we always satisfy $\log \delta = \Otil(1).$
\end{proof}

We can now prove the central collaries regardling smoothed $\ell_{p}$-norm
flows and $\ell_{p}$-norm flows.

\smoothedpflows*

\begin{proof}
  First note that the Problem~\eqref{eq:intro:problem} is a smoothed
  $p$-norm instance after flipping the sign of $\gg,$ and the sign of
  the objective function. We can solve this smoothed $p$-norm instance
  $\calG$ by using Theorem~\ref{thm:RecPrecon} to compute the desired
  approximate solution to the residual problems. We start with
  $\ff^{(0)}$ as our initial solution. At iteration $t,$ we invoke
  Theorem~\ref{thm:RecPrecon} using $\ff^{(t-1)}$ as the initial
  solution,
  \[\ff{(t)} \leftarrow \textsc{RecursivePreconditioning}\left(\calG, \bb, \ff^{(t-1)}, \kappa,
      \frac{1}{\poly(n)} \right),\]
  where $\kappa$ is given by Theorem~\ref{thm:RecPrecon}.
  
  We know that $\ff^{(t)}$
  satisfies,
  \[     \obj^{\star}(\calG) - \obj^{\calG}(\ff^{(t)}) \le \frac{1}{2}
    (\obj^{\star}(\calG) - \obj^{\calG}(\ff^{(t-1)})) + 
    \frac{1}{\poly(n)}.
  \] Iterating $O(\log n)$ times, we obtain,
  \[ \obj^{\star}(\calG) - \obj^{\calG}(\ff^{(T)}) \le \frac{1}{\poly(n)}
    (\obj^{\star}(\calG) - \obj^{\calG}(\ff^{(0)})) +
    \frac{1}{\poly(n)}.
  \]
  Finally, noting that we had flipped the sign of the objective
  function in Problem~\eqref{eq:intro:problem}, we obtain our claim.
\end{proof}

\pflows*

\begin{proof}
   The pseudocode for our procedure $\textsc{pFlows}(\calG, \bb)$ for this problem is given in Algorithm~\ref{alg:pFlows}.
 Our goal is to compute a flow $\fftilde$ satisfying
  $\BB^{G\top} \fftilde = \bb,$ such that
  \[\norm{\fftilde}^{p}_p \le \left( 1+ \frac{1}{\poly(m)} \right)
    \norm{\ff^{\star}}^{p}_p,
  \]
    where $\ff^{\star}$ is the flow minimizing the $\ell_p$-norm with
    residue $\bb$.
  For concreteness, we take this to mean $\norm{\fftilde}^{p}_p \le \left( 1+3m^{-c} \right)
  \norm{\ff^{\star}}^{p}_p $, for some constant $c$.
  We construct a smoothed $p$-norm instance
  $\calG = (V, E, \vzero, \vzero, 1).$ Note that the smoothed $p$-norm
  flow problem $(\calG, \bb)$ finds a flow satisfying residues $\bb,$
  and maximizing $\obj^{\calG}(\ff) = - \norm{\ff}^{p}_{p}.$ We can
  solve this smoothed $p$-norm instance by iteratively refining using
  Corollary~\ref{cor:iterative-refinement:flow}, and using
  Theorem~\ref{thm:RecPrecon} to compute the desired approximate
  solution to the residual problems.

  Formally, we use Laplacian solvers to compute in $\Otil(m)$ time
  $\ff^{(0)}$ as a 2-approximation to
  $ \min_{\ff : \BB^{\top} \ff = \bb} \norm{\ff}.$ We
  have,
  \[ \norm{\ff^{(0)}}_p \le \norm{\ff^{(0)}}_2 \le 2 \min_{\BB^{G\top}
      \ff = \bb} \norm{\ff}_2 \le 2 \norm{\ff^{\star}}_2 \le 2
    m^{\frac{1}{2} - \frac{1}{p}} \norm{\ff^{\star}}_{p}.\]

  At each iteration
  $t,$ we construct the residual smoothed $p$-norm circulation problem
  $\calH_t$ for $(\calG, \bb)$ with the current solution $\ff^{(t)},$
  given by Corollary~\ref{cor:iterative-refinement:flow}. We then
  invoke
  \[\DDelta^{(t)} \leftarrow \textsc{RecursivePreconditioning}(\calH_t, \vzero, \vzero, \kappa,
  \frac{1}{2^{p}m^{\nfrac{p}{2}}m^c} \cdot \norm{\ff^{(0)}}^{p}),\] where
  $\kappa = \widetilde{\Theta}(m^{\frac{1}{\sqrt{p-1}}})$ is given by
  Theorem~\ref{thm:RecPrecon}. Let $\DDelta^{(t)}$ be the flow
  returned.
  We know
  \[\obj^{\star}(\calH_t) - \obj^{\calH_t}(\DDelta^{(t)}) \le \frac{1}{2}
    \obj^{\star}(\calH_t) + \frac{1}{2^{p}m^{\nfrac{p}{2}}m^c} \norm{\ff^{(0)}}^{p}_{p}\]
  We let $\ff^{(t+1)} \leftarrow \ff^{(t)} + 2^{-3p} \DDelta^{(t)}.$ Thus,
  by Corollary~\ref{cor:iterative-refinement:flow}, at every
  iteration, we have
  \begin{align*}
    -\norm{\ff^{(t+1)}}^{p}_{p} = \obj^{\calG}(\ff^{(t+1)})
    & \ge
      \obj^{\calG}(\ff^{(t)}) + 2^{-4p} \obj^{\calH_t}(\DDelta^{(t)}) \\
    & \ge - \norm{\ff^{(t)}}_p^p + 2^{-4p} \left( \frac{1}{2}
      \obj^{\star}(\calH_t)- \frac{1}{2^{p}m^{\nfrac{p}{2}}m^c}
      \norm{\ff^{(0)}}_p^{p} \right) \\
    & \ge - \norm{\ff^{(t)}}_p^p + 2^{-4p} \left( \frac{1}{2}
      2^{p}(\obj^{\star}(\calG)-\obj^{\calG}(\ff^{(t)}))- 
      \frac{1}{m^c}\norm{\ff^{\star}}_p^{p} \right) \\
    & \ge - \norm{\ff^{(t)}}_p^p + 2^{-4p} \left (\norm{\ff^{(t)}}_p^p
      - \norm{\ff^{\star}}_p^p \right) -
      \frac{2^{-4p}}{m^c}\norm{\ff^{\star}}_p^{p}
      .
  \end{align*}
  where we have used,
  $
  \norm{\ff^{(0)}}_p \le 2 m^{\frac{1}{2} -
    \frac{1}{p}} \norm{\ff^{\star}}_{p}$,
  \[ \norm{\ff^{(t+1)}}_p^p - \norm{\ff^{\star}}_p^p \le
    (1-2^{-4p})\left(  \norm{\ff^{(t)}}_p^p - \norm{\ff^{\star}}_p^p
    \right)
    + \frac{2^{-4p}}{m^c}\norm{\ff^{\star}}_p^{p}.
  \]
  Thus
  \begin{align*}
    \norm{\ff^{(t)}}_p^p - \norm{\ff^{\star}}_p^p
    &\le
    (1-2^{-4p})\left(  \norm{\ff^{(t)}}_p^p - \norm{\ff^{\star}}_p^p
    \right)
    + \frac{2^{-4p}}{m^c}\norm{\ff^{\star}}_p^{p}
      \\
    &\le
    \max\left((1-2^{-(4p+1)})\left(
        \norm{\ff^{(t)}}_p^p -
        \norm{\ff^{\star}}_p^p
        \right), 
      3m^{-c} \norm{\ff^{\star}}_p^p\right)
    .
  \end{align*}
Where first step follows by rearranging terms.
To establish the second inequality, first consider
the case when $\left(
        \norm{\ff^{(t)}}_p^p -
        \norm{\ff^{\star}}_p^p
      \right) \ge 2m^{-c}\norm{\ff^{\star}}_p^{p}$ and
      hence
      $\norm{\ff^{(t)}}_p^p - \norm{\ff^{\star}}_p^p \le 
       (1-2^{-4p}/2)\left(  \norm{\ff^{(t)}}_p^p -
         \norm{\ff^{\star}}_p^p \right)$,
meanwhile, when $\left(
        \norm{\ff^{(t)}}_p^p -
        \norm{\ff^{\star}}_p^p
      \right) < 2m^{-c}\norm{\ff^{\star}}_p^{p}$, the inequality is immediate.
   
  Iterating $T=\Theta((c+p)2^{4p} \log m)$ times gives us
  \[ \norm{\ff^{(T)}}_p^p - \norm{\ff^{\star}}_p^p 
\le
    \max\left((1-2^{-(4p+1)})^T\left(
        \norm{\ff^{(0)}}_p^p -
        \norm{\ff^{\star}}_p^p
        \right), 
      3m^{-c} \norm{\ff^{\star}}_p^p\right)
  \le
3m^{-c} \norm{\ff^{\star}}_p^p
.
\]

\end{proof}

\begin{algorithm}
  \caption{Computing $p$-norm minimizing flows.
    Given constant $p$ and $c$,
    the routine computes  $\fftilde$ with residues $\bb$ and $p$-norm
    that is within a factor $(1+3m^{-c})$ of the minimum $p$-norm
    achievable for these residues.
  }
  \label{alg:pFlows}
  \begin{algorithmic}[1]
    \Procedure{pFlows}{$\calG, \bb$}

    \State Use Laplacian solvers to compute $\ff^{(0)}$ as a
    2-approximation to
    $ \min_{\ff : \BB^{\top} \ff = \bb} \norm{\ff}$.


  
    \State
    $\delta \leftarrow \min \left\{ 1,
      \frac{1}{2^{p}m^{\nfrac{p}{2}}m^c} \cdot \norm{\ff^{(0)}}_p^{p}
    \right\}$
    \State $\kappa \leftarrow \widetilde{\Theta}(m^{\frac{1}{\sqrt{p-1}}})$ 
    \State $T\leftarrow\Theta((c+p)2^{4p} \log m)$
    \For{$t = 0$ to $T-1$} 

    \State Construct the residual smoothed $p$-norm circulation
    problem $\calH_t$ for $(\calG, \bb)$ with the current solution
    $\ff^{(t)},$ given by
    Corollary~\ref{cor:iterative-refinement:flow}.

    \State $\DDelta^{(t)} \leftarrow \textsc{RecursivePreconditioning}(\calH_t,
    \vzero, \vzero, \kappa, \delta)$
    \State $\ff^{(t+1)} \leftarrow \ff^{(t)} + 2^{-3p} \DDelta^{(t)}$
    
    \EndFor

    \State \Return $\fftilde \leftarrow \ff^{(T)}$

    \EndProcedure
 \end{algorithmic}
\end{algorithm}



\section{Graph Theoretic Preconditioners}
\label{sec:graph}
In this section, we discuss at a high level of the construction of
ultra-sparsifiers for a smooth $\ell_{p}$-norm instance.
We start by restating the main theorem of our ultra-sparsifier. After
establishing the necessary tools, we prove this theorem at the end of
this section.

\UltraSparsify*

Our high-level approach is the same as Spielman-Teng~\cite{SpielmanTengSolver:journal}, where we utilize a low-stretch spanning tree, and move off-tree edges to a small set of portal nodes. Once most off-tree edges are only between a small set of portal nodes, we sparsify the graph over the portal nodes to reduce the number of edges. As we need to map flow solutions between the original instance and the sparsified instance, our main concern is to carry out these step without incurring too much error on the objective function values. In our case we have $\ell_p^{p}$ resistances and gradients on edges in addition to $\ell_2^{2}$ resistances, and thus the main challenge is to simultaneously preserve their respective terms. 


\subsection{Tree-Portal Routing}
\label{subsec:treerouting}

Faced with a sparse graph or dense graph, we wish to move most edges
onto a few vertices of a tree, so that many of the remaining vertices
are low degree and can be eliminated.
Our high-level approach is the same as
Spielman-Teng~\cite{SpielmanTengSolver:journal}, where we utilize a
low-stretch spanning tree, and move off-tree edges to a small set of
portal nodes. 
However, when we reroute flow on the graph where most edges are moved to be among a few vertices using the tree, we need to (approximately) preserve three different
properties of the flow on the original graph, namely the inner product between
gradients and flows $\sum_e \gg_e \ff_e$, the 2-norm energy 
$\rr_e \ff_e^2$, and the $\ell_p$-norm energy $\sum_e \ff_e^p$.
It turns out that we can move edges around on our graphs to produce a
new graph while exactly preserving the linear term $\sum_e \gg_e
\ff_e$ for flows mapped between one graph and the other. This means
any tree is acceptable from the point of preserving the linear
term.
To move edges around and bound distortion of solutions w.r.t. the
quadratic $\rr_e \ff_e^2$ term, we use a low stretch tree w.r.t. the
$\rr$ weights as resistances.
This leaves us with little flexibility for the  $\sum_e \ff_e^p$
term.
However, for large $p$, provided \emph{every} $p$-th order term
is weighted the same (i.e. we have $s\sum_e \ff_e^p$ instead of $\sum_e \ss_e \ff_e^p$), it
turns out that, moving edges along any tree will result in bounded
distortion of the solution, provided we are careful about how we move
those edges.
Thus, we can move edges around carefully to be among a small subset of the portal nodes while simultaneously controlling all linear, $2$-nd order and $p$-th order terms.
But, this only works if all the $p$-th order terms are weighted
the same.
This leads us to maintain uniform-weighted $p$-th order terms 
as an invariant throughout the algorithm. 
The iterative refinement steps naturally weigh all $p$-th order terms
the same provided the original function does. 
However, our sparsification procedures do not immediately achieve
this, but we show we can enforce this uniform-weight invariant with
only a manageable additional distortion of our solutions.
Elimination also does not naturally weigh all $p$-th order terms
the same even in our case when the original function does, but we can bound the
distortion incurred by explicitly making the weights uniform.
Our tree-based edge re-routing naturally creates maps between
solutions on the old and new graphs.

We first formalize what we mean by moving off-tree edges. Suppose we have a spanning tree $T$ of a graph $(V,E)$ and a subset set of nodes $\Vhat\subset V$ designated as {\em portal nodes}, for any off-tree edge $e=\{u,v\}\in E\setminus T$, there is a unique tree path $\Pcal_{T}(u,v)$ in $T$ from $u$ to $v$. We define a {\em tree-portal path} $\Pcal_{T,\Vhat}(u,v)$, which is not necessarily a simple path.
\begin{definition}[Tree-portal path and edge moving]
\label{defn:tree-portal}
Given spanning tree $T$ and set of portal nodes $\Vhat$, let $e=\{u,v\}$ be any edge not in $T$, and $P_{T}(u,v)$ the unique tree path in $T$ from $u$ to $v$. We define $e$'s {\em tree-portal path} $P_{T,\Vhat}(u,v)$ and $e$'s image under {\em tree-portal edge moving} as follows
\begin{enumerate}
\item If $P_{T}(u,v)$ doesn't go through any portal vertex. In this case, we replace $\{u,v\}$ with a distinct self-loop of $v$. We let $P_{T,\Vhat}(u,v)$ be the path $P_T(u,v)$ followed by the self-loop at $v$.
\item If $P_{T}(u,v)$ goes through exactly one portal vertex $\vhat$. In this case, we replace $\{u,v\}$ with a distinct self-loop at $\vhat$. We let $P_{T,\Vhat}(u,v)$ be the tree path $P_T(u,\vhat)$ followed by the self-loop at $\vhat$ and then the tree path $P_T(\vhat,v)$.
\item If $P_{uv}$ goes through at least two portal vertices. In this case, let $\uhat$ (and $\vhat$) be closest the portal vertex to $u$ (and $v$) on $P_{uv}$, we replace $\{u,v\}$ with a distinct edge\footnote{We will keep multi-edges explicitly between portal nodes.} $\{\uhat,\vhat\}$. We let $P_{T,\Vhat}(u,v)$ be the tree path $P_T(u,\uhat)$ followed by the new edge from $\uhat$ to $\vhat$ and then the tree path $P_T(\vhat,v)$.
\end{enumerate}
This maps any off-tree edge $e$ to a unique (edge or self-loop) $\ehat$ given any $T,\Vhat$. 
 We denote the tree-portal edge moving with the map $\ehat=\mbox{\textsc{Move}}_{T,\Vhat}(e)$.
\end{definition}
 Although we will get self-loops in tree-portal routing, to keep the discussion simple, we ignore the possibility of getting self-loops. This still captures all the main ideas, and the algorithm/analysis extends to self-loops in a very straightforward but slightly tedious way. We discuss self-loops briefly at the end of the section.
 
Tree-portal routing is a mapping from flow solutions on the original off-tree edges to a flow solution (with the same residue) using the edges they are mapped to. Any flow along off-tree edge $(u,v)$ in the original graph is rerouted (again from $u$ to $v$) using the tree-portal path $P_{T,\Vhat}(u,v)$ instead. Rerouting the flow of any off-tree edge along its tree-portal path increases the congestion on tree edges, which in turn incurs error in the $\ell^2_2$ and $\ell^p_p$ terms in the objective function. We need to pick the tree and portal nodes carefully to bound the error. 
\begin{definition} Given any graph $(V,E)$, resistance $\rr$ on edges, a spanning tree $T$, and a set of portals $\Vhat\subset V$, for any $e=\{u,v\} \in E$, let $\ehat=\mbox{\textsc{Move}}_{T,\Vhat}(e)$ and $P_{T,\Vhat}(u,v)$ be as specified above. The {\em stretch} of $e=\{u,v\}\in E\setminus T$ in the tree-portal routing is
\[
\str_{T,\Vhat} \left( e \right) \defeq \frac{1}{\rr_e}
    \sum_{e' \in P_{T,\Vhat}\left(e\right)\setminus\{\ehat\}} \rr_{e'},
\]
and the stretch of a tree edge $e\in T$ is $\str_{T,\Vhat} \left( e \right)=1$. Note with our definition $\str_{T,\emptyset}(e)$ gives the standard stretch. 
\end{definition}

The starting point is low stretch spanning trees~\cite{AbrahamN12}, which provide good bounds on the total $\ell_2^2$ stretch. 
\begin{lemma}[Low-Stretch Trees~\cite{AbrahamN12}]
\label{lem:LSST}
Given any graph $G=(V,E)$ of $m$ edges and $n$ nodes, as well as resistance $\rr$,
$\textsc{LSST}(\rr)$ finds a spanning tree in $O\left( m \log{n} \log\log{n} \right)$ time such that
\[
\sum_{e\in E} \str_{T,\emptyset} \left( e \right)
\leq
O\left( m \log{n} \log\log{n} \right).
\]
\end{lemma}
We will construct a low-stretch spanning tree $T$ of $(V^{\calG},E^{\calG},\rr^{\calG})$ using the above result. Still, the error will be too large if we only use tree edges to reroute the flow of all the off-tree edges, since the low average stretch doesn't prevent one tree edge to be on the tree path for many off-tree edges. Thus, we need to add portal nodes so we can shortcut between them to reduce the extra congestion on tree edges. 



\subsection{Partitioning Trees into Subtrees and Portals}
\label{subsec:treepartition}
Next, we show how to find a small set of good portal nodes so that rerouting flow on
off-tree edges using their tree-portal paths incurs small error in the objective function.
Pseudocode of this routine is in Algorithm~\ref{alg:findportal},
and its guarantees are stated in Lemma~\ref{lem:FindPortal} below. 

\begin{algorithm}
	\caption{Find portal nodes for tree-portal routing}
	\label{alg:findportal}
	\begin{algorithmic}[1]
		\Procedure{FindPortal}{$T$,$E$,$\nhat$}
		\State $\forall e\in E: \eta(e)\leftarrow \max \left(\str_{T,\emptyset}(e), \frac{\sum_{e'\in E}\str_{T,\emptyset}(e')}{|E|}\right)$
		\State Call {\em decompose} in \cite{SpielmanTengSolver:journal} (page $881$ of journal version) with $(T,E,\eta,\nhat)$. 
		\State The subroutine breaks $T$ into at most $\nhat$ edge-disjoint induced tree pieces to divide up the $\eta(e)$'s roughly evenly so that the sum of $\eta(e)$ for all $e$ attached to each non-singleton piece is not too big.
		\State The subroutine works by recursively cut off sub-trees from $T$ whenever the sum of $\eta(e)$ of all $e$ attached to a sub-tree is above $\frac{2\sum_e\eta(e)}{\nhat}$.
		\State Let $\Vhat$ be the set of nodes where the tree pieces intersect. 
		\EndProcedure 
	\end{algorithmic}
\end{algorithm}

\begin{lemma}
\label{lem:FindPortal}
	There is a linear-time routine \textsc{FindPortals} that given any graph $G$,
	a spanning tree $T$, with $\mhat$ off-tree edges, and a portal count
	$\nhat \leq \mhat$, returns a subset of $\Vhat$ of $\nhat$ vertices
	so that for all edges $\ehat \in T$, we have
	\begin{align*}
	\sum_{e: \ehat \in P_{T, \Vhat}\left( e \right)}
	\str_{T, \Vhat}\left( e \right)
	& \leq \frac{10}{\nhat}
	\sum_{e} \str_{T, \emptyset}\left( e \right) \\\
	\abs{e: \ehat \in P_{T, \Vhat}\left( e \right)}
	& \leq \frac{10 \mhat}{\nhat}.
	\end{align*}
\end{lemma}

This lemma will be a fairly straightforward using the tree decomposition subroutine (page $881$ of journal version) from Spielman and Teng~\cite{SpielmanTengSolver:journal}, which we include below for completeness.
\begin{definition}[\cite{SpielmanTengSolver:journal} Definition $10.2$]
Given a tree $T$ that spans a set of vertices $V$, a $T$-decomposition is a decomposition of $V$ into sets $W_1,\ldots,W_h$ such that $V=\bigcup W_i$, the graph induced by $T$ on each $W_i$ is a tree, possibly with just one vertex, and for all $i\neq j$, $|W_i\bigcap W_j|\leq 1$.

Given an additional set of edges $E$ on $V$, a $(T,E)$-decomposition is a pair $(\{W_1,\ldots,W_h\},\rho)$ where $\{W_1,\ldots,W_h\}$ is a $T$-decomposition and $\rho$ is a map that sends each edge of $E$ to a set or pair of sets in $\{W_1,\ldots,W_h\}$ so that for each edge in $(u,v)\in E$,
\begin{enumerate}
\item if $\rho(u,v)=\{W_i\}$, then $\{u,v\}\subset W_i$, and
\item if $\rho(u,v)=\{W_i,W_j\}$, then either $u\in W_i$ and $v\in W_j$ or vice versa.
\end{enumerate}
\end{definition}
\begin{theorem}[\cite{SpielmanTengSolver:journal} Theorem $10.3$]
\label{thm:STtreedecomp}
There exists a linear-time algorithm such that on input a set of edges $E$ on vertex set $V$, a spanning tree $T$ on $V$, a function $\eta:E\rightarrow \rea^{+}$, and an integer $1<t\leq \sum_{e\in E}\eta(e)$, outputs a $(T,E)$-decomposition $(\{W_1,\ldots,W_h\},\rho)$, such that
\begin{enumerate}
\item $h\leq t$
\item for all $W_i$ such that $|W_i|>1$,
\[
\sum_{e\in E:W_i\in \rho(e)}\eta(e)\leq \frac{4}{t}\sum_{e\in E}\eta(e)
\]
\end{enumerate}
\end{theorem}
We can use the above theorem to show Lemma~\ref{lem:FindPortal}.
\begin{proof}[Proof of Lemma~\ref{lem:FindPortal}]
We will apply Theorem~\ref{thm:STtreedecomp} with $t=\nhat$, and the function $\eta$ will be
\[
\eta(e)=\max \left(\str_{T,\emptyset}(e), \frac{\sum_{e'\in E}\str_{T,\emptyset}(e')}{\mhat}\right)
\]
Note by construction $\sum_e \eta(e) \leq 2\sum_{e'\in E}\str_{T,\emptyset}(e')$.
We get $\{W_1,\ldots,W_{\nhat}\}$ back, and let $T_i$ be the tree induced by $T$ on $W_i$. Note the $T_i$'s will be edge disjoint, and cover all tree edges of $T$. Our set of portals will be the set of nodes that are in more than one of the $W_i$'s, i.e. the nodes where different $T_i$'s overlap. Note the number of portals is at most the number of $T_i$'s by an inductive argument from any $T_i$ that is a sub-tree in $T$. Such $T_i$ have exactly one portal, and we can remove $T_i$ from $T$ and continue the argument until all that remain in $T$ is one sub-tree. 

Consider any tree edge $\ehat\in T$, suppose it in $T_i$ for some $i$. $\ehat$ can only be on the tree-portal routing for some edge $\{u,v\}$ when $W_i\in \rho(u,v)$. Note as $T_i$ contains at least one tree edge, we know $|W_i|>1$, the second guarantee in Theorem~\ref{thm:STtreedecomp} gives 
\[
\sum_{e:W_i\in\rho(e)} \max \left(\str_{T,\emptyset}(e), \frac{\sum_{e'\in E}\str_{T,\emptyset}(e')}{\mhat}\right) \leq \frac{4}{t} \left(2\sum_{e'\in E}\str_{T,\emptyset}(e')\right)
\]
which directly gives the bounds we want in the lemma.
\end{proof}

%



\subsection{Graph Sparsification}
\label{subsec:sparsification}

Once we are able to move most of the edges onto a small subset of vertices, we wish
to sparsify the resulting dense graph over those vertices.
This sparsification has to simultaneously preserve properties of 
$1$-st, $2$-nd and
$p$-th order terms, as well as the interactions between them, which turns out to be challenging. 
We resort to expander decomposition which allows us to partition the
vertex set s.t. the edges internal to each subset form an expander and
not too many edges cross the partitions. Just having an expander graph is not enough to allow us to sample the
graph due to the need of preserving the linear terms. Thus, we also require that on each expander the
orthogonal projection of the gradient to the cycle space of the sub-graph has its maximum squared entry not much larger than
the the average squared entry.
We refer to this as a uniform (projected) gradient.
We discuss how to obtain an expander decomposition that guarantees
the projected gradients are uniform in the expanders later in this overview. Given the uniform projected gradient condition, we show that we can uniformly sample edges
of these expanders to create sparsified versions of them.
We construct maps between the flows on an original expander and its
sampled version that work for \emph{any}
flow, not only a circulation.
These maps preserve the linear term $\sum_e \gg_e \ff_e$ exactly,
while bounding the cost of the $2$-norm and $\ell_p$-norm terms by relating them to the cost of the
optimal routing of a flow with the same demands and same gradient
inner product, and showing that optimal solutions are similar on the
original expander and its sampled version.
This strategy resembles the flow maps developed in \cite{KelnerLOS14},
and like their maps, we route demands using electrical flows on
individual expanders, but additionally we need create a flow in the
cycle space that depends on projection of the gradient onto
the cycle space.

Tree-portal routing will give us an instance where all the off-tree edges
are between portal nodes.
We can look at the sub-graph restricted to the portal nodes and
the off-tree edges between them.
This graph has many fewer nodes comparing to the original graph but
roughly the same number of edges, and thus is much denser. We can then
sparsify this graph to reduce the number of off-tree edges similar to
the construction of spectral sparsifiers. The main technical
difficulty is that in the sparsified graph, we still want the
$\ell_p^p$ terms in our objective function to have a same scalar $s$
for every edge, but similar to the case of how resistances are scaled
in spectral sparsification, to preserve the total value of the
$\ell_p^p$ terms, we would naturally want to scale $s$ according to
the probability we sample an edge $e$. Thus, to get a same scalar $s$
for all sampled edges, we are limited to uniform sampling. We know
uniform sampling works in expanders
(c.f. \cite{SpielmanTengSolver:journal, SpielmanS08:journal} and
\cite{KelnerLOS14,Sherman13} for $\ell_2$ and $\ell_{\infty}$ respectively),
so the natural approach is to first decompose the graph into expanders, and sampling uniformly inside each expander. However, because of the presence of a gradient, we need to be a bit more careful than even expanderdecomposition-based sparsification steps.
Thus, we work with {\em uniform expanders}.
\begin{definition}
\label{defn:uniformexpander}
  A graph~\footnote{We use an instance and its underlying graph interchangeably in our discussion.} $G$ is a $\alpha$-uniform $\phi$-expander (or {\em uniform expander} when parameters not spelled out explicitly) if
  \begin{tight_enumerate}
  \item $\rr$ on all edges are the same.
  \item $\ss$ on all edges are the same.
  \item $G$ has {\em conductance}\footnote{$\rr$ are uniform, so conductance is defined as in unweighted graphs. We use the standard definition of conductance. For graph $G=(V,E)$, the conductance of any $\emptyset\neq S\subsetneq V$ is $\phi(S)=\frac{\delta(S)}{\min\left(vol(S),vol(V\setminus S)\right)}$ where $\delta(S)$ is the number of edges on the cut $(S,V\setminus S)$ and $vol(S)$ is the sum of the degree of nodes in $S$. The conductance of a graph is $\phi_G=\min_{S\neq \emptyset,V}\phi(S)$.} at least $\phi$.
  \item The projection of $\gg$ onto the cycle space of $G$,
    $\gghat^G = (I-\BB \LL^{\dag} \BB^{\top})\gg$, is {\em $\alpha$-uniform} (see next definition), where $\BB$ is the edge-vertex incidence matrix of $G$, and $\LL=\BB^{\top}\BB$ is the Laplacian.
  \end{tight_enumerate}
\end{definition}
\begin{definition}
\label{defn:alphauniform}
  A vector $\yy \in \rea^m$ is said to be {\em $\alpha$-uniform} if
  \[\norm{\yy}_{\infty}^{2} \le \frac{\alpha}{m}
    \norm{\yy}^{2}_{2}.\]
 We abuse the notation to also let the all zero vector $\vzero$ be $1$-uniform.
\end{definition}
In Section~\ref{sec:decomposition} we show how to decompose the graph consisting of portals and the off-tree edges between them into vertex disjoint uniform expanders such that more than half of the edges are inside the expanders.\footnote{Some of the expanders we find actually won't satisfy the projected gradient being $\alpha$-uniform constraint (case $3(b)$ in Theorem~\ref{thm:Decompose}).
For those expanders, the projection of the gradient in the cycle space is tiny so we make it $0$. This leads to the additive error in Theorem~\ref{thm:ultrasparsify}.}

 	\begin{restatable}[Decomposition into Uniform Expanders]{theorem}{Decompose}
 		\label{thm:Decompose}
 		Given any graph/gradient/resistance instance $\calG$
 		with $n$ vertices, $m$ edges, unit resistances, and gradient $\gg^{\calG}$,
 		along with a parameter $\delta$, $\textsc{Decompose}(\calG, \delta)$
 		returns vertex disjoint subgraphs $\calG_1, \calG_2, \ldots$
 		in $O(m \log^{7}n \log^2(n / \delta))$ time such that
 		at least $m/2$ edges are contained in these subgraphs,
 		and each $\calG_i$ satisfies (for some absolute constant
 		$c_{partition}$):
 		\begin{enumerate}
 			\item The graph $(V^{\calG_i}, E^{\calG_i})$ has conductance
 			at least
 			\[
 			\phi = 
 			\frac{1}{c_{partition}
 				\cdot \log^{3}n
 				\cdot \log\left(n / \delta\right)},
 			\]
 			and degrees at least $\phi \cdot \frac{m}{3n}$,
 			where $c_{partition}$ is an absolute constant.
 			\item The projection of its gradient $\gg^{\calG_i}$ into the cycle space
 			of $\calG_i$, $\gghat^{\calG_i}$ satisfies one of:
 			\begin{enumerate}
 				\item $\gghat^{\calG_i}$ is $O(\log^{8}{n} \log^{3}(n / \delta))$-uniform,
 				\[
 				\left(\gghat_{e}^{\calG_i}\right)^2
 				\leq
 				\frac{O\left( \log^{14}n \log^{5}\left( n / \delta\right)\right)}
 				{m_{i}} \norm{\gghat^{\calG_i}}_2^2
 				\qquad
 				\forall e \in E\left( \calG_i \right).
 				\]
 				Here $m_i$ is the number of edges in $\calG^{\calG_i}$.
 				\label{case:Uniform}
 				\item The $\ell_2^2$ norm of $\gghat^{\calG_i}$ is smaller by a factor
 				of $\delta$ than the unprojected gradient:
 				\[
 				\norm{\gghat^{\calG_i}}^2_2
 				\leq
 				\delta \cdot \norm{\gg^{\calG}}_2^2.
 				\]
 				\label{case:Tiny}
 			\end{enumerate}
 		\end{enumerate}
 	\end{restatable}

Moreover, the min degree of any node in the expanders is up to a polylog factor close to the average degree. For the off-tree edges not included in these uniform expanders, we work on the pre-image\footnote{By pre-image of $\ehat$ we mean the original off-tree edge $e$ that gets moved to $\ehat$ in the tree-portal routing, i.e. $e=\textsc{Move}^{-1}_{T,\Vhat}(\ehat)$.} of them in the next iteration. That is, for any edge $\ehat$ inside one of the expanders, we remove its pre-image from the instance $\calG$, and work on the remaining off-tree edges in $\calG$ in the next iteration. This iterative process terminates when the number of remaining off-tree edges is small enough (i.e. $\Otil(|E^{\calG}|/\kappa)$). This takes $O(\log |E^{\calG}|)$ iterations as a constant fraction of off-tree edges are moved to be inside the expanders each iteration.
 
\paragraph{Sparsify Uniform Expanders} If we append a column containing the gradient of edges to the edge-vertex incidence matrix $\BB$, the conditions of a $\alpha$-uniform $\phi$-expander is equivalent to each row of $\BB$ having leverage
score at most $\frac{n \alpha \phi^{-1}}{m}$ where $n,m$ are number of nodes and edges. An underlying connection with the $\ell_p$-norm row sampling
result by Cohen and Peng~\cite{CohenP15} is that this
is also a setting under which $\ell_q$-norm functionals
are preserved under uniform sampling. We refrain from developing a more complete picture of
such machinery here, and will utilize ideas closer to
routing on expanders~\cite{KelnerM09,KelnerLOS14} to show
a cruder approximation in Section~\ref{sec:expander}.

\begin{restatable}[Sampling Uniform Expanders]{theorem}{ExpanderSparsify}
	
	\label{thm:sampAndFixGrad}
	Given an $\alpha$-uniform $\phi$-expander
	$\calG = (V^{\calG} E^{\calG} , r^{\calG},s^{\calG},\gg^{\calG})$
	with $m$ edges and vertex degrees at least $d_{\min}$,
	for any sampling probability $\tau$ satisfying
	\[
	\tau \geq c_{sample} \cdot \log{n} \cdot
	\left( \frac{\alpha}{m}
	+ \frac{1}{\phi^2 d_{\min}} \right),
	\]
	where $c_{sample}$ is some absolute constant,
	$\textsc{SampleAndFixGradient}(\calG, \tau)$ 
	w.h.p. returns 
	a partial instance $\calH = (H,r^{\calH},s^{\calH},\gg^{\calH})$ and maps
	$\map{\calG}{\calH}$ and $\map{\calH}{\calG}$.
	The graph $H$ has the same vertex
	set as $G$, and $H$ has at most $2 \tau m$ edges.
	Furthermore, $r^{\calH} = \tau \cdot r^{\calG}$ and $s^{\calH} = \tau^p \cdot s^{\calG}$.
	The maps
	$\map{\calG}{\calH}$ and $\map{\calH}{\calG}$
	certify 
	\[
	\calH \preceq_{\kappa} \calG \text{ and } \calG \preceq_{\kappa}
	\calH
	,
	\]
	where $\kappa = m^{1/(p-1)} \phi^{-9} \log^{3} n  $.
\end{restatable}

\subsection{Ultra-sparsification Algorithm and Error Analysis}
\label{subsec:graph_ultrasparsify}

Now we put all the pieces together.
We need to show that adding together our individual sparsifiers results in a sparsifier of the overall graph.
This is fairly immediate given the strong guarantees we
established on the individual graphs.
We also need to be able to repeatedly
decompose and sparsify enough times that the overall graph becomes
sparse.
To address this issue, we use ideas from \cite{KoutisMP11} that suggest scaling up the tree from the tree routing section limits the
error incurred during sampling.
Here it again becomes important that because we rely on
\cite{SaranurakW18:prelim}, we know exactly which edges
belong to a sparsifier.
This guarantee limits the interaction 
between sparsification of different
expanders.

After constructing a low-stretch spanning tree $T$, we round each $\rr^{\calG}_e$ of off-tree edges $e\in E^{\calG}\setminus T$ down to the nearest power of $2$ (can be less than $1$) if $\rr^{\calG}_e\geq \delta$, and round to $0$ otherwise. This gives of $\nbucket$ bucket of edges with uniform resistances, and we work with one bucket of edges at a time, since the edges in a uniform expander need to have uniform $\rr_e$. If $\calG'$ is the instance of $\calG$ after rounding the resistance of off-tree edges, it is easy to see the following error guarantee.
\begin{lemma}
\label{lemma:rounding}
$\calG \preceq_{1} \calG'$ with the identity mapping, 
, and for any flow solution $\ff^{\calG'}$ of $\calG'$, again using the identity mapping, we have 
\[
\obj_{\calG}(\frac{1}{2}\ff^{\calG'})\geq \frac{1}{2}\obj_{\calG'}(\ff^{\calG'})-\delta\norm{\ff^{\calG'}}^2_2.
\]
\end{lemma}
To avoid using too many symbols, we reuse $\calG$ to refer to the original instance after the resistance rounding (i.e. the $\calG'$ above). Denote $E^r$ the subset of edges in $E^{\calG}\setminus T$ containing edges with $\rr_e=r$ for some particular $r$, note there are at most $\nbucket$ possible value of $r$. We work iteratively on the set $E^r$, starting with $E^r_0=E^r$. In the $i$-th iteration, we use $\textsc{FindPortal}(T,E^r_i,m/\kappa)$ (Lemma~\ref{lem:FindPortal}) to find a set of $m/\kappa$ portal vertices for the edges remaining in $E^r_i$, note the low-stretch spanning tree is fixed through the process, but each iteration we find a new set of portals using \textsc{FinalPortal} as introduced in
Section~\ref{subsec:treepartition}.

We then move edges in $E^r_i$ using the tree-portal routing. 
 We let $\Ghat^{r}_i$ to be the graph of the $m/\kappa$ portal nodes and the off-tree edges between them. Note the number of edges in $\Ghat^{r}_i$ is $|E^r_i|$ and the number of nodes is $m/\kappa$. 
 
So far we haven't specified the $\rr,s,\gg$ values on the edges in $\Ghat^{r}_i$, and these values will depend on the tree-portal routing as well as the average degree in $\Ghat^{r}_i$. For now we focus on discuss the edge set in our final sparsified instance, and assume we have $\rr,s,\gg$ values for $\Ghat^{r}_i$. We will come back to specify these quantities later. 

We use \textsc{Decompose} (Theorem~\ref{thm:Decompose}) on the graph $\Ghat^{r}_i$ to compute a collection of vertex disjoint sub-graphs $\{\Ghat^{r}_{i,1},\Ghat^{r}_{i,2},\ldots\}$, and at least half of the edges in $\Ghat^{r}_i$ are inside these sub-graphs. We let $\Ehat^r_i$ to be the edges contained in these sub-graphs, and 
$\Etil^r_i$ be the set of pre-images of edges in $\Ehat^r_i$ (in terms of tree-portal off-tree edge moving).
We remove $\Etil^r_i$ from $E^r_i$ and proceed to iteration $i+1$. If at the beginning of some iteration $i$, the size of $E^r_i$ is at most $\Otil(m/\kappa)$, we leave them as off-tree edges, and denote $E^{r}_{last}$ as the set containing them. Note for any $r$, the iterative process must finish in $O(\log \kappa)$ iterations as we start with $|E^r|<m$ edges. We do this for all $r$.   

So far any edge in the original instance $\calG$ we get either $(1)$ a tree edge in $T$, or 
$(2)$ an off-tree edge in a $\Ghat^{r}_{i,j}$ for some resistance value $r$, iteration $i$, and $j$-th expander computed in that iteration, or $(3)$ an off-tree edge remaining in $\E^{r}_{last}$ for some resistance value $r$.
There are $n-1$ edges in group $(1)$, and $\Otil(m/\kappa)$ edges in group $(3)$, so we can keep all these edges in the ultra-sparsifier $\calH$. 
 For the off-tree edges in group $(2)$, we uniformly sample the edges in each $\Ghat^{r}_{i,j}$ to get a sparsified graph $\Hbar^{r}_{i,j}$. Technically our sampling result only applies to an $\alpha$-uniform $\phi$-expander $\Ghat^{r}_{i,j}$ (i.e. case $3(a)$ in Theorem~\ref{thm:Decompose}). If the $\Ghat^{r}_{i,j}$ we get back from \textsc{Decompose} is in case $3(b)$ of, we perturb the gradient on edges so that the projection of the gradient to the cycle space of the expander is $0$, i.e. project the gradient to the space orthogonal to the cycle space. Then we have $\Ghat^{r}_{i,j}$ after perturbation is a $1$-uniform $\phi$-expander. 

The edges in our final ultra-sparsifier $\calH$ will be the tree edges in $T$, 
 the off-tree edges in the $\E^{r}_{last}$'s over all resistance bucket value $r$, and the off-tree edges in the $\Hbar^{r}_{i,j}$'s over all resistance $r$, iteration $i$ and expander $j$. We argued about the size of all but the edges in the $\Hbar^{r}_{i,j}$'s, which we will do now. 
\paragraph{Sampling Probability} We first specify the probability we sample each edge in $\Ghat^r_{i,j}$ to get $\Hbar^r_{i,j}$ which we denote by $\tau_{r,i}$ (same across all the expanders, i.e. $j$'s, for any resistance $r$ and iteration $i$). By Theorem \ref{thm:sampAndFixGrad} we need the probability to be at least $c_{sample}\log{n}(d_{\min}^{-1}\phi^{-2}+ \alpha m^{-1})$. Here $c_{sample}$ is a fixed constant across all $r,i,j$'s, and the guarantees on $\Ghat^r_{i,j}$ from Theorem \ref{thm:Decompose} allow us to use some fixed polylog$n$ as $\phi^{-2}$ and $\alpha$ across all $r,i,j$'s. The only parameter that varies across different $r,i$'s is $d_{\min}$, a lower bound on the minimum vertex degree in $\Ghat^r_{i,j}$, which by Theorem~\ref{thm:Decompose} is within a (fixed) polylog factor of the average degree in $\Ghat^{r}_i$. As there are $m_{r,i}$ edges and $m/\kappa$ nodes in $\Ghat^{r}_i$, the average degree is $m_{r,i}\kappa/m$. Thus, we can write $\tau_{r,i}=\frac{c_1 m \log^{c_2} n}{\kappa m_{r,i}}$ for some global constants $c_1,c_2$, and since both $m_{r,i}$ and $\kappa$ is at most $m$, $\tau_{r,i}$ satisfies the requirement on $\tau$ in Theorem \ref{thm:sampAndFixGrad}. With this particular choice of $\tau_{r,i}$ we can use \textsc{SampleAndFixGradient} to sample $\Ghat^r_{i,j}$ and the the guarantees from Theorem~\ref{thm:sampAndFixGrad}. Now we can prove the statement about the number of off-tree edges in $\calH$.
\begin{lemma}
\label{lem:size}
The total number of edges over all $\Hbar^{r}_{i,j}$'s is $\Otil(\frac{m}{\kappa})$ with high probability.
\end{lemma}
\begin{proof}
Pick any $r,i$, recall when we call \textsc{Decompose} in that iteration, we have $\Ghat^{r}_i$ with uniform $\rr$, $m_{r,i}= |E^r_i|$ edges and $n_i=m/\kappa$ nodes. From the previous discussion of the sampling probability, we know it is sufficient to call \textsc{SampleAndFixGradient} on $\Ghat^{r}_{i,j}$ with probability 
\[
\tau_{r,i}=\frac{c_1 m \log^{c_2} n}{\kappa m_{r,i}}
\]
for some constants $c_1,c_2$. By Theorem~\ref{thm:sampAndFixGrad}, the number of edges in $\Hbar^{r}_{i,j}$'s over all $j$ is at most $\widetilde{\Theta}(\frac{m}{\kappa})$ with high probability since over all $j$ the $\Ghat^{r}_{i,j}$'s contain $\Theta(m_{r,i})$ edges. 

Since for each $r$ the number of iterations is at most $i\leq \log\kappa$, and there are $\nbucket$ possible $r$ values, the final bound in the lemma follows from summing over all $r,i$. Note we can hide all $\log$ factors as $\log n$ factors by our assumption in Theorem~\ref{thm:ultrasparsify} that $\log\norm{\rr^{\calG}}_{\infty}$ and $\log \frac{1}{\delta}$ are both polylog in $n$.
\end{proof}

\begin{algorithm}
\caption{Producing Ultra-Sparsifier $\calH$
with unit $s^{\calH}=s^{\calG}$}
\label{alg:UltraSparsify}
 \begin{algorithmic}[1]
 \Procedure{UltraSparsify}{$\calG$, $\kappa$,$\delta$}
 \State{$T \leftarrow \textsc{LSST}(\rr^{\calG)}$. (low-stretch spanning tree)}
 \State{Initiate $\calH$ with $T$, and the identity flow mapping.}
 \State{Round $\rr^{\calG}$ down to nearest power of $2$, or $0$ if less than $\delta$}
 \State{$\nhat \leftarrow m / \kappa \quad$ (number of portal nodes per batch)}
	\For{Each bucket of resistance value $r$}
	 \State Let $i\leftarrow 0$, $E^r\leftarrow\{e|e\in E^{\calG}\setminus T, \rr^{\calG}_e=r\}$
	 \While{$E^r$ has more than $\Otil(m / \kappa)$ off-tree edges}
		 \State Let $m_{r,i}$ be the number of edges in $E^r$.
		 \State{Find $\nhat$ portal nodes to short-cut tree routing:
            \[
                \Vhat \leftarrow \textsc{FindPortal}\left(T,E^r,\nhat\right).
            \]
        }
		 \State{Route edges in $E^r$ along $T$, using portal nodes to short-cut tree-portal routing:}
		 \State $\Ghat^r_i\leftarrow \textsc{TreePortalRoute}(E^r,T,\Vhat)$\label{ln:TreeRoute},,
		 \State{Decompose the graph after tree-portal routing into uniform expanders:
		 \[\left\{\Ghat^r_{i,1},\Ghat^r_{i,2},\ldots\right\}\leftarrow \textsc{Decompose}\left(\Ghat^r_i,\delta/m^5\right).
        \]
        }
		 \State Remove the pre-image of edges in $\Ghat^r_{i,1},\Ghat^r_{i,2},\ldots$ from $E^r$.
		 \State{Set $\tau_{r,i}\leftarrow \frac{c_{1} m \log^{c_2}n}{m_{r,i} \kappa}$
		 (for sampling $\Ghat^r_{i,j}$'s in \textsc{SampleAndFixGradient})}
		 \For{ each $\Ghat^r_{i,j}$}
				\State Rescale the gradients and $\ell^p_p$ scalar as
				\begin{align}
				\rr^{\Ghat^r_{i,j}} &= r\kappa\log^2 n\\
				s^{\Ghat^r_{i,j}} & = \tau_{r,i}^{-p}\cdot s^{\calG}
				\end{align}
				\State Let $\Hbar^r_{i,j}\leftarrow \textsc{SampleAndFixGradient}\left(\Ghat^r_{i,j}, \tau_{r,i}\right)$. 
				
				\State {Add $\Hbar^r_{i,j}$ to $H$, and incorporate the flow mappings between $\Ghat^r_{i,j}$ and $\Hbar^r_{i,j}$
                \\ \qquad \qquad \qquad \qquad \qquad
                (composed with the tree-portal routing between $\Ghat^r_{i,j}$ and
                \\ \qquad \qquad \qquad \qquad  \qquad
                its pre-image) to the mapping between $\calG$ and $\calH$.}
		 \EndFor
	 \State $i\leftarrow i+1$
	 \EndWhile
	 \State{Add all remaining edges of $E^r$
	 to $H$ with the identity flow mapping on them}
	\EndFor
\State \Return $\calH$, $\map{\calH}{\calG}$, and $\kap{\calH}{\calG}=\Otil(\kappa m^{3/(p-1)})$.
 \EndProcedure 
 \end{algorithmic}
\end{algorithm}
\begin{algorithm}
\caption{Tree-Portal Routing of Edges}
\label{alg:treeportalroute}
 \begin{algorithmic}[1]
 \Procedure{TreePortalRoute}{$E$,$T$,$\Vhat$}
 \State Initialize $\Ehat\leftarrow \emptyset$ 
 \For{each $e=\{u,v\}\in E$}
 \State Let $\ehat\leftarrow \textsc{Move}_{T,\Vhat}(e)$, and $P_{T,\Vhat}(u,v)$ be its tree-portal path.
 \State (See Definition~\ref{defn:tree-portal})
 \State Let $\rr_{\ehat}$,$\ss_{\ehat}$ be the same as $\rr_e$ and $\ss_e$.
 \State Set $\gg_{\ehat}$ so that sending $1$ unit of flow from $u$ to $v$ along $P_{T,\Vhat}(u,v)$ has the same flow dot gradient as $\gg_e$, i.e. the flow dot gradient of sending directly along $e$. Note all edges on $P_{T,\Vhat}(u,v)$ other than $\ehat$ have known gradients.
 \State Add $\ehat$ to $\Ehat$ with $\gg_{\ehat},\rr_{\ehat},\ss_{\ehat}$ as specified. Note $\Ehat$ may contain multi-edges.
 \EndFor
 \Return $\Ehat$
 \EndProcedure 
 \end{algorithmic}
\end{algorithm}
Now we discuss the $\gg,\rr,s$ values we put on the edges in all the steps. Note we need the final instance $\calH$ to have a uniform scalar $s^{\calH}$ for every $|f_e|^p$ term, so we can recursively optimize the instance. However, in the intermediate steps, we will divide the instance into sub-instances induced by the different subsets of edges, e.g. $\Ghat^r_{i,j}$'s, and later combine sub-instances induced by the sampled sub-graphs $\Hbar^r_{i,j}$'s to get $\calH$. Each of these sub-instances will have its own scalar, e.g. $s^G_{r,i},s^H_{r,i}$, but in general they won't necessarily have the same value across different sub-instances. Notation-wise, in the following discussion, we assume each edge has its own scalar $\ss_e$ associated with the term $|f_e|^p$ in the intermediate instances. Eventually, the different scaling we do to $s$ in the intermediate steps will cancel so that in $\calH$ we have the scalar $s^{\calH}$. The input $\calG$ has a uniform scalar $s^{\calG}$, and we will make $s^{\calH}=s^{\calG}$.

Now we specify the $\gg$, $\rr$ and $s$ values of the edges in the final instance $\calH$ as well as in some of the key intermediate sub-instances we consider.
\begin{enumerate}
\item $e\in T$: The gradient, resistance and $\ss_e$ on these edges in $\calH$ remain the same as in $\calG$, that is $\gg^{\calH}_e=\gg^{\calG}_e$, $\rr^{\calH}_e=\rr^{\calG}_e$, and $\ss_e=s^{\calG}$.
\item $e\in E^r_{last}$: These off-tree edges remain at the end for each bucket $E^r$. We keep their gradient, resistance, and $\ss_e=s^{\calG}$ as in the original instance. 
\item $\ehat \in \Ghat^r_{i,j}$, the $j$-th expander computed in iteration $i$ for resistance $r$: In the intermediate sub-instance induced by $\Ghat^r_{i,j}$, we have $\rr_{\ehat}=r\kappa\log^2 n$, $\ss_{\ehat}=\tau_{r,i}^{-p}s^{\calG}$. For the gradient on $\ehat$, recall $\ehat$ is the image of some off-tree edge $e$ under the mapping $\textsc{Move}_{T,\Vhat^r_i}$ where $\Vhat^r_i$ is the set of portals in the $i$-th iteration for resistance $r$. Under the tree-portal routing, any flow along $e=(u,v)$ will be rerouted along the tree-portal path $P_{T,\Vhat^r_i}(u,v)$. We want the linear term (i.e. gradient times flow) in the objective function to remain the same under this rerouting, so routing $1$ unit of flow from $u$ to $v$ along $P_{T,\Vhat^r_i}(u,v)$ should give the same dot product with the gradients as routing $1$ unit of flow from $u$ to $v$ along $e$ in the original instance (i.e. $\gg^{\calG}_e$). As the only off-tree edge on the tree-portal path is $\ehat$, and we are keeping the original gradients on all the tree edges, this uniquely determines $\gg_{\ehat}$.
\item $\ebar \in \Hbar^r_{i,j}$ for some $r,i,j$: As specified in Theorem~\ref{thm:sampAndFixGrad}, if the edge $\ebar$ is sampled (with uniform probability $\tau_{r,i}$), and $\rr_{\ebar}$,$\ss_{\ebar}$ are their corresponding values in $\Ghat^r_{i,j}$ scaled up by $\tau_{r,i}$ and $\tau_{r,i}^p$ respectively. In particular we get back $\ss_{\ebar}=s^{\calG}$ as the $\tau_{r,i}^p$ scaling cancels the $\tau_{r,i}^{-p}$ scaling in $\Ghat^r_{i,j}$. 
\end{enumerate}
Note all the edges in our $\calH$ (i.e. group $1,2,4$ above) end up with the same scalar $s^{\calH}=s^{\calG}$.
\begin{table}                   %
\centering
\begin{threeparttable}
 \caption{Glossary of Notations in Algorithm and Analysis.}
\label{table:notations}
\begin{tabular}{|p{0.1\columnwidth} p{0.85\columnwidth}|} 
 \hline
 \multicolumn{2}{|c|}{Notations in \textsc{Ultrasparsify}}\\[0.5ex]
 \hline
 $\calG$ & Input instance with $\left(V^{\calG},E^{\calG},\gg^{\calG},\rr^{\calG},s^{\calG}\right)$.\\
$T$ & Low stretch spanning tree of $\calG$ (stretch with respect to $\rr^{\calG}$).\\
$E^r$ & All in $E^{\calG}\setminus T$ whose resistance after rounding is $r$.\\
$E^r_i$ & The remaining edges in $E^r$ at the $i$-th iteration of tree-portal routing $E^r$. \\
$\Ghat^r_i$ & The image of edges in $E^r_i$ by the mapping $\textsc{Move}_{T,\Vhat}$, i.e. moving off-tree edges along tree-portal path. The gradients of edges in $\Ghat^r_i$ are set to preserve the linear flow dot gradient term under tree-portal routing.\\
$m_{r,i}$ & The number of edges in $E^r_i$ (also the size of $\Ghat^r_i$).\\
$\Ghat^r_{i,j}$ & The $j$-th expander we get from decomposing $\Ghat^r_i$. Edges keep their gradients from $\Ghat^r_i$, and $\rr,s$ are scaled.\\
$\Ehat^r_i$ & The union of edges contained in the expanders $\Ghat^r_{i,j}$ (i.e. over all $j$'s).\\
$\Etil^r_i$ & The pre-image of edges in $\Ehat^r_i$. \\
$\E^r_{last}$ & The set of edges remaining in $E^r$ after the last iteration for $r$.\\
$\tau_{r,i}$ & The probability we use in \textsc{SampleAndFixGradient} to uniformly sample $\Ghat^r_{i,j}$.\\
$\Hbar^r_{i,j}$ & The sparsified graph of $\Ghat^r_{i,j}$ computed by \textsc{SampleAndFixGradient}. $\gg,\rr,s$ on edges are computed by the subroutine.\\
\hline
\multicolumn{2}{|c|}{Additional notations in the analysis}\\[0.5ex]
\hline
$\Gbar$ & The instance with the same edge set as $\calG$. Note{\tnote{$\dagger$}}\, $E^{\calG}=T + \sum_r E^r_{last} + \sum_{r,i}\Etil^r_i$. Edges in $\Gbar$ has the same $\gg,\rr,s$ as in $\calG$ except for those in $\sum_{r,i}\Etil^r_i$. For any resistance $r$ and iteration $i$, $e\in \Etil^r_i$ has the same gradient as in $\calG$, but $\rr_e=r\kappa\log^2n$, $s_e=\tau_{r,i}^{-p}s^{\calG}$ are scaled.\\
$\Gbar^r_i$ & The instance $\Gbar$ restricted to the set of edges in $\Etil^r_i$.\\
$\Gbar_{rest}$ & The instance $\Gbar$ restricted to the set of edges in $\sum_r E^r_{last}$.\\
\hline
\end{tabular}
\begin{tablenotes}
\item[$\dagger$] We use addition on sets as union but signify that the sets are disjoint.
\end{tablenotes}
\end{threeparttable}
\end{table}

Now we bound the approximation error. For simplicity, we carry out the analysis ignoring the additive errors in the bound, and defer the discussion of them to the end. In particular, additive errors come in at two cases. The first is when we round an original resistance to $0$ when it is less than $\delta$, and the second is in \textsc{Decompose}, we may get an expander whose projected gradient is not $\Otil(1)$-uniform but has tiny norm (i.e. case $3(b)$), and we zero out its projection to the cycle space before sampling. For now we assume we don't have these cases.

We summarize the notations in our algorithm and analysis in Table~\ref{table:notations}. We explicitly point out whenever we change the gradient, resistance or $s$ value on an edge. We will use instances and their underlying graphs interchangeably, and when we refer to a subgraph as an instance, it will be clear what are the $\gg,\rr,s$ values for the instance. 

First we let $\Gbar$ be the instance on the same nodes and edges as $\calG$, but for any $e\in \Etil^r_i$ (i.e. $e$ will be mapped to some $\ehat$ in $\Ghat^r_{i,j}$), we rescale the resistance and $s$ to be $\rr_e=r\kappa\log^2n$, and $\ss_e=\tau_{r,i}^{-p}s^{\calG}$ . Note the gradient of $e$ in $\Gbar$ stays the same as in $\calG$. We first bound the approximation error between $\calG$ and this rescaled instance $\Gbar$.
\begin{lemma}
\label{lem:scaling}
$\calG\preceq_{\Otil(m^{1/(p-1)}\kappa)} \Gbar \preceq_1 \calG$ with the identity mapping in both directions. 
\end{lemma}
\begin{proof}
For any edge $e$, we have $\gg^{\Gbar}_e=\gg^{\calG}_e$. As to the $\ell^p_p$ scalar, we have either $\ss^{\Gbar}_e= \ss^{\calG}_e$, or if $e$ is eventually moved to some $\Ghat^r_{i,j}$ then 
\[
\ss^{\Gbar}_e
=
\tau_{r,i}^{-p}s^{\calG}
=
\left(\frac{m_{r,i}\kappa}{c_1m\log^{c_2}n}\right)^p s^{\calG}
\]
as $m_{i,r}\geq \Otil{m/\kappa}$ or otherwise we would have stopped for resistance value $r$, we can assume $m_{r,i}\kappa\geq c_1m\log^{c_2}n$ so
\[
 \ss^{\calG}_e\leq  \ss^{\Gbar}_e \leq (\kappa^{p/(p-1)})^{p-1}\ss^{\calG}_e\leq (m^{1/(p-1)}\kappa)^{p-1}\ss^{\calG}_e\]
where the second inequality is by $m_{r,i}\leq m$, and the third inequality is by $\kappa<m$. Similar calculation gives $\rr^{\calG}_e\leq \rr^{\Gbar}_e \leq \kappa\log^2n\cdot\rr^{\calG}_e$. Our result directly follow by Lemma~\ref{lem:PerturbResistances}.
\end{proof} 
 Now we break $\Gbar$ into sub-instances induced on the disjoint edge sets. Let $\Ebar^r_i$ be the instance of $\Gbar$ restricted to edges in $\Etil^r_i$, $T$ the instance restricted to the tree edges, and $\Gbar_{rest}$ the instance restricted to edges in any of the $E^r_{last}$'s. When use addition as union on sets when the sets are disjoint. The objective of the sum of two instances is simply the sum of the individual instances objectives.
\begin{lemma}
\label{lem:treeroute}
For any resistance value $r$, round $i$, we have 
\[
T + \Ebar^r_i \preceq_{\Otil(m^{1/(p-1)})} T+\Ehat^r_i \preceq_{\Otil(m^{1/(p-1)})} T + \Ebar^r_i 
\] where the flow mapping is the tree-portal routing and its reverse. 
\end{lemma}
\begin{proof}
Fix any resistance value $r$ and iteration $i$, the set of remaining off-tree edges of resistance $r$ in iteration $i$ is $E^r_i$, and these edges have a total stretch at most $O(m\log n\log\log n)$ with $T$ by Lemma~\ref{lem:LSST}, and $E^r_i=m_{r,i}$. As we use \textsc{FindPortal} to get a set of $m/\kappa$ portal nodes $\Vhat$ in that iteration, by Lemma~\ref{lem:FindPortal}, for any edge $e'$ on $T$, we have in $T+E^r_i$
\[
W_{e'}\defeq \sum_{e\in E^r_i: e' \in P_{T, \Vhat}\left( e \right)}
    \str_{T, \Vhat}\left( e \right)  \leq \frac{10}{\nhat}\sum_{e\in E^r_i} \str_{T, \Vhat}\left( e \right)\leq 10\kappa\log n \log\log n\leq \kappa\log^2 n
\] and 
\[
K_{e'} \defeq \abs{e\in E^r_i: e' \in P_{T, \Vhat}\left( e \right)} \leq \frac{10 \mhat}{\nhat}\leq \frac{10\kappa m_{r,i}}{m}
\]
We first look at the direction from from $T + \Ebar^r_i$ to $T+\Ehat^r_i$. 
Let $\ff$ be the flow in $T + \Ebar^r_i$, and $\ffhat$ be the tree-portal routing of 
$\ff$. In the tree-portal routing, flow on tree edges is mapped to the same 
flow, while any flow along an off-tree edge $\ebar=(u,v)\in \Ebar^r_i$ is 
rerouted along the tree-portal path $P_{T,\Vhat}(u,v)$. This rerouting 
clearly preserves the residue between $\ff,\ffhat$, and if $\ehat\in \Ehat^r_i$ is 
the image of $(u,v)$, its gradient $\gg_{\ehat}$ in $\Ehat^r_i$ is by 
construction set to be the value which preserves the linear term in the 
objective function for $\ff$ and $\ffhat$. The cost of $\ell^2_2$ and $\ell^p_p$ terms 
for $f_{\ebar}$ is the same as the corresponding costs for $\ffhat_{\ehat}$, since 
$\ehat$ is only used for the rerouting of $\ebar$ (so $|\ffhat_{\ehat}|=|f_{\ebar}|
$), and they have the same $\rr,\ss$ values. Thus, the contribution to the $
\ell^2_2,\ell^p_p$ terms in objective function from the off-tree edges are the 
same for $\ff$ and $\ffhat$. The only extra cost comes from the $\ell^2_2$ and $\ell^p_p$
 terms of tree edges for $\ffhat$ since we put additional flow through them.  
First consider the sum of the $\ell^p_p$ terms over all tree edges for $\ffhat$ in $T
+\Ehat^r_i$. Recall we don't scale the $s$ value for tree edges, so the 
scalar is still $s^{\calG}$ on tree edges, while for off-tree edges in $\Ebar^
r_i$, the value $s$ is scaled to be $\left(\frac{\kappa m_{r,i}}{c_1 m \log^{c
_2}n}\right)^ps^{\calG}$
\begin{align*}
\sum_{e'\in T} s^{\calG}\abs{\ffhat_{e'}}^p =&\sum_{e'\in T}s^{\calG}\abs{\sum_{\ebar
:e'\in P_{T,\Vhat}(\ebar)} \ff_{\ebar}}^p\\
=& \sum_{e'\in T}s^{\calG}K^p_{e'}\abs{\sum_{\ebar:e'\in P_{T,\Vhat}(\ebar)}
\frac{1}{K_{e'}}  \ff_{\ebar}}^p\\
\leq & \sum_{e'\in T}s^{\calG}K^p_{e'}\sum_{\ebar:e'\in P_{T,\Vhat}(\ebar)}\frac{1
}{K_{e'}} \abs{\ff_{\ebar}}^p && (\textrm{Using Jensen's
     inequality}) \\
= & \sum_{e'\in T}s^{\calG}K^{p-1}_{e'}\sum_{\ebar:e'\in P_{T,\Vhat}(\ebar)
} \abs{\ff_{\ebar}}^p\\
\leq & \sum_{\ebar}\abs{\ff_{\ebar}}^p\sum_{e'\in P_{T,\Vhat}(\ebar)}s^{\calG}K^{p-1}_
{e'} \\
\leq & \sum_{\ebar}\abs{\ff_{\ebar}}^p m \cdot s^{\calG}K^{p-1}_{e'} && (\textrm{Tree-
portal path's length $<m$}) \\
\leq & \sum_{\ebar} (10c_1\log^{c_2} n)^{p-1} m\cdot \ss_{\ebar} \abs{\ff_{\ebar}}^p
\end{align*}
So the $\ell_p^p$ term goes up by at most a factor $(10c_1\log^{c_2} n)^{p-1}
 m$. Similar calculation shows that the $\ell_2^2$ term goes up by at most a 
constant factor by the tree-portal routing. Thus, we get $T + \Ebar^r_i 
\preceq_{\Otil(m^{1/(p-1)})} T+\Ehat^r_i$. The other direction is symmetric 
using the reverse tree-portal routing, and the calculation stays the same 
since the tree-portal routing in reverse incurs the same load/congestion on 
tree edges.
\end{proof}
If we put the $\Gbar^r_i$ over all resistance $r$'s and round $i$'s together, we get
\begin{lemma}
\label{lem:treeroute2}
\[T + \sum_{r,i}\Gbar^r_i \preceq_{\Otil(m^{1/(p-1)})} T+\sum_{r,i}\Ehat^r_i
\preceq_{\Otil(m^{1/(p-1)})} T + \sum_{r,i}\Gbar^r_i \]
The sum is over all possible resistance value $r$'s, and over all iterations $
i$ for $r$.
\end{lemma}
\begin{proof}
By Lemma~\ref{lem:treeroute} and Lemma~\ref{lem:Composition} we have 
\[
\bigcup_{r,i} T + \Gbar^r_i\preceq_{\Otil(m^{1/(p-1)})} \bigcup_{r,i}T+\Ehat^r_i
\preceq_{\Otil(m^{1/(p-1)})} \bigcup_{r,i} T + \Gbar^r_i
\]
Note the $\Gbar^r_i$'s (and the $\Ehat^r_i$'s) are disjoint for different resistance values or different iterations,
thus these edges contribution to the objective function value simply adds up. 
For the tree edges, since there are at most $\log^2 n$ different pairs of resistance and iteration pairs, we have 
\[
T \preceq_{1} \bigcup_{r,i} T +  \preceq_{\log^2 n} T 
\]
by considering the mapping that split flow on one tree edge to $\log^2n$ copies of it and the reverse mapping of merging. Note $\abs{a_1}^x+\ldots+\abs{a_1}^x \leq (\abs{a_1}+\ldots+\abs{a_k})^x\leq k(\abs{a_1}^x+\ldots+\abs{a_1}^x)$. This gives the final result we want.
\end{proof}
Note that $\Gbar$ is the disjoint union of $T+\sum_{r,i}\Gbar^r_i+\Gbar_
{rest}$, while $\calH$ is the disjoint union of $T+\sum_{r,i}\Hbar^r_i+
\Gbar_{rest}$. Thus, we can show the following
\begin{lemma}
\label{lem:sampling}
$\Gbar \preceq_{\Otil(m^{2/(p-1)})} \calH \preceq_{\Otil(m^{2/(p-1)})} \Gbar$.
\end{lemma}
\begin{proof}
Recall for each resistance value $r$, in the $i$-th round, $\Ghat^r_{i,j}$ is 
the $j$-th uniform expander we find, and $\Hbar^r_{i,j}$ is the sparsified graph of $\Ghat^r_{i,j}$.
\begin{align*}
\Gbar & = T+\Gbar_{rest}+\sum_{r,i}\Gbar^r_i && (\textrm{valid as the sets 
are disjoint})\\
   & \preceq_{\Otil(m^{1/(p-1)})} T+\Gbar_{rest}+\sum_{r,i}\Ehat^r_i && (
\textrm{Lemma~\ref{lem:treeroute2}})\\
	 & = T+\Gbar_{rest}+\sum_{r,i,j}\Ghat^r_{i,j}  && (\textrm{$\Ehat^r_i$ is 
the disjoint union of $\Ghat^r_{i,j}$ over all $j$})\\
	 & \preceq_{\Otil(m^{1/(p-1)})} T+\Gbar_{rest}+\sum_{r,i,j}\Hhat^r_{i,j} && 
(\textrm{By Theorem~\ref{thm:sampAndFixGrad}, and sets being disjoint})\\
	&= \calH
\end{align*}
$\Gbar \preceq_{\Otil(m^{2/(
p-1)})} \calH $ follows by taking the composition of all the intermediate steps, and multiplying the approximation error by Lemma~\ref{lem:approximations:composition}. The other direction is similar.
\end{proof}
Now we can prove the main ultra-sparsification theorem.
\begin{proof}[Proof of Theorem~\ref{thm:ultrasparsify}]
Other than the additive error terms and the self-loops, everything in the theorem statement follow directly from Lemma~\ref{lem:size} (the number of off-tree edges), and composition of Lemma~\ref{lem:scaling} with Lemma~\ref{lem:sampling} (the approximation error). We explicitly spell out the flow mappings between $\calG$ and $\calH$. We start with the $\calG$ to $\calH$ direction. We break the flow in $\calG$ as the sum of flow on disjoint edge subsets $T$,$\Gbar_{rest}$, and $\Etil^r_i$, specify the mapping from each piece to $\calH$, and later take the sum of the mappings. The mapping from $T$ and $\Gbar_{rest}$ to $\calH$ is just the identity. For flow on $\Etil^r_i$, we get a flow on $T+\Ehat^r_i$ by tree-portal routing. As $\Ehat^r_i$ is the sum of $\Ghat^r_{i,j}$'s, for the flow mapped to $\Ghat^r_{i,j}$, we map it to a flow on $\Hbar^r_{i,j}$ using the flow mapping in \textsc{SampleAndFixGradient}. We add these mapping over all $j$'s to get a mapping from the flow on $T+\Ehat^r_i$ to a flow on $\calH$, and take the composition with the tree-portal routing to get a mapping from $\Etil^r_i$ to $\calH$. Summing over all $r,i$ (together with the identity on $T$ and $\Gbar_{rest}$ gives the mapping from $\calG$ to $\calH$. The mapping from $\calH$ to $\calG$ is symmetric, and in the part from $\Ehat^r_i$ to $T+ \Etil^r_i$ we use the reverse of tree-portal routing.

All the subroutines take nearly linear time, and we have at most $\log n$ different $r$, and for each $r$ there are at most $\log m$ iterations, so the overall running time is $\Otil(m)$. The flow mappings can also be applied in $\Otil(m)$ time, and they are linear maps.

Now we look at the additive error terms. In particular, additive errors come in at two places. The first is when we round an original resistance to $0$ when it is less than $\delta$, and we have Lemma~\ref{lemma:rounding} to bound the error (at that step). The second place is in \textsc{Decompose} (Algorithm~\ref{alg:Decompose}), we may get an expander $\Ghat^r_{i,j}$ whose projected gradient is not $\alpha$-uniform but has tiny norm (i.e. case $3(b)$), and we zero out its projection to the cycle space before sampling to make it $1$-uniform. If we have a flow $f$ on such an $\Ghat^r_{i,j}$, the additive error is in the linear term, and is equal to the dot product of $f$ with the removed gradient. We let $\gg^r_i$, $\gg^r_{i,j}$ be the gradient on edges in $\Ghat^r_i$,$\Ghat^r_{i,j}$ respectively, and $\gghat^r_i$, $\gghat^r_{i,j}$ as the projection of $\gg^r_i$ (and $\gg^r_{i,j}$) to the cycle space of $\Ghat^r_i$ (and $\Ghat^r_{i,j}$). We remove $\gghat^r_{i,j}$ from the gradient $\gg^r_{i,j}$ when $\gghat^r_{i,j}\leq \delta'\gghat^r_i$ for some parameter $\delta'$, so the additive error we introduce is $f^T\gghat^r_{i,j}$, which is at most $\norm{f}_2\norm{\gghat^r_{i,j}}_2$, which is in turn at most $\delta'\norm{f}_2\norm{\gg^r_i}_2$ as $\gghat^r_i$ is a projection of $\gg^r_i$. Now we look at how this additive error propagates in terms of the overall approximation error between $\calG$ and $\calH$. We will get an additional factor $m$ when we combine the additive errors over all the individual expanders where we carry out this perturbation. Note we are not really introducing more error here, but simply because $\sqrt{m}\norm{\sum_i f_i}\geq \sum_i\norm{f_i}\geq \norm{\sum_i f_i}$ when $f_i$'s have disjoint support and total size $m$. The additive error is also amplified through the intermediate steps, but since the multiplicative approximation errors are $m^O{1/p}$, we lose at most another polynomial factor. Additional polynomial factor comes in because the norm of the gradient vector after tree-routing can be off by a polynomial factor comparing to the norm of the original gradient. However, overall the blowup is at most polynomial, and we use a polynomially smaller $\delta'$ in \textsc{Decompose} to accommodate these factors to get the additive error in our final result. The same argument applies to the additive error introduced by resistance rounding (e.g. round to $0$ when the gradient is at most $\delta/m^c$ for some large enough $c$). 
\end{proof}
We brief go over the case when tree-portal routing gives self-loops. We treat self-loops the same way as the edges that are in the uniform expanders except they don't go through the expander decomposition and sampling steps. Once we get a self-loop $\ehat$ from tree-portal routing of some edge $e\in E^{\calG}$, we add $\ehat$ to $\calH$, where the gradient on $\ehat$ is set (the same way as non self-loops) to preserve the flow dot gradient term under tree-portal routing. We remove its pre-image $e$ from $E^r_i$, but if in some iteration, more than half of the edges in $E^r_i$ are mapped to self-loops by tree-portal routing, we skip the decomposition and sampling steps also for other edges, as we don't have a dense enough graph between the portal nodes to sparsify. We still have the size of $E^r_i$ drop by at least $1/2$ across each iteration as before. The final caveat is that since self-loops don't go through \textsc{SampleAndFixGradient}, and thus their $s$ values are not scaled to be the same as the rest of the edges in $\calH$. This is not an issue because we will remove them from the instance and optimize them individually (see Lemma~\ref{lem:elimination:selfLoops}), so they won't exist in the instance that we recursively solve, so uniform $s$ scalar is not required for them.
%
%



\section{Decomposing into Uniform Expanders}
\label{sec:decomposition}

In this section we prove our decomposition result necessary
for finding large portions of edges that can be sampled.
This and the subsequent sampling step in Appendix~\ref{sec:expander}
are critical for reducing the number of edges between portal vertices,
after they were routed there in Line~\ref{ln:TreeRoute}
of \textsc{UltraSparsify} (Algorithm~\ref{alg:UltraSparsify}).
The main algorithmic guarantees can be summarized as below in
Theorem~\ref{thm:Decompose}.

\Decompose*

We will obtain the expansion properties via expander
decompositions.
Specifically we will invoke the following result
from~\cite{SaranurakW18:prelim} as a black box.

\begin{lemma}
\label{lem:ExpanderDecompose}
There is a routine \textsc{ExpanderDecompose} that when
given any graph $G$ and any degrees $\dd$ such that
$\dd_u \geq deg_{G}(u)$ for all $u$, along with a parameter
$0 < \phi < 1$,
$\textsc{ExpanderDecompose}(G, \dd, \phi)$
returns a partition of the vertices of $G$
into $V_1, V_2, \ldots$ in $O(m \phi^{-1} \log^{4}n)$ time
such that $G[V_i]$ has conductance at least $\phi$ w.r.t. $\dd_u$,
and the number of edges between the $V_i$s is at most
$O(\sum_{u} \dd_{u} \phi \log^{3}n)$. 
\end{lemma}

Note that we explicitly introduce the $\dd$ vector containing
the degrees of the initial graph because we will repeatedly invoke
this partition routine.
This is due to our other half of the routine, which is to
repeatedly project $\gg$ among the remaining edges, and removing
the ones that contribute to too much of its $\ell_2^2$-norm
in order to ensure uniformity as given in Case~\ref{case:Uniform}
of Theorem~\ref{thm:Decompose}.
To see that this process makes progress, we need the key observation
from Lemma~\ref{lem:EnergyDecrease} that projections can only decrease
the $\ell_2^2$ norm of $\gghat$, the projection of the gradient.

This leads to an approach where we alternate between
dropping the edges with high energy, and repartitioning
the remaining edges into expanders.
Pseudocode of this routine is in Algorithm~\ref{alg:Decompose},
which calls a recursive routine, \textsc{DecomposeRecursive}
shown in Algorithm~\ref{alg:DecomposeRecursive} with a suitable
value of $\phi$ and number of layers.
Note that we also need to trim the initial graph so
that we only work with large degree vertices.

\begin{algorithm}[t]
\caption{Decomposition into Uniform Expanders}
\label{alg:Decompose}
 \begin{algorithmic}[1]
  \Procedure{Decompose}{$\calG$, $\delta$}
  \State Set $\phi \leftarrow c_{partition} \log^{3}n \log(1 / \delta)$
    for some absolute constant $c_{partition}$.
  \State Iteratively remove all vertices with degree less than $\frac{m}{10n}$
  to form $\calG_{large}$.
  \State Compute $\overline{\dd}$, the degrees of $\calG_{large}$
  \State Return $\textsc{RecursiveDecompose}(\calG_{large}, \phi, 1, log(n / \delta))$.
\EndProcedure 
\end{algorithmic}
\end{algorithm}

\begin{algorithm}[t]
\caption{Recursive Helper for Decomposition}
\label{alg:DecomposeRecursive}
 \begin{algorithmic}[1]
  \State{Compute the projection of $\gg^{\calG}$ into its cycle space, $\gghat^{\calG}$.}
  \Procedure{DecomposeRecursive}{$\calG$, $\ddbar$, $\phi$, i, L}
  \State{Form $\calG_{trimmed}$ by removing all edges
    $e \in E^{\calG}$ such that $(\gghat^{\calG}_e)^2 \geq \frac{10 L}{m^{\calG}} \cdot \norm{\gghat^{\calG}}_2^2$.}
  \label{ln:Trim}
  \State $(G_1, G_2, \ldots, G_t) \leftarrow \textsc{ExpanderDecompose}((V^{\calG_{trimmed}}, E^{\calG_{trimmed}}), \ddbar_{V^{\calG_{trimmed}}}, \phi)$.
  \State Initialize collection of results, $\calP^{\calG} \leftarrow \emptyset$.
  \For{$i = 1 \ldots t$}
    \State{Form $\calG_i$ from the edges in $\calG$ corresponding to $G_i$}
    \State{Compute $\gghat^{\calG_i}$, the projection of $\gg(\calG_{i})$ onto its cycle space.}
    \If{$i = L
      \text{~or~}
        (\|\gghat^{\calG_i}\|_2^2 \geq \frac{1}{2} \|\gghat^{\calG}\|_2^2
          \text{~and~}
          m^{\calG_{i}} \geq m^{\calG}/2)$}
      \State{Add $\calG_{i}$ to the results, $\calP^{\calG} \leftarrow \calP^{\calG} + \calG_{i}$.}
    \Else
      \State{Recurse on $\calP^{\calG}$: $\calP^{\calG} \leftarrow \calP^{\calG} + \textsc{DecomposeRecursive}(\calG_{i}, \ddbar, \phi, i + 1, L)$.} \label{ln:Recurse}
    \EndIf
  \EndFor
  \State Return $\calP^{\calG}$.
\EndProcedure 
\end{algorithmic}
\end{algorithm}

We will also need the following result (Lemma 28 of
\cite{KelnerLOS14}, see also \cite{KelnerM11}).
\begin{lemma}
\label{lem:electricalOblInfRoute}
Suppose $G$ is a unit weight graph with conductance $\phi$.
Then the projection operations into cycle and potential flow
spaces both have $\ell_{\infty}$ norms bounded by $O(\phi^{-2} \log{n})$:
\[
\norm{\BB^{\calG} \left( {\BB^{\calG}}^{\top} \BB^{\calG}
  \right)^{\dag} {\BB^{\calG}}^{\top}}_{\infty}
\leq
O(\phi^{-2} \log n)
\]
and
\[
\norm{I-\BB^{\calG} \left( {\BB^{\calG}}^{\top} \BB^{\calG}
  \right)^{\dag} {\BB^{\calG}}^{\top}}_{\infty}
\leq
O(\phi^{-2} \log n).
\]
\end{lemma}

\begin{proof} (of Theorem~\ref{thm:Decompose})

We start by bounding the qualities of the $\calG$
pieces returned.
As we only return pieces that are the outputs of
$\textsc{ExpanderDecompose}$, all of them have conductance
at least $\phi$ by Lemma~\ref{lem:ExpanderDecompose}.
Also, since we only keep the non-trivial pieces containing edges,
taking the singleton cuts gives that the degrees in these pieces
are at least
\[
\phi \cdot \dd_{u}
\geq
\frac{m}{2n}
\cdot \frac{1}{c_{partition} \log^3{n}
\cdot \log\left( n / \delta \right)}
=
\frac{m}{10 n \log^3{n} \log\left( n / \delta \right)}.
\]

Now consider the quality of each $\gghat^{\calG_i}$:
if it was returned due to $i = L$, then the energy of the
projected gradient must have been halved
at least $L - \log{n}$ times, or by a factor of
$2^{L - \log{n}} = 2^{\log(1 / \delta)} = 1/\delta$.
Thus we would have 
\[
\norm{\gghat^{\calG_i}}_2^2
\leq
\delta \norm{\gghat^{\calG}}_2^2
\leq
\delta \norm{\gg^{\calG}}_2^2.
\]

Otherwise, we must have terminated because both the energy
and edge count did not decrease too much.
An edge $e$ was kept in the trimmed set only if
\[
\left(\gghat^{\calG}_e\right)^2
\leq
\frac{10 L}{m^{\calG}} \norm{\gghat^{\calG}}_2^2
\leq
\frac{20 L}{m^{\calG_{i}}} \norm{\gghat^{\calG}}_2^2.
\]
Combining this with the termination requirement
of $\norm{\gghat^{\calG_{i}}}_2^2 \geq \frac{1}{2}
\norm{\gghat^{\calG}}_2^2$ gives that
the $\ell_{\infty}$ norm of the pre-projection
gradient on $\calG_i$, $\gghat^{\calG}_{E^{\calG_{i}}}$
satisfies
\[
\norm{\gghat^{\calG}_{E^{\calG_{i}}}}_{\infty}^2
\leq
\frac{40 L}{m^{\calG_{i}}} \norm{\gghat^{\calG_i}}_2^2.
\]
On the other hand, because $\calG_{i}$ has expansion $\phi$,
doing an orthogonal cycle projection on it can only increase
the $\ell_{\infty}$-norm of a vector by a factor of
$O(\phi^{-2} \log{n})$ by Lemma~\ref{lem:electricalOblInfRoute}.
Thus we have
\begin{multline*}
\norm{\gghat^{\calG_{i}}}_{\infty}^2
=
\norm{\left( \II -
\BB^{\calG} \left( {\BB^{\calG}}^{\top} \BB^{\calG} \right)^{\dag}
{\BB^{\calG}}^{\top} \right) \gghat^{\calG}_{E^{\calG_{i}}}}_{\infty}^2
\leq
\norm{\II -
\BB^{\calG} \left( {\BB^{\calG}}^{\top} \BB^{\calG} \right)^{\dag} {\BB^{\calG}}^{\top}}_{\infty}^{2}
\norm{\gghat^{\calG}_{E^{\calG_{i}}}}_{\infty}^2\\
\leq
O\left( \phi^{-4} \log^{2}n \right)
\cdot
\norm{\gghat^{\calG}_{E^{\calG_{i}}}}_{\infty}^2
\leq
O\left( \phi^{-4} \log^{2}n \right)
\cdot 
O\left( \frac{L}{m^{\calG_{i}}} \right) \norm{\gghat^{\calG_{i}}}_2^2
=
\frac{O\left( \log^{14}n \log^{5}\left( n / \delta \right)\right)}
{m^{\calG_{i}}} \norm{\gghat^{\calG_{i}}}_2^2,
\end{multline*}
which is the desired (post-projection) uniformity bound.

We now bound the number of edges removed during all
the recursive calls.
The bound on $L$ means this recursion has
at most $O(\log(n / \delta))$ levels.
Lemma~\ref{lem:ExpanderDecompose} gives that the number
of edges between the expander pieces is
\[
O\left( \sum_{u} \dd_{u} \phi \log^{3}n \right) 
\cdot \log\left( n / \delta \right)
=
O\left( m \phi \log^{3}n \log\left( n / \delta \right) \right),
\]
so the setting of $\phi = \frac{1}{c_{partition} \log^{3}n \log(n / \delta))}$
gives at most $m / 10$ edges between the pieces
for an appropriate choice of $c_{partition}$.

Furthermore, as each edge's contribution to $\gghat$ is
non-negative, the number of edges whose relative contribution
exceed $\frac{10 L}{m}$ is at most $\frac{m}{10 L}$.
Summing this over all levels gives at most $m / 10$
edges removed from the trimming step on Line~\ref{ln:Trim}
of \textsc{DecomposeRecursive} in Algorithm~\ref{alg:DecomposeRecursive}.

Finally, the running time is dominated by the expander
decomposition calls.
As there are $O(\log(n / \delta))$ levels of recursion and
each level deals with edge-disjoint subsets, we obtain
the total running time by substituting the value of $\phi$
into the runtime of expander decompositions as given in 
Lemma~\ref{lem:ExpanderDecompose}.
\end{proof}


\section*{Acknowledgements}
\label{sec:acks}
\addcontentsline{toc}{section}{\nameref{sec:acks}}

This project would not have been possible without Dan Spielman's
optimism about the existence of analogs of numerical methods
for $\ell_{p}$-norms, which he has expressed to us on multiple occasions over the past six years.
We also thank Ainesh Bakshi, Jelani Nelson, Aaron Schild, and Junxing Wang
for comments and suggestions on earlier drafts and presentations of these ideas.

As with many recent works in optimization algorithms on graphs,
this project has its large share of influence by the late Michael B. Cohen.
In fact, Michael's first papers on recursive
preconditioning~\cite{CohenKMPPRX14} and $\ell_{p}$-norm preserving
sampling of matrices~\cite{CohenP15} directly influenced the
constructions of preconditioners (Section~\ref{subsec:graph_ultrasparsify})
and uniform expanders (Section~\ref{sec:decomposition} and
Appendix~\ref{sec:expander}) respectively.
While our overall algorithm falls short of what Michael would consider `snazzy',
it's also striking how many aspects of it he predicted, including:
the use of expander decompositions;
the $p \rightarrow \infty$ case being different than the $p \rightarrow 1$
case;
and the large initial dependence on $p$ that's also eventually
fixable (see Section~\ref{subsec:open}).


Richard regrets not being able to convince Michael to systematically investigate
preconditioning for $\ell_{p}$-norm flows.
He is deeply grateful to
Aleksander M\c{a}dry, Jon Kelner,
Ian Munro, Tom Cohen, Marie Cohen, Sebastian Bubeck,
and Ilya Razenshteyn for many helpful conversations following Michael's passing.

\newcommand{\etalchar}[1]{$^{#1}$}


\begin{thebibliography}{ALdOW17}

\bibitem[AKPS19]{AdilKPS19}
Deeksha Adil, Rasmus Kyng, Richard Peng, and Sushant Sachdeva.
\newblock Iterative refinement for $\ell_p$-norm regression.
\newblock In {\em Proceedings of the Thirtieth Annual {ACM-SIAM} Symposium on
  Discrete Algorithms, {SODA} 2019}, 2019.

\bibitem[AL11]{AlamgirL11}
Morteza Alamgir and Ulrike~V Luxburg.
\newblock Phase transition in the family of p-resistances.
\newblock In {\em Advances in Neural Information Processing Systems}, pages
  379--387, 2011.

\bibitem[ALdOW17]{AllenZhuLOW17}
Zeyuan Allen{-}Zhu, Yuanzhi Li, Rafael~Mendes de~Oliveira, and Avi Wigderson.
\newblock Much faster algorithms for matrix scaling.
\newblock In {\em Symposium on Foundations of Computer Science (FOCS)}, pages
  890--901, 2017.
\newblock Available at: https://arxiv.org/abs/1704.02315.

\bibitem[AN12]{AbrahamN12}
Ittai Abraham and Ofer Neiman.
\newblock Using petal-decompositions to build a low stretch spanning tree.
\newblock In {\em STOC}, 2012.

\bibitem[BCLL18]{BubeckCLL18}
S{\'e}bastien Bubeck, Michael~B. Cohen, Yin~Tat Lee, and Yuanzhi Li.
\newblock An homotopy method for lp regression provably beyond self-concordance
  and in input-sparsity time.
\newblock In {\em Proceedings of the 50th Annual ACM SIGACT Symposium on Theory
  of Computing}, STOC 2018, pages 1130--1137, New York, NY, USA, 2018. ACM.

\bibitem[BK96]{BenczurK96}
Andr\'{a}s~A. Bencz\'{u}r and David~R. Karger.
\newblock Approximating s-t minimum cuts in $\tilde{O}(n^2)$ time.
\newblock In {\em Proceedings of the twenty-eighth annual ACM symposium on
  Theory of computing}, STOC '96, pages 47--55, New York, NY, USA, 1996. ACM.

\bibitem[BKKL17]{BeckerKKL17}
Ruben Becker, Andreas Karrenbauer, Sebastian Krinninger, and Christoph Lenzen.
\newblock Near-optimal approximate shortest paths and transshipment in
  distributed and streaming models.
\newblock In {\em 31st International Symposium on Distributed Computing, {DISC}
  2017, October 16-20, 2017, Vienna, Austria}, pages 7:1--7:16, 2017.
\newblock Available at: https://arxiv.org/abs/1607.05127.

\bibitem[CKM{\etalchar{+}}11]{ChristianoKMST10}
Paul Christiano, Jonathan~A. Kelner, Aleksander Madry, Daniel~A. Spielman, and
  Shang-Hua Teng.
\newblock Electrical flows, laplacian systems, and faster approximation of
  maximum flow in undirected graphs.
\newblock In {\em Proceedings of the 43rd annual ACM symposium on Theory of
  computing}, STOC '11, pages 273--282, New York, NY, USA, 2011. ACM.
\newblock Available at http://arxiv.org/abs/1010.2921.

\bibitem[CKM{\etalchar{+}}14]{CohenKMPPRX14}
Michael~B. Cohen, Rasmus Kyng, Gary~L. Miller, Jakub~W. Pachocki, Richard Peng,
  Anup Rao, and Shen~Chen Xu.
\newblock Solving {SDD} linear systems in nearly $m \log^{1/2} n$ time.
\newblock In {\em STOC}, pages 343--352, 2014.

\bibitem[CMMP13]{ChinMMP13}
Hui~Han Chin, Aleksander M{{a}}dry, Gary~L. Miller, and Richard Peng.
\newblock Runtime guarantees for regression problems.
\newblock In {\em Proceedings of the 4\textsuperscript{th} conference on
  Innovations in Theoretical Computer Science}, ITCS '13, pages 269--282, New
  York, NY, USA, 2013. ACM.
\newblock Available at http://arxiv.org/abs/1110.1358.

\bibitem[CMTV17]{CohenMTV17}
Michael~B. Cohen, Aleksander Madry, Dimitris Tsipras, and Adrian Vladu.
\newblock Matrix scaling and balancing via box constrained newton's method and
  interior point methods.
\newblock In {\em 58th {IEEE} Annual Symposium on Foundations of Computer
  Science, {FOCS} 2017, Berkeley, CA, USA, October 15-17, 2017}, pages
  902--913, 2017.
\newblock Available at: https://arxiv.org/abs/1704.02310.

\bibitem[CP15]{CohenP15}
Michael~B. Cohen and Richard Peng.
\newblock $\ell_p$ row sampling by {L}ewis weights.
\newblock In {\em Proceedings of the Forty-Seventh Annual ACM on Symposium on
  Theory of Computing}, STOC '15, pages 183--192, New York, NY, USA, 2015. ACM.
\newblock Available at http://arxiv.org/abs/1412.0588.

\bibitem[DS08]{DaitchS08}
Samuel~I. Daitch and Daniel~A. Spielman.
\newblock Faster approximate lossy generalized flow via interior point
  algorithms.
\newblock In {\em Proceedings of the 40th annual ACM symposium on Theory of
  computing}, STOC '08, pages 451--460, New York, NY, USA, 2008. ACM.
\newblock Available at http://arxiv.org/abs/0803.0988.

\bibitem[EACR{\etalchar{+}}16]{ElalaouiCRWJ16}
Ahmed El~Alaoui, Xiang Cheng, Aaditya Ramdas, Martin~J Wainwright, and
  Michael~I Jordan.
\newblock Asymptotic behavior of $\ell_p$-based {L}aplacian regularization in
  semi-supervised learning.
\newblock In {\em Conference on Learning Theory}, pages 879--906, 2016.

\bibitem[Edm65]{Edmonds65}
Jack Edmonds.
\newblock Paths, trees, and flowers.
\newblock {\em Canadian Journal of mathematics}, 17(3):449--467, 1965.
\newblock Available at: https://cms.math.ca/10.4153/CJM-1965-045-4.

\bibitem[EK72]{EdmondsK72}
Jack Edmonds and Richard~M. Karp.
\newblock Theoretical improvements in algorithmic efficiency for network flow
  problems.
\newblock {\em J. {ACM}}, 19(2):248--264, 1972.

\bibitem[ET75]{EvenT75}
Shimon Even and Robert~Endre Tarjan.
\newblock Network flow and testing graph connectivity.
\newblock {\em {SIAM} J. Comput.}, 4(4):507--518, 1975.

\bibitem[GKK{\etalchar{+}}15]{GhaffariKKLP15}
Mohsen Ghaffari, Andreas Karrenbauer, Fabian Kuhn, Christoph Lenzen, and Boaz
  Patt{-}Shamir.
\newblock Near-optimal distributed maximum flow: Extended abstract.
\newblock In {\em Proceedings of the 2015 {ACM} Symposium on Principles of
  Distributed Computing, {PODC} 2015, Donostia-San Sebasti{\'{a}}n, Spain, July
  21 - 23, 2015}, pages 81--90, 2015.
\newblock Available at: https://arxiv.org/abs/1508.04747.

\bibitem[GN79]{GalilN79}
Zvi Galil and Amnon Naamad.
\newblock Network flow and generalized path compression.
\newblock In {\em Proceedings of the 11h Annual {ACM} Symposium on Theory of
  Computing, April 30 - May 2, 1979, Atlanta, Georgia, {USA}}, pages 13--26,
  1979.

\bibitem[GR98]{GoldbergR98}
Andrew~V. Goldberg and Satish Rao.
\newblock Beyond the flow decomposition barrier.
\newblock {\em J. {ACM}}, 45(5):783--797, 1998.

\bibitem[GT88]{GoldbergT88}
Andrew~V. Goldberg and Robert~Endre Tarjan.
\newblock A new approach to the maximum-flow problem.
\newblock {\em J. {ACM}}, 35(4):921--940, 1988.

\bibitem[GT14]{GoldbergT14}
Andrew~V. Goldberg and Robert~Endre Tarjan.
\newblock Efficient maximum flow algorithms.
\newblock {\em Commun. {ACM}}, 57(8):82--89, 2014.

\bibitem[HK73]{HopcroftK73}
John~E. Hopcroft and Richard~M. Karp.
\newblock An $n^{5/2}$ algorithm for maximum matchings in bipartite graphs.
\newblock {\em {SIAM} J. Comput.}, 2(4):225--231, 1973.

\bibitem[HO13]{HochbaumO13}
Dorit~S. Hochbaum and James~B. Orlin.
\newblock Simplifications and speedups of the pseudoflow algorithm.
\newblock {\em Networks}, 61(1):40--57, 2013.

\bibitem[Hoc08]{Hochbaum08}
Dorit~S Hochbaum.
\newblock The pseudoflow algorithm: A new algorithm for the maximum-flow
  problem.
\newblock {\em Operations research}, 56(4):992--1009, 2008.

\bibitem[Kar73]{Karzanov73}
Alexander~V. Karzanov.
\newblock O nakhozhdenii maksimal\' nogo potoka v setyakh spetsial\'nogo vida i
  nekotorykh prilozheniyakh.
\newblock {\em Matematicheskie Voprosy Upravleniya Proizvodstvom}, 5:81--94,
  1973.
\newblock In Russian, title translation: on finding maximum flows in networks
  with special structure and some applications.

\bibitem[KBR07]{KolmogorovBR07}
Vladimir Kolmogorov, Yuri Boykov, and Carsten Rother.
\newblock Applications of parametric maxflow in computer vision.
\newblock In {\em Computer Vision, 2007. ICCV 2007. IEEE 11th International
  Conference on}, pages 1--8. IEEE, 2007.

\bibitem[KLOS14]{KelnerLOS14}
Jonathan~A. Kelner, Yin~Tat Lee, Lorenzo Orecchia, and Aaron Sidford.
\newblock An almost-linear-time algorithm for approximate max flow in
  undirected graphs, and its multicommodity generalizations.
\newblock In {\em Proceedings of the Twenty-Fifth Annual {ACM-SIAM} Symposium
  on Discrete Algorithms, {SODA} 2014, Portland, Oregon, USA, January 5-7,
  2014}, pages 217--226, 2014.
\newblock Available at http://arxiv.org/abs/1304.2338.

\bibitem[KM09]{KelnerM09}
J.A. Kelner and A.~Madry.
\newblock {Faster generation of random spanning trees}.
\newblock In {\em FOCS}, 2009.

\bibitem[KM11]{KelnerM11}
Jonathan~A. Kelner and Petar Maymounkov.
\newblock Electric routing and concurrent flow cutting.
\newblock {\em Theor. Comput. Sci.}, 412(32):4123--4135, 2011.
\newblock Available at: https://arxiv.org/abs/0909.2859.

\bibitem[KMP11]{KoutisMP11}
Ioannis Koutis, Gary~L. Miller, and Richard Peng.
\newblock A nearly-m log n time solver for {SDD} linear systems.
\newblock In {\em Proceedings of the 2011 IEEE 52nd Annual Symposium on
  Foundations of Computer Science}, FOCS '11, pages 590--598, Washington, DC,
  USA, 2011. IEEE Computer Society.
\newblock Available at http://arxiv.org/abs/1102.4842.

\bibitem[KMP12]{KoutisMP12}
Ioannis Koutis, Gary~L. Miller, and Richard Peng.
\newblock A fast solver for a class of linear systems.
\newblock {\em Communications of the ACM}, 55(10):99--107, October 2012.
\newblock Available at
  https://cacm.acm.org/magazines/2012/10/155538-a-fast-solver-for-a-class-of-linear-systems/fulltext.

\bibitem[KMP14]{KoutisMP10}
I.~Koutis, G.~Miller, and R.~Peng.
\newblock Approaching optimality for solving sdd linear systems.
\newblock {\em SIAM Journal on Computing}, 43(1):337--354, 2014.
\newblock Available at http://arxiv.org/abs/1003.2958.

\bibitem[KOSZ13]{KelnerOSZ13}
J.~A. Kelner, L.~Orecchia, A.~Sidford, and Z.~A. Zhu.
\newblock A simple, combinatorial algorithm for solving sdd systems in
  nearly-linear time.
\newblock In {\em STOC}, 2013.

\bibitem[KRSS15]{KyngRSS15}
R.~Kyng, A.~B. Rao, S.~Sachdeva, and D.~A Spielman.
\newblock Algorithms for lipschitz learning on graphs.
\newblock In {\em COLT}, 2015.

\bibitem[KS96]{KargerS96}
David~R. Karger and Clifford Stein.
\newblock A new approach to the minimum cut problem.
\newblock {\em J. {ACM}}, 43(4):601--640, 1996.

\bibitem[KS16]{KyngS16}
Rasmus Kyng and Sushant Sachdeva.
\newblock Approximate gaussian elimination for laplacians - fast, sparse, and
  simple.
\newblock In {\em {FOCS}}, pages 573--582. {IEEE} Computer Society, 2016.
\newblock Available at http://arxiv.org/abs/1605.02353.

\bibitem[LPS15]{LeePengSpielman}
Y.~T. Lee, R.~Peng, and D.~A. Spielman.
\newblock Sparsified cholesky solvers for {SDD} linear systems.
\newblock {\em CoRR}, abs/1506.08204, 2015.

\bibitem[LS13]{LeeS13}
Yin~Tat Lee and Aaron Sidford.
\newblock Efficient accelerated coordinate descent methods and faster
  algorithms for solving linear systems.
\newblock In {\em Proceedings of the 2013 IEEE 54th Annual Symposium on
  Foundations of Computer Science}, FOCS '13, pages 147--156, Washington, DC,
  USA, 2013. IEEE Computer Society.

\bibitem[LS14]{LeeS14}
Y.~T. Lee and A.~Sidford.
\newblock Path finding methods for linear programming: Solving linear programs
  in $\tilde{O}(\textrm{vrank})$ iterations and faster algorithms for maximum
  flow.
\newblock In {\em FOCS}, 2014.

\bibitem[LSBG13]{LiSBG13}
Bingdong Li, Jeff Springer, George Bebis, and Mehmet~Hadi Gunes.
\newblock A survey of network flow applications.
\newblock {\em Journal of Network and Computer Applications}, 36(2):567--581,
  2013.

\bibitem[Mad10]{Madry10}
Aleksander Madry.
\newblock Fast approximation algorithms for cut-based problems in undirected
  graphs.
\newblock In {\em Foundations of Computer Science (FOCS), 2010 51st Annual IEEE
  Symposium on}, pages 245--254. IEEE, 2010.
\newblock Available at http://arxiv.org/abs/1008.1975.

\bibitem[Mad11]{Madry11:thesis}
Aleksander Madry.
\newblock {\em From graphs to matrices, and back: new techniques for graph
  algorithms}.
\newblock PhD thesis, Massachusetts Institute of Technology, 2011.

\bibitem[Mad13]{Madry13}
A.~Madry.
\newblock Navigating central path with electrical flows: From flows to
  matchings, and back.
\newblock In {\em FOCS}, 2013.

\bibitem[Mad16]{Madry16}
Aleksander Madry.
\newblock Computing maximum flow with augmenting electrical flows.
\newblock In {\em {IEEE} 57th Annual Symposium on Foundations of Computer
  Science, {FOCS} 2016, 9-11 October 2016, Hyatt Regency, New Brunswick, New
  Jersey, {USA}}, pages 593--602, 2016.
\newblock Available at: https://arxiv.org/abs/1608.06016.

\bibitem[NN94]{NesterovN94}
Y.~Nesterov and A.~Nemirovskii.
\newblock {\em Interior-Point Polynomial Algorithms in Convex Programming}.
\newblock Society for Industrial and Applied Mathematics, 1994.

\bibitem[Orl13]{Orlin13}
James~B. Orlin.
\newblock Max flows in $o(nm)$ time, or better.
\newblock In {\em Symposium on Theory of Computing Conference, STOC'13, Palo
  Alto, CA, USA, June 1-4, 2013}, pages 765--774, 2013.

\bibitem[Pen16]{Peng16}
Richard Peng.
\newblock Approximate undirected maximum flows in ${O}(m \textrm{ polylog}(n))$
  time.
\newblock In {\em Proceedings of the Twenty-Seventh Annual ACM-SIAM Symposium
  on Discrete Algorithms}, pages 1862--1867. SIAM, 2016.
\newblock Available at http://arxiv.org/abs/1411.7631.

\bibitem[PS14]{PengS14}
Richard Peng and Daniel~A. Spielman.
\newblock An efficient parallel solver for {SDD} linear systems.
\newblock In {\em Proceedings of the 46th Annual ACM Symposium on Theory of
  Computing}, STOC '14, pages 333--342, New York, NY, USA, 2014. ACM.
\newblock Available at http://arxiv.org/abs/1311.3286.

\bibitem[PZZ13]{PengZZ13}
Bo~Peng, Lei Zhang, and David Zhang.
\newblock A survey of graph theoretical approaches to image segmentation.
\newblock {\em Pattern Recognition}, 46(3):1020--1038, 2013.

\bibitem[ROF92]{RudinOF92}
Leonid~I Rudin, Stanley Osher, and Emad Fatemi.
\newblock Nonlinear total variation based noise removal algorithms.
\newblock {\em Physica D: nonlinear phenomena}, 60(1-4):259--268, 1992.

\bibitem[Sac19]{Sachdeva:communication}
Sushant Sachdeva.
\newblock Private Communication, 2019.

\bibitem[Sch02]{Schrijver02}
Alexander Schrijver.
\newblock On the history of the transportation and maximum flow problems.
\newblock {\em Math. Program.}, 91(3):437--445, 2002.

\bibitem[She13]{Sherman13}
Jonah Sherman.
\newblock Nearly maximum flows in nearly linear time.
\newblock In {\em 54th Annual {IEEE} Symposium on Foundations of Computer
  Science, {FOCS} 2013, 26-29 October, 2013, Berkeley, CA, {USA}}, pages
  263--269, 2013.
\newblock Available at http://arxiv.org/abs/1304.2077.

\bibitem[She17a]{Sherman17b}
Jonah Sherman.
\newblock Area-convexity, l\({}_{\mbox{{\(\infty\)}}}\) regularization, and
  undirected multicommodity flow.
\newblock In {\em Proceedings of the 49th Annual {ACM} {SIGACT} Symposium on
  Theory of Computing, {STOC} 2017, Montreal, QC, Canada, June 19-23, 2017},
  pages 452--460, 2017.

\bibitem[She17b]{Sherman17a}
Jonah Sherman.
\newblock Generalized preconditioning and undirected minimum-cost flow.
\newblock In {\em Proceedings of the Twenty-Eighth Annual ACM-SIAM Symposium on
  Discrete Algorithms}, SODA '17, pages 772--780, 2017.
\newblock Available at: https://arxiv.org/abs/1606.07425.

\bibitem[Spi18]{SpielmanCheegerLect18}
D.~A. Spielman.
\newblock {Conductance, the Normalized Laplacian, and Cheeger’s Inequality}.
\newblock \url{http://www.cs.yale.edu/homes/spielman/561/lect11-18.pdf}, 2018.

\bibitem[SS11]{SpielmanS08:journal}
D.~Spielman and N.~Srivastava.
\newblock Graph sparsification by effective resistances.
\newblock {\em SIAM Journal on Computing}, 40(6):1913--1926, 2011.
\newblock Available at http://arxiv.org/abs/0803.0929.

\bibitem[ST83]{SleatorT83}
Daniel~D Sleator and Robert~Endre Tarjan.
\newblock A data structure for dynamic trees.
\newblock {\em Journal of computer and system sciences}, 26(3):362--391, 1983.
\newblock Announced at STOC'81.

\bibitem[ST85]{SleatorT85}
Daniel~Dominic Sleator and Robert~Endre Tarjan.
\newblock Self-adjusting binary search trees.
\newblock {\em J. {ACM}}, 32(3):652--686, 1985.

\bibitem[ST14]{SpielmanTengSolver:journal}
D.~Spielman and S.~Teng.
\newblock Nearly linear time algorithms for preconditioning and solving
  symmetric, diagonally dominant linear systems.
\newblock {\em SIAM Journal on Matrix Analysis and Applications},
  35(3):835--885, 2014.
\newblock Available at http://arxiv.org/abs/cs/0607105.

\bibitem[SW18]{SaranurakW18:prelim}
Thatchaphol Saranurak and Di~Wang.
\newblock Expander decomposition and pruning: Faster, stronger, and simpler,
  2018.
\newblock To appear at SODA 2019.
  \url{https://dw236.github.io/papers/main_decomp.pdf}.

\bibitem[Tro12]{Tropp12}
Joel~A. Tropp.
\newblock User-friendly tail bounds for sums of random matrices.
\newblock {\em Found. Comput. Math.}, 12(4):389--434, August 2012.
\newblock Available at http://arxiv.org/abs/1004.4389.

\bibitem[ZWC10]{ZhuWC10}
Mingqiang Zhu, Stephen~J Wright, and Tony~F Chan.
\newblock Duality-based algorithms for total-variation-regularized image
  restoration.
\newblock {\em Computational Optimization and Applications}, 47(3):377--400,
  2010.

\end{thebibliography}

\appendix
\section{Deferred Proofs from Prelims, Section~\ref{sec:overview}}
\identity*
\begin{proof}
  Consider the map $\map{\calG}{\calG}$ such that for every flow
  $\ff^{\calG}$ on $\calG,$ we have $\map{\calG}{\calG}(\ff^{\calG}) =
  \ff^{\calG}.$ Thus,
  \begin{align*}
    \obj^{\calG}\left( {\kappa}^{-1}
    \map{\calG}{\calG}(\ff^{\calG}) \right)
    & = \obj^{\calG} \left( {\kappa}^{-1} \ff^{\calG} \right) \\
    & = \left( \gg^{\calG}  \right)^{\top} \left( {\kappa}^{-1}
      \ff^{\calG} \right) - h_p(\rr, \kappa^{-1} \ff^{\calG}) \\
    & \ge \kappa^{-1}  \left( \gg^{\calG}  \right)^{\top} \ff^{\calG}
      - \kappa^{-2} h_p( \rr, \ff^{\calG})
    && \text{(Using
       Lemma \ref{lem:iterative-refinement:rescaling})} \\
    & \ge \kappa^{-1}  \left( \gg^{\calG}  \right)^{\top} \ff^{\calG}
      - \kappa^{-1} h_p( \rr, \ff^{\calG}) = \kappa^{-1} \obj^{\calG}(\ff^{\calG}).
  \end{align*}
  Moreover $(\BB^{\calG})^{\top} \map{\calG}{\calG}(\ff^{\calG}) = \BB^{\calG}
  \ff^{\calG}.$ Thus, the claims follow.
\end{proof}

\composition*
\begin{proof}
  It is easy to observe that the given mapping is linear.  Given a
  flow $\ff^{\calG_1}$ on $\calG_1,$ we have,
  \begin{align*}    
    (\BB^{\calG_3})^{\top} \left( \map{\calG_2}{\calG_3} \circ
    \map{\calG_1}{\calG_2} (\ff^{\calG_1}) \right)
    & =
      (\BB^{\calG_3})^{\top} \left( \map{\calG_2}{\calG_1} \left(
      \map{\calG_3}{\calG_2} (\ff^{\calG_1}) \right) \right)  \\
    & =
      (\BB^{\calG_2})^{\top} \left( \map{\calG_3}{\calG_2} (\ff^{\calG_1})
      \right)
      = (\BB^{\calG_1})^{\top} (\ff^{\calG_1}).
  \end{align*}
  Moreover,
  \begin{align*}
    \obj_{\calG_3} \left( (\kappa_1 \kappa_2)^{-1} 
    \map{\calG_2}{\calG_3} \left( \map{\calG_1}{\calG_2}
    (\ff^{\calG_1}) \right) \right)
    & \ge  \obj_{\calG_3} \left( {\kappa_2}^{-1}
      \map{\calG_2}{\calG_3} \left( {\kappa_1}^{-1} \map{\calG_1}{\calG_2}
      (\ff^{\calG_1}) \right) \right)
    && \text{(Using linearity)} \\
    & \ge  {\kappa_2}^{-1} \obj_{\calG_2} \left( {\kappa_1}^{-1}
      \map{\calG_1} {\calG_2}
      (\ff^{\calG_1}) \right)
    && \text{(Using $\calG_2 \preceq_{\kappa_2} \calG_3$)} \\
    & \ge  (\kappa_2\kappa_1)^{-1} \obj_{\calG_1} \left(
      \ff^{\calG_1} \right)
    && \text{(Using $\calG_1 \preceq_{\kappa_1} \calG_2$)}
  \end{align*}
  The same proof works for $\circapprox.$
\end{proof}

\union*
\begin{proof}
  Let $\ff^{\calH}$ be a flow on $\calH.$ We write $\ff^{\calH} =
  (\ff^{\calH_1}, \ff^{\calH_2}).$ Let $\ff^{\calG} \defeq
  \map{\calH}{\calG}(\ff^{\calH}).$ If $\ff^{\calG_i}$ denotes
  $\map{\calH_i}{\calG_i}(\ff^{\calH_i})$ for $i=1,2,$ then we know
  that $\ff^{\calG} = (\ff^{\calG_1}, \ff^{\calG_2}).$ Thus, the
  objectives satisfy
  \begin{align*}
    \obj^{\calG}(\kappa^{-1} \ff^{\calG})
    & = \obj^{\calG_1}(\kappa^{-1} \ff^{\calG_1}) +
      \obj^{\calG_2}(\kappa^{-1} \ff^{\calG_2}) \\
    & \ge 
      \kappa^{-1} \obj^{\calH_1}(\ff^{\calH_1}) +
      \kappa^{-1} \obj^{\calH_2}(\ff^{\calH_2})  
      = \kappa^{-1} \obj^{\calH}(\ff^{\calH})
  \end{align*}
  For the residues, we have,
  \begin{align*}
    (\BB^{\calG})^{\top}(\ff^{\calG})
    & = (\BB^{\calG_1})^{\top}(\ff^{\calG_1}) +
      (\BB^{\calG_2})^{\top}(\ff^{\calG_2}) \\
    & = (\BB^{\calH_1})^{\top}(\ff^{\calH_1}) +
      (\BB^{\calH_2})^{\top}(\ff^{\calH_2}) = (\BB^{\calH})^{\top}(\ff^{\calH}).
  \end{align*}
  Thus, $\calH \preceq_{\kappa} \calG.$
\end{proof}

\PerturbResistances*
\begin{proof}
  Consider the map $\map{\calH}{\calG}(\ff) = \ff.$ Thus, since the
  underlying graphs are the same, we immediately have
  $(\BB^{\calG})^{\top} \ff = (\BB^{\calH})^{\top} \ff.$ For the
  objective, we have
  \begin{align*}
    \obj^{\calG}(\kappa^{-1} \ff)
    & = \sum_{e} \left( \kappa^{-1} \gg^{\calG}_e \ff_e -
      \kappa^{-2} \rr^{\calG}_e \ff_e^{2} - \kappa^{-p} s^{\calG} \abs{\ff_e}^{p} \right) \\
    & \ge \kappa^{-1} \sum_{e} \left( \gg^{\calH}_e \ff_e -
      \rr^{\calH}_e \ff_e^{2} - s^{\calH} \abs{\ff_e}^{p} \right) =
      \kappa^{-1} \obj^{\calH}(\ff).
  \end{align*}
\end{proof}


\section{Deferred Proofs for Numerical Methods from 
  Section~\ref{sec:numerical}}
\label{sec:numerical-proofs}
The following simple lemma characterizes the change in smoothed
$\ell_p$-norms under rescaling of the input vector.
\begin{restatable}{lemma}{Rescaling}
  \label{lem:iterative-refinement:rescaling}
  For all
  $\xx \in \rea^{m}, \rr \in \rea^{m}_{\ge 0},s \in \rea_{\ge 0},$ and
  $\lambda \in \rea$, we have,
  \begin{equation*}
    \min\{\abs{\lambda}^2,\abs{\lambda}^p\} h_p(\rr, s, \xx)  \leq
    h_p(\rr, s, \lambda \xx) \leq \max\{\abs{\lambda}^2,
    \abs{\lambda}^p\} h_p(\rr , s, \xx).
  \end{equation*}
\end{restatable}
\begin{proof}
  It suffices to prove the claim for $x \in \rea, r \in \rea_{\ge 0},
  s \in \rea_{\ge 0}$
  \begin{align*}
    h_p(r, s, \lambda x)
    & = r(\lambda x)^{2} + s \abs{\lambda x}^{p} \\
    & = \abs{\lambda}^{2} \cdot r x^{2} + \abs{\lambda}^{p} \cdot 
      s \abs{x}^{p}
  \end{align*}
  Since all terms are non-negative, we get,
  \begin{align*}
    h_p(r, s, \lambda x )
    & \ge \min\{\abs{\lambda}^{2},
      \abs{\lambda}^{p}\} \cdot (r x^{2} + s \abs{x}^{p}),  \\
    \text{and} \quad h_p(r, s, \lambda x )
    & \le \max\{\abs{\lambda}^{2},
      \abs{\lambda}^{p}\} \cdot (r x^{2} + s \abs{x}^{p}).
  \end{align*}
\end{proof}

\RestateIterativeRefinementApprox*
\begin{proof}
  Note that all the terms are a sum over the coordinates. Thus, it
  suffices to prove the inequality for $x, \delta \in \rea,$ and $r, s
  \in \rea_{\ge 0}.$
  We have,
  \begin{align*}
    h_p(r, s, x+\delta) - h_p(r, s, x) - \delta \frac{\partial}{\partial
    x}h_{p}(r, s, x)
    & = r (x+\delta)^2 + s |x+\delta|^p - rx^2 - s|x|^p - \delta (2rx +
      ps|x|^{p-2}x) \\
    & = r \delta^2 + s \abs{x+\delta}^{p} - s |x|^{p} - ps \delta |x|^{p-2}x
    \\
    & = r \delta^2 + s \abs{x}^{p}\left( \abs{1+\delta'}^{p} - 1 -
      p\delta' \right),
  \end{align*}
  where $\delta' = \delta / x.$

  Lemma~\ref{lem:iterative-refinement:p-norm-upper-bound}, proved
  later, proves that for all $\delta',$ and $p \ge 2,$ we have,
  \[ \abs{1+\delta'}^{p} - 1 - p\delta' \le {p2^{p-1}} \left( \delta'^2
      + |\delta'|^{p} \right).\]
  Thus,
  \begin{align*}
    h_p(r, s, x+\delta) - h_p(r, s, x) - \delta \frac{\partial}{\partial
    x}h_{p}(r, s, x)
    & \le  r \delta^2 + s \abs{x}^{p}  {p2^{p-1}} \left(  \delta'^2
      + |\delta'|^{p}  \right) \\
    & =  r \delta^2 + s p2^{p-1}\abs{x}^{p-2}\delta^{2} + s p2^{p-1}
      \abs{\delta}^{p} \\
    & \le p2^{p-1} ((r + s \abs{x}^{p-2}) \delta^{2} + s \abs{\delta}^{p})
    \\
    & = p2^{p-1} h_p(r + s \abs{x}^{p-2}, s, \delta) \\
    & \le 2^{2p} h_p(r + s \abs{x}^{p-2}, s, \delta).
  \end{align*}

  Lemma~\ref{lem:iterative-refinement:p-norm-lower-bound}, proved
  later, shows that for all $\delta',$ and $p \ge 2,$ we have,
  \[ \abs{1+\delta'}^{p} - 1 - p\delta' \ge {2^{-p}} \left( \delta'^2
      + |\delta'|^{p} \right).\]
  \begin{align*}
    h_p(r, s, x+\delta) - h_p(r, s, x) - \delta \frac{\partial}{\partial
    x}h_{p}(r, s, x)
    & \ge  r \delta^2 + s \abs{x}^{p}  {2^{-p}} \left(  \delta'^2
      + |\delta'|^{p}  \right) \\
    & =  r \delta^2 + 2^{-p}s \abs{x}^{p-2}\delta^{2} + 2^{-p}s
      \abs{\delta}^{p} \\
    & \ge 2^{-p} ((r + s \abs{x}^{p-2}) \delta^{2} + s \abs{\delta}^{p})
    \\
    & = 2^{-p} h_p(r + s \abs{x}^{p-2}, s, \delta).
  \end{align*}

\end{proof}

\begin{lemma}
  \label{lem:iterative-refinement:p-norm-upper-bound}
  For all $\delta \in \rea,$ $p \ge 1,$ we have,
  \[ \abs{1+\delta}^{p} - 1 - p\delta
    \le {p2^{p-1}} \left(  \delta^2
      + |\delta|^{p}  \right).
  \]
\end{lemma}
\begin{proof}
  The proof has to consider several cases.
  \paragraph{$\boldsymbol{\delta \ge 1}$.}
  Using mean-value theorem, we know there is some $z \in [0,\delta]$
  such that
  \begin{align*}
    \abs{1+\delta}^{p} - 1 - p\delta
    & = (1+\delta)^{p} - 1 - p\delta
    \\ 
    & = p\delta\left( (1+z)^{p-1} - 1 \right) \\
    & \le p\delta(1+\delta)^{p-1}  \\
    & \le p \delta (2\delta)^{p-1}.
  \end{align*}

  \paragraph{$\boldsymbol{0 \le \delta \le 1}$.}
  Using mean-value theorem, we know there is some
  $z \in [0,\delta]$ such that
  \begin{align*}
    \abs{1+\delta}^{p} - 1 - p\delta
    & = (1+\delta)^{p} - 1 - p\delta
    \\
    & = p\delta\left( (1+z)^{p-1} - 1 \right) \\
    & \le p\delta\left( (1+\delta)^{p-1} - 1 \right).
  \end{align*}
  If $p \le 2,$ we have $(1+\delta)^{p-1}$ is a concave function, and
  hence $(1+\delta)^{p-1} \le 1 + (p-1)\delta.$ If $p \ge 2,$  we have
  $(1+\delta)^{p-1}$ is a convex function, and hence for $\delta \in
  [0,1],$ we have $(1+\delta)^{p-1} \le 1 + (2^{p-1} - 1)\delta.$
  Thus,
  \[ \abs{1+\delta}^{p} - 1 - p\delta \le p\max\{(p-1)\delta^2,
    (2^{p-1}-1) \delta^2 \}.
  \]

  \paragraph{$\boldsymbol{-1 \le \delta \le 0}$.}
  Using mean-value theorem, we know there is some
  $z \in [-\abs{\delta},0]$ such that
  \begin{align*}
    \abs{1+\delta}^{p} - 1 - p\delta
    & = (1+\delta)^{p} - 1 - p\delta
    \\
    & = p\delta\left( (1+z)^{p-1} - 1 \right) \\
    & \le p\abs{\delta}\left( 1 - (1+\delta)^{p-1} \right).
  \end{align*}
  If $p \le 2,$ we have $(1+\delta)^{p-1}$ is a concave function, and
  hence for $\delta \in [-1,1],$ we have
  $(1+\delta)^{p-1} \ge 1 + \delta.$ If $p \ge 2,$ we have
  $(1+\delta)^{p-1}$ is a convex function, and hence
  $(1+\delta)^{p-1} \ge 1 + (p-1)\delta.$ Thus,
  \[ \abs{1+\delta}^{p} - 1 - p\delta \le p\abs{\delta}
    \max\{\abs{\delta}, (p-1) \abs{\delta} \}.\]
 
  \paragraph{$\boldsymbol{\delta \le -1}$.}
  We have,
  \begin{align*}
    \abs{1+\delta}^{p} - 1 - p\delta
    & = (\abs{\delta}-1)^{p} - 1 + p\abs{\delta}
    \\
    & \le \abs{\delta}^p + p\abs{\delta} \\
    & \le \abs{\delta}^{p} + p \abs{\delta}^{2},  \end{align*}
  since $\abs{\delta} \ge 1.$
\end{proof}

\begin{lemma}
  \label{lem:iterative-refinement:p-norm-lower-bound}
  For all $\delta \in \rea, p \ge 2,$ we have
  \[\abs{1+\delta}^{p} - 1 - p\delta \ge 2^{-p}(\delta^2 + \abs{\delta}^{p}).\]
\end{lemma}
\begin{proof}
  Let $h(\delta)$ denote the function
  \[h(\delta)= \abs{1+\delta}^{p} - 1 - p\delta - 2^{-p}(\delta^2 +
    \abs{\delta}^{p}).\] Thus, $h(0)=0.$ As for the previous proof, we consider
  several cases:
  \paragraph{$\boldsymbol{\delta \ge 0}$.} We have $h(\delta) = (1+\delta)^p - 1 -
  p\delta - 2^{-p}(\delta^2 +\delta^{p}).$ Thus,
  \begin{align*}
    h'(\delta)
    & = p(1+\delta)^{p-1} -p -2^{-p+1}\delta
      -p2^{-p}\delta^{p-1} \\
    h{''}(\delta) & = p(p-1)(1+\delta)^{p-2} - 2^{-p+1} - p(p-1)2^{-p}\delta^{p-2}
  \end{align*}
  Observe that since $p \ge 2,$ we have
  $(1+\delta)^{p-2} \ge \max\{1,\delta^{p-2}\} \ge
  2^{-1}(1+\delta^{p-2}),$ and $p(p-1) \ge 2.$ Thus,
  \begin{align*}
    h''(\delta) & \ge 2^{-1}p(p-1) + 2^{-1}p(p-1)\delta^{p-2} -
                  2^{-p+1} - p(p-1)2^{-p}\delta^{p-2}  \ge 0.
  \end{align*}
  Since $h(0) = h'(0) = 0,$ and $h''(\delta) \ge 0,$ for all
  $\delta \ge 0,$ we must have $h(\delta) \ge 0$ for all
  $\delta \ge 0.$

  \paragraph{$\boldsymbol{-1 \le \delta \le 0}$.} We have,
  \begin{align*}
    h(\delta)
    & = (1+\delta)^p - 1 -
      p\delta - 2^{-p}\delta^2 - 2^{-p} \abs{\delta}^{p} \\
    h'(\delta)
    & = p(1+\delta)^{p-1} -p -2^{-p+1}\delta
      + p2^{-p}\abs{\delta}^{p-1}.    
  \end{align*}
  Since $0 \le 1+\delta \le 1,$ and $p-1 \ge 1,$ we have
  $(1+\delta)^{p-1} \le 1+\delta,$ and
  $\abs{\delta}^{p-1} \le \abs{\delta}^{1} = -\delta.$ Thus,
  \begin{align*}
    h'(\delta)
    & \le p(1+\delta) - p -2^{-p+1}\delta
      + p2^{-p}\abs{\delta} \\
    & = -p\abs{\delta} + (2+p)2^{-p}\abs{\delta} \\
    & \le -p\abs{\delta} + 2^{-2}(p+2) \abs{\delta} \le 0.
  \end{align*}

  \paragraph{$\boldsymbol{\delta \le -1}$.} We have,
  \begin{align*}
    h(\delta)
    & = (- 1 - \delta)^p - 1 -
      p\delta - 2^{-p}\delta^2 - 2^{-p} \abs{\delta}^{p} \\
    h'(\delta)
    & = -p(- 1-\delta)^{p-1} -p -2^{-p+1}\delta
      + p2^{-p}\abs{\delta}^{p-1}.    
  \end{align*}
  Since $\abs{\delta} \ge 1,$ and $p-1 \ge 1,$ we have $-\delta =
  \abs{\delta} \le \abs{\delta}^{p-1}.$ Thus,
  \begin{align*}
    h'(\delta)
    & \le  -p(- 1-\delta)^{p-1} -p + 2^{-p}(2+p)\abs{\delta}^{p-1}  \\
    & \le -p\left( (-1-\delta)^{p-1} + 1 - 2^{-p+1}\abs{\delta}^{p-1} \right).
  \end{align*}
  Now, observe that since $0 \le -1 -\delta,$ and $p-1 \ge 1,$
  \[ \abs{\delta}^{p-1} = (-\delta)^{p-1} = (1 + (-1 - \delta))^{p-1}
    \le 2^{p-1}(1^{p-1} + (-1-\delta)^{p-1}), \]
  Thus,
  \[(-1-\delta)^{p-1} + 1 - 2^{-p+1} \abs{\delta}^{p-1} \ge 0,\]
  and hence $h'(\delta) \le 0$ for $\delta \le -1.$

  \medskip \medskip
  \noindent For the last two cases, since $h(0)=0,$ and
  $h'(\delta) \le 0,$ for all $\delta \le 0.$ Thus, we must have
  $h(\delta) \ge 0,$ for $\delta \le 0.$
\end{proof}

\Residual*
\begin{proof}
  Let $\xx^{\star}$ to be an optimal solution to
  Problem~\ref{eq:iterative-refinement:original}. Consider
  $\widetilde{\ddelta} = \xx^{\star} - \xx_0.$ Thus,
  \[\AA\widetilde{\ddelta} = \AA\xx^{\star} - \AA\xx_0 = \bb - \bb =
    0.\]
  Thus, $\widetilde{\ddelta}$ is a feasible solution to
  Problem~\ref{eq:iterative-refinement:residual}. Moreover, it
  satisfies,
  \begin{align*}
    \obj_2(\widetilde{\ddelta})
    & = (\widetilde{\ddelta})^{\top }2^{p}
      \left( \gg - \nabla_{\xx}h(\rr, s,  \xx)|_{\xx = \xx_0}
      \right)  - h_p( \rr + s \abs{\xx_0}^{p-2}, s, \widetilde{\ddelta}) \\
    & = 2^{p} \gg^{\top }\widetilde{\ddelta}
      - 2^{p}\left( 2^{-p} h_p( \rr + s \abs{\xx_0}^{p-2}, s, \widetilde{\ddelta})
      + (\widetilde{\ddelta})^{\top}  \nabla_{\xx}h(\rr, s,
      \xx)|_{\xx = \xx_0}   \right) \\
    & \ge 2^{p} \gg^{\top }\widetilde{\ddelta}
      - 2^{p} \left( h_p(\rr, s, \xx_0 + \widetilde{\ddelta}) -
      h_{p}(\rr, s, \xx_0) \right) && \textrm{(Using
                                   Lemma~\ref{lem:iterative-refinement:approximation})}
    \\
    & = 2^{p} \gg^{\top} (\xx^{\star} - \xx_0) - 2^{p} \left( h_p(\rr, s,
      \xx^{\star} ) - h_p(\rr, s, \xx_0) \right) \\
    & = 2^{p} (\obj_{1}(\xx^{\star}) - \obj_{1}(\xx_0)).
  \end{align*}

  Now, given a feasible solution $\delta$ to
  Problem~\ref{eq:iterative-refinement:residual}, we must have
  $\AA\ddelta = \vzero.$ Thus,
  $\AA\xx_1 = \AA\xx_{0} + 2^{-3p}\AA\ddelta = b,$ and $\xx_{1}$ is a
  feasible solution to
  Problem~\ref{eq:iterative-refinement:original}. Moreover,
  \begin{align*}
    \obj_{1}(\xx_1)
    & = \gg^{\top}( \xx_0 + {2^{-3p}}\ddelta) - h_p(\rr, s,  \xx_0 +
      {2^{-3p}}\ddelta) \\
    & \ge \gg^{\top} \xx_0 + 2^{-3p}\gg^{\top}\ddelta - h_p(\rr, s,
      \xx_0) - 2^{-3p}\ddelta^{\top} \nabla_{\xx}h_p(\rr, s, \xx)|_{\xx =
      \xx_0} - 2^{2p}h_p(\rr + s \abs{\xx_0}^{p-2}, s, 2^{-3p}\ddelta)\\
    & \qquad \qquad \qquad \qquad \qquad \qquad \qquad \qquad \qquad
      \qquad \qquad \qquad \qquad \ \qquad \qquad \text{(Using
      Lemma~\ref{lem:iterative-refinement:approximation})} \\
    & \ge \obj_{1}(\xx_0) + 2^{-4p}\ddelta^{\top} \gg' - 2^{2p}\cdot
      2^{-6p} h_p(\rr + s \abs{\xx_0}^{p-2}, s, \ddelta)
      \qquad \qquad \qquad \text{(Using Lemma~\ref{lem:iterative-refinement:rescaling})}\\
    & = \obj_{1}(\xx_0) + 2^{-4p}\obj_{2}(\ddelta). 
  \end{align*}
\end{proof}


\section{Elimination of Low-Degree Vertices, and Loops}
\label{sec:Elimination}
In this section, we that the instance $\calH$ returned by
\textsc{UltraSparsify} can be reduced to a smaller graph by
repeatedly eliminating vertices of degree at most 2.
This step is analogous to the partial Cholesky factorization in the
Laplacian solver of Spielman and Teng~\cite{SpielmanTengSolver:journal}. A
slight technical issue is that if we run into a cycle where at most
$1$ vertex on the cycle has edge(s) to the rest of the graph, the
elimination of the degree $2$ nodes on the cycle essentially becomes
an optimization problem on only the cycle edges that can be solved
independently from the rest of the graph.

\begin{algorithm}
\caption{Elimination of Degree $1$ and $2$ vertices and Self-loops}
\label{alg:elimination}
 \begin{algorithmic}[1]
 \Procedure{Eliminate}{$\calH$}
 \State Initiate $\calH'\leftarrow \calH$
 \Repeat
 \State For every edge with non-selfloop degree 1, remove the only
 non-selfloop edge incident on it
 \Until No vertex has non-selfloop degree 1
 \For{every maximal path with all internal nodes having non-selfloop
   degree 2}
 \State Replace such a path with a single edge in $\calH'$ with the
 end points as the end points of the path, and,
\begin{itemize}
\item resistance is the sum of the resistances of the edges on the path
\item gradient is the sum of the gradients of the edges on the path
\item $s$ the same as before
\item Flow on the new edge is mapped to a flow along the original path (or cycle) in $\calH$.
\end{itemize}
\EndFor
\State Move all self-loops from $\calH'$ to $\calH_{loop}$
\State 
\Return $\calH',\calH_{loop},\map{(\calH'+\calH_{loop})}{\calH}$
 \EndProcedure 
 \end{algorithmic}
\end{algorithm}

\elimination*

\begin{proof}
  We first observe that a self-loop $e \in E^{\calG}$ on a vertex
  $v \in V^{\calG}$ does not contribute to the residue at any vertex,
  including $v.$ Thus, the circulation constraint on a flow
  $\ff^{\calG}$ does not impose any constraint on $\ff^{\calG}_e.$
  Moreover, since the objective $\alpha^{\calG}$ can be written as a
  sum over the edges, for every self-loop $e,$ the variable
  $\ff^{\calG}_e$ is independent of all other variables. Thus, we can
  ignore the self-loops in remainder of the proof.
  
  We first prove that we can repeatedly eliminate vertices of
  non-selfloop degree 1 in $G$ while preserving $\obj$ exactly for a
  circulation. Consider one such vertex $v \in V^{\calG},$ and let
  $e = ({v,u}) \in E^{\calG}$ be the only non-selfloop edge incident
  on $v$ (the argument for the reverse direction is identical). Given
  any circulation $\ff^{\calG},$ since the only non-selfloop edge
  incident on $v$ is $e,$ we must have $\ff^{\calG}_{e} = 0.$ Thus, we
  can drop $e$ entirely from the instance. Formally, we define
  \[V^{\calG'} = V^{\calG}, \quad E^{\calG'} =
    E^{\calG}\setminus \{e\},\quad \gg^{\calG'} =
    \gg^{\calG}|_{E^{\calG'}}, \quad \rr^{\calG'} =
    \rr^{\calG}|_{E^{\calG'}}, \quad\text{and}\quad  s^{\calG'} = s^{\calG}.\] We let the mapping
  $\map{\calG}{\calG'}$ to be just the projection on to $E^{\calG'}.$
  Thus,
  $\map{\calG}{\calG'}(\ff^{\calG}) = \ff^{\calG}|_{E^{\calG'}}.$
  Since $\ff^{\calG}_e = 0,$ we immediately get
  $ \obj^{\calG'}(\map{\calG}{\calG'}(\ff^{\calG})) =
  \obj^{\calG}(\ff^{\calG}).$ Thus, $\calG \circapprox_{1} \calG'.$

  Now, consider the mapping $\map{\calG'}{\calG}$ that pads a
  circulation $\ff^{\calG'}$ on $\calG'$ with 0 on $e,$ \emph{i.e.},
  \[
    \left( \map{\calG'}{\calG}(\ff^{\calG'}) \right)_{e'} =
    \begin{cases}
      0 & \text{if } e'=e, \\
      \ff^{\calG'}_{e'} & \text{otherwise}.
    \end{cases}
  \]
  Again, it is immediate that
  $ \obj^{\calG}(\map{\calG'}{\calG}(\ff^{\calG'})) =
  \obj^{\calG'}(\ff^{\calG'}).$ Thus, $\calG' \circapprox_{1} \calG.$

  We can repeatedly apply the above transformation to eliminate all
  vertices of non-selfloop degree 1 in $G.$ For convenience, we let
  $\calG'$ denote the final instance obtained. Thus, we have,
  $\calG' \circapprox_{1} \calG \circapprox_{1} \calG'.$

  Now, we will replace maximal paths with all internal vertices of
  non-selfloop degree 2 with single edges. Consider such a path $P.$
  Formally, $P$ is a path of length $l$ in $G,$ say
  $P = (v_0, v_1, \ldots, v_{l-1}, v_{l}),$ with all of
  $v_1,\ldots,v_{l-1}$ having degree exactly 2, and $v_0,v_{l}$ have
  non-selfloop degree at least 3. For convenience, we assume that all
  edges $(v_{i-1}, v_{i})$ are oriented in the same direction. Observe
  that for a circulation $\ff^{\calG'},$ the flow on all the edges
  $(v_{i-1}, v_{i})$ must be the same, \emph{i.e.},
  $\ff^{\calG'}_{(v_{i-1},v_i)}$ must all be equal.
  Thus, we can replace $P$ with a single edge $e_P$ while
  preserving the amount of flow and the direction.

  Formally, let $\calP = \{P_1, \ldots, P_t\}$ denote the set of all
  maximal paths in $G'=(V^{\calG'}, E^{\calG'})$ such that all their
  internal vertices have non-selfloop degree exactly 2 in $G'.$ We
  replace each of these paths with a new edge connecting its
  endpoints. Let,
  \begin{align*}
    V^{\calG''}
    & = V^{\calG'}, \\
    E^{\calG''}
    & = E^{\calG'} \cup_{P \in \calP} \{e_p=(v_0,v_l) | P = (v_0,\ldots,v_l)
      \} \setminus
      \cup_{P \in \calP}\{e=(v_{i-1},v_{i}) \in P  | P =
      (v_0,\ldots,v_l)\}, \\
    \gg^{\calG''}_{e}
    & =
      \begin{cases}
        \gg^{\calG'}_{e} & \text{if } e \in E^{\calG'} \cap
        E^{\calG''}, \\
        \sum_{e' \in P} \gg^{\calG'}_{e'} & \text{if } e = e_P \text{
          for } P \in \calP,
      \end{cases} \\
    \rr^{\calG''}_{e}
    & =
      \begin{cases}
        \rr^{\calG'}_{e} & \text{if } e \in E^{\calG'} \cap
        E^{\calG''}, \\
        \sum_{e' \in P} \rr^{\calG'}_{e'} & \text{if } e = e_P \text{
          for } P \in \calP,
      \end{cases} \\
    s^{\calG''}
    & = s^{\calG}.
  \end{align*}

  We define the mapping $\map{\calG'}{\calG''}$ as follows
  \[
    \left( \map{\calG'}{\calG''}(\ff^{\calG'}) \right)_{e} =
    \begin{cases}
      \ff^{\calG'}_{e} & \text{if } e \in E^{\calG'} \cap E^{\calG''}, \\
      \ff^{\calG'}_{(v_0, v_1)} & \text{if } e = e_P, \textrm{ where }
      P=(v_0,\ldots,v_l), P \in \calP.
    \end{cases}
  \]
  We define $\map{\calG''}{\calG'}$ to be the inverse map of $\map{\calG'}{\calG''}.$
  \[
    \left( \map{\calG''}{\calG'}(\ff^{\calG''}) \right)_{e} =
    \begin{cases}
      \ff^{\calG''}_{e} & \text{if } e \in E^{\calG'} \cap E^{\calG''}, \\
      \ff^{\calG'}_{P} & \text{if } e = (v_{i-1}, v_{i}), i \in [l] , \textrm{ where }
      P=(v_0,\ldots,v_l) \in \calP.
    \end{cases}
  \]
  
  It follows from the definitions that for every circulation
  $\ff^{\calG'},$ letting $\ff^{\calG''}$ denote
  $\map{\calG'}{\calG''}(\ff^{\calG'}),$ we have,
  $\map{\calG''}{\calG'}(\ff^{\calG''}) = \ff^{\calG'}.$ Moreover,
  \begin{align*}
    \left( \gg^{\calG'} \right)^{\top} \ff^{\calG'}
    & = \left(
      \gg^{\calG'} \right)^{\top} \ff^{\calG''}  \\
    \sum_{e \in E^{\calG'}} \rr_e^{\calG'} \left( \ff^{\calG'}_e
    \right)^{2} 
    & =     \sum_{e \in E^{\calG''}} \rr_e^{\calG''} \left(
      \ff^{\calG''}_e
      \right)^{2} \\
    s^{\calG'} \sum_{e \in E^{\calG'}}  \abs{\ff^{\calG'}_e}^{p}
    & \ge     s^{\calG''} \sum_{e \in E^{\calG''}}  \abs{\ff^{\calG''}_e}^{p}
      \ge   \frac{1}{n}  s^{\calG'} \sum_{e \in E^{\calG'}}  \abs{\ff^{\calG'}_e}^{p},
  \end{align*}
  where the last inequality follows since for every path of length
  $l$, the contribution to the $\ell_p^{p}$ changes by a factor of
  $l^{-1},$ and since the paths must be vertex-disjoint,
  $l \le |V^{\calG}| \le n.$ The above inequalities imply,
  \[\obj^{\calG'}(\ff^{\calG'}) \le \obj^{\calG''}(\ff^{\calG''}),\]
  and hence $\calG' \circapprox_{1} \calG''.$ Combined with $\calG
  \circapprox_{1} \calG',$ we get $\calG \circapprox_{1} \calG''.$
  Moreover, we have, for $\kappa = n^{\frac{1}{p-1}},$
  \begin{align*}
    \obj^{\calG'} \left( \kappa^{-1} \ff^{\calG'} \right)
    & =     \left( \gg^{\calG'} \right)^{\top} \kappa^{-1} \ff^{\calG'}
      -  \kappa^{-2} \sum_{e \in E^{\calG'}} \rr_e^{\calG'} \left( \ff^{\calG'}_e
      \right)^{2}
      - \kappa^{-p}    s^{\calG'} \sum_{e \in E^{\calG'}}
      \abs{\ff^{\calG'}_e}^{p} \\
    & = \kappa^{-1} \left( \left( \gg^{\calG'} \right)^{\top} \ff^{\calG'}
      -  \kappa^{-1} \sum_{e \in E^{\calG'}} \rr_e^{\calG'} \left( \ff^{\calG'}_e
      \right)^{2}  - \frac{1}{n} s^{\calG'} \sum_{e \in E^{\calG'}}
      \abs{\ff^{\calG'}_e}^{p} \right) \\
    & \ge \kappa^{-1} \left( \left( \gg^{\calG''} \right)^{\top} \ff^{\calG''}
      -  \sum_{e \in E^{\calG''}} \rr_e^{\calG''} \left( \ff^{\calG''}_e
      \right)^{2}  - s^{\calG''} \sum_{e \in E^{\calG''}}
      \abs{\ff^{\calG''}_e}^{p} \right) = \kappa^{-1} \obj^{\calG''}(\ff^{\calG''}).
  \end{align*}
  Thus, $\calG'' \circapprox_{\kappa} \calG'.$ Combining with
  $\calG' \circapprox_{1} \calG,$ we obtain
  $\calG'' \circapprox_{\kappa} \calG.$ The final instance returned is
  $\calG'',$ giving us our theorem.
\end{proof}

\selfLoops*

\begin{proof}
Let $E'$ denote the set of all self-loops in $E^{\calG}.$ Then, we
define $\calG_1$ to be the instance obtained by removing all edges in
$E'.$ Formally,
\[\calG_1 \defeq (V^{\calG}, E^{\calG}\setminus E',
  \gg^{\calG}|_{E^{\calG}\setminus E'},
  \rr^{\calG}|_{E^{\calG}\setminus E'}, s^{\calG}).\]
We define $\calG_2$ be the instance $\calG$ restricted to $E'.$ Thus,
\[\calG_2 \defeq (V^{\calG}, E',
  \gg^{\calG}|_{E'}, \rr^{\calG}|_{E'}, s^{\calG}).\] It is immediate
that $\calG = \calG_1 \cup \calG_2.$ Since $\calG_2$ only has
self-loops, we have that for every $\ff^{\calG_2},$ we have
$(\BB^{\calG_2})^{\top}\ff^{\calG_2} = \vzero.$ Thus, the constraint
$(\BB^{\calG_2})^{\top}\ff^{\calG_2} = \vzero$ is vacuous.

Now, observe that in the absence of linear constraints on $\ff^{\calG_2},$ the
variables $\ff^{\calG_2}_e$ are independent for all $e \in
E^{\calG_2}.$ Moreover, we have
\[\obj^{\calG_2} (\ff^{\calG_2}) = \sum_{e \in E^{\calG_2}}
  \obj^{\calG_2}_{e}(\ff^{\calG_2}_e).\] Thus, we can solve for each
$\ff^{\calG_2}_e$ independently. Now, consider a fixed
$e \in E^{\calG_2}.$ We write $f_e$ for $\ff^{\calG_2}_{e}.$ We wish
to solve
\[\max_{f_e} \obj_{e}^{\calG_2}(f_e) = \max_{f_e} \gg^{\calG_2}_ef_e -
  \rr_e^{\calG_2} f_e^2 - s^{\calG_2} \abs{f_e}^{p}.\]
Note that the objective function is concave. The gradient of
$\obj^{\calG_2}_e(f_e)$ with respect to $f_e$ is
\[ \left( \obj^{\calG_2}_e \right)'(f_e) =
  \frac{\mathsf{d}}{{\mathsf{d}}f_e} \obj^{\calG_2}_e(f_e) =
  \gg^{\calG_2}_e - (2 \rr_e^{\calG_2} + ps^{\calG_2}
  \abs{f_e}^{p-2})f_e. \] First observe that if $f^\star_e$ is the
optimal solution, it must have the same sign as $\gg_e^{\calG_2}.$
Without loss of generality, we assume that $\gg^{\calG_2}_e \ge 0.$
Observe that for $f_e \ge \frac{\gg^{\calG_2}_e}{2\rr_e^{\calG_2}},$
we have $\frac{\mathsf{d}}{{\mathsf{d}}f_e} \obj^{\calG_2}_e \le 0.$
Thus, $f^{\star}_e \le \frac{\gg^{\calG_2}_e}{2\rr_e^{\calG_2}}.$
Similarly, we have,
$f^{\star}_e \le \left( \frac{\gg^{\calG_2}_e}{ps^{\calG_2}}
\right)^{\frac{1}{p-1}},$ where $f^{\star}_e$ is the optimal
solution. Thus, if we define $z$ as
\[z \defeq \min\left\{ \frac{\gg^{\calG_2}_e}{2\rr_e^{\calG_2}},
  \left( \frac{\gg^{\calG_2}_e}{ps^{\calG_2}} \right)^{\frac{1}{p-1}}
\right\},\]
then $f^{\star}_e \le z.$

Moreover, for $f_e \le \frac{z}{2},$ we have,
\[ \left( \obj^{\calG_2}_e \right)'(f_e)
  \ge   \gg^{\calG_2}_e - \frac{2 z \rr^{\calG_2}_e}{2} - \frac{p
    s^{\calG_2} z^{p-1}}{2^{p-1}}
  \ge \gg^{\calG_2}_e - \frac{\gg^{\calG_2}_e}{2} -
  \frac{\gg^{\calG_2}_e}{2^{p-1}} \ge 0.
\]
Thus, $f^{\star}_e \ge \frac{z}{2},$ and hence $z$ gives a
2-approximation to $f^{\star}$ that can be computed in $O(1)$
time. Now, applying binary search allows us to find
$f_e \in [(1-\nfrac{\delta}{p})f_e^{\star},
(1+\nfrac{\delta}{p})f^{\star}_e]$ in $O(\log \nfrac{1}{\delta})$
time. Now, we show that
such an estimate is good enough. Consider the point $\frac{3}{4}z.$
We have
\[ \max_{f_e}
  \obj^{\calG_2}_e(f_e) \ge \obj^{\calG_2}_e(\frac{3}{4}z) \ge
  \frac{3}{4} \gg_e^{\calG} z(1 - \frac{1}{2}\frac{3}{4} - \frac{1}{p}
  \frac{3^{p-1}}{4^{p-1}}) \ge \frac{1}{4}z \gg_e^{\calG}.
\]
Now,
\begin{align*}
  \obj^{\calG_2}_e(f_e) - \obj^{\calG_2}_e(f_e^{\star})
  & \le \delta f_e^{\star} \max\left\{
    \abs{\left( \obj^{\calG_2}_e \right)'((1-\delta)f^{\star}_e)},
    \abs{\left( \obj^{\calG_2}_e \right)'((1+\delta)f^{\star}_e)}
    \right\} \\
  & \qquad \qquad \text{(Using mean-value theorem and concavity)} \\
  & \le \delta f_e^{\star} \max\left\{
    \gg_e^{\calG}, 
    - \gg_e^{\calG} + (1+\delta)2\rr_e^{\calG}f^{\star}_e +
    (1+\delta)^{p-1} p s^{\calG} \abs{f^{\star}_e}^{p-1}
    \right\} \\
  & \le \delta f_e^{\star} \max\left\{
    \gg_e^{\calG}, 
    - \gg_e^{\calG} + (1+\delta) \gg_e^{\calG} +
    (1+\delta)^{p-1} \gg_e^{\calG} 
    \right\} \\
  & \le 4 \delta f_e^{\star} \gg_e^{\calG} \qquad \text{(Using $\delta
    \le \nfrac{1}{p}$)}
  \\
  &\le 4 \delta z
    \gg_e^{\calG} \le 16 \delta  \max_{f_e}
    \obj^{\calG_2}_e(f_e)
\end{align*}
Rewriting, we get
$\obj^{\calG_2}_e(f_e) \ge (1-16\delta) \max_{f_e}
\obj^{\calG_2}_e(f_e).$ Rescaling $\delta,$ we obtain our claim.

We can compute such an estimate for all the edges in
$O(\abs{E^{\calG_2}} \log \nfrac{1}{\delta})$ time. 
\end{proof}


\section{Sparsifying Uniform Expanders}
\label{sec:expander}

We now verify that sparsifying that sampling $\alpha$-uniform
expanders preserve the objectives of the optimizations.
Pseudocode of our routine and the flow maps constructed by it
are in Algorithm~\ref{alg:SampleAndFixGradient}.
We remark that the maps are identical to the ones used for
flow sparsifiers by Kelner et al.~\cite{KelnerLOS14}.

\begin{algorithm}[H]
\caption{Producing Sparsifier}
\label{alg:SampleAndFixGradient}
 \begin{algorithmic}[1]
 \Procedure{SampleAndFixGradient}{$\calG=(G,r^{\calG},s^{\calG},\gg^{\calG}),\tau$}
 \State{Initialize $\calH$ with $V^{\calH} = V^{\calG}$.}
 \State{\label{alg:SampleAndFixGradient:samp}
Sample each edge of $E^{\calG}$ independently
   w. probability $\tau$ to form $E^{\calH}$.}
 \State{\label{alg:SampleAndFixGradient:scaling}
Let $r^{\calH} \leftarrow \tau \cdot r^{\calG}$ and $s^{\calH} = \tau^p \cdot s^{\calG}$}
 \State{\label{alg:SampleAndFixGradient:projG}
Compute the decomposition $\gg^{\calG} =\gghat^{\calG}+\BB^{\calG} \ppsi$,
s.t. $\gghat^{\calG}$ is the cycle-space
 projection of $\gg^{\calG}$.} 
 \State{Let $\ggtil^{\calH} \leftarrow (\gghat^{\calG})_{|F}$, i.e $\ggtil^{\calH}$ is the
 restriction of $\gghat^{\calG}$ to $F$.}
 \State{\label{alg:SampleAndFixGradient:projH}
Let $\gghat^{\calH} \leftarrow\left(I -\BB^{\calH} \left(\BB^{H \top} \BB^{\calH}
   \right)^{\dag} \BB^{H\top} \right)\ggtil^{\calH}$ i.e. the cycle-space
 projection of $\ggtil^{\calH}$.} 
 \State{\label{alg:SampleAndFixGradient:pot}
Let $\gg^{\calH} \leftarrow \gghat^{\calH}+\BB^{\calH} \ppsi$}
 \State{\label{alg:SampleAndFixGradient:mapHG}
Let $\map{\calG}{\calH}$ be the map 
   $\ff \to \BB^{\calH} \left(
     \BB^{H \top} \BB^{\calH} \right)^{\dag} \BB^{\calG\top} \ff 
+ \frac{1}{\norm{\gghat^{\calH}}_2^2} \gghat^{\calH} \gghat^{\calG\top} \ff 
$ }
 \State{\label{alg:SampleAndFixGradient:mapGH}
Let $\map{\calH}{\calG}$ be the map 
   $\ff \to \BB^{\calG} \left(
     \BB^{\calG\top} \BB^{\calG} \right)^{\dag} \BB^{H\top} \ff 
+ \frac{1}{\norm{\gghat^{\calG}}_2^2} \gghat^{\calG} \gghat^{H\top} \ff 
$ }
\State \Return $\calH=(V^{\calH}, E^{\calH},r^{\calH},s^{\calH},\gg^{\calH}),  \map{\calG}{\calH},  \map{\calH}{\calG}$
 \EndProcedure 
 \end{algorithmic}
\end{algorithm}
Note that the decomposition in Line~\eqref{alg:SampleAndFixGradient:projG}
can be found by first computing $\ppsi = \left(\BB^{\calG\top} \BB^{\calG}
   \right)^{\dag} \BB^{\calG\top} \gg^{\calG}$.
The only randomness in the Algorithm~\ref{alg:SampleAndFixGradient} is
in the sampling in Line~\eqref{alg:SampleAndFixGradient:samp}.
In Lines~\eqref{alg:SampleAndFixGradient:projG}, and \eqref{alg:SampleAndFixGradient:projH}-\eqref{alg:SampleAndFixGradient:mapGH},
when the pseudo-inverse of a Laplacian is applied, we can rely
deterministically on a high-accuracy approximation based on the fact
that if the earlier sampling succeeded, both matrices are Laplacians
of expanders, and hence well-conditioned.
Alternatively, we can call a high-accuracy Laplacian solver.
This encurs another small failure probability. In either case, we can
ensure that an implicit representation of the operator is only
computed once, and succeeds with high probability.

%

\ExpanderSparsify*

\begin{remark}
  If in Algorithm~\ref{alg:SampleAndFixGradient}, the input gradient
  $\gg^{\calG}$ has zero cycle-space projection, i.e. $\gghat^{\calG} = \vzero$,
  then the cycle-space gradient terms in
  Lines~\eqref{alg:SampleAndFixGradient:mapHG} and
  \eqref{alg:SampleAndFixGradient:mapGH} should be set to zero, so
  that
  $\map{\calG}{\calH}$ is the map
   $\ff \to \BB^{\calH} \left(
     \BB^{H \top} \BB^{\calH} \right)^{\dag} \BB^{\calG\top} \ff $ and 
  $\map{\calH}{\calG}$ is the map
   $\ff \to \BB^{\calG} \left(
     \BB^{\calG\top} \BB^{\calG} \right)^{\dag} \BB^{H\top} \ff $. 
   The proof of this case is simpler, we omit all terms that deal with
   cycle-space projected gradients and everything else stays the same
   as in the proof given in this section.
\end{remark}
To prove this theorem, we first collect a number of
observations that will help us.
The most basic of these is that 
Line~\eqref{alg:SampleAndFixGradient:samp} succeeds
in producing a sparsifier in the spectral approximation sense,
and with edge set $F$
satisfying $0.5 \tau m \leq |F| \leq 2 \tau m$.
This is a direct consequence of matrix concentration bounds~\cite{Tropp12}.

\begin{lemma}
\label{lem:L2Operator}
Consider the edge-vertex incidence matrices with
gradients (projected via $\ppsi^{\calG}$)
appended as an extra column for both $\calG$ and $\calH$,
$[\BB^{\calG}, \gghat^{\calG}]$ and $[\BB^{\calH}, \ggtil^{\calH}]$.
With high probability we have that
for any vector $\xx$
\begin{equation}
\tau 
\norm{
\left[ \BB^{\calG}, \gghat^{\calG}\right] \xx
}_2^2
\approx_{0.1}
\norm{
\left[ \BB^{\calH},
\ggtil^{\calH}\right] \xx
}_2^2
\label{eq:expanderSampConc}
\end{equation}
and
the edge set $F$ of $H$ satisfies
$0.5 \tau m \leq |F| \leq 2 \tau m$.
\end{lemma}

\begin{proof}
The bounds on $\abs{F}$ follow from a scalar Chernoff bound.
For the matrix approximation bound, we will invoke
matrix Chernoff bounds~\cite{Tropp12}, which give such
a bound as long as the rows of $[\BB^{\calG}, \gghat^{\calG}]$
are sampled with probaiblity exceeding $c_{sample} \log{n}$
times their leverage scores.

So it suffices to bound the leverage scores of the
rows of this matrix.
As $\BB^{\calG}$ and $\gghat^{\calG}$ are orthogonal to each other,
we can bound the leverage scores of the rows in these
two matrices and add them.

The fact that the graph $(V^{\calG}, E^{\calG})$ has
expansion $\phi$ means that
its normalized Laplacian has eigenvalue at least
$\phi^{-2}$.
So the leverage score of a row of $\BB^{\calG}$ is at least
$\phi^{-2} d_{\min}$.
The leverage score of $\gghat^{\calG}_e$ in $\gghat$ on
the other hand is at most $\alpha / m$ due to the
$\alpha$-uniform assumption.
Thus, the sampling probablity $\tau$ meets the requirements
of matrix Chernoff bounds, and we get the approximation with
high probability.

\end{proof}

\begin{corollary}
Assuming Equation~\eqref{eq:expanderSampConc}, the graphs
underlying $\calG$ and $\calH$ (with resistances 
$r^{\calG}$ and $r^{\calH}$) are spectral approximations
of each other:
\begin{equation}
  \label{eq:expanderSampConc:graph}
\tau \BB^{\calG\top}\BB^{\calG} \approx_{0.1} \BB^{H\top}\BB^{\calH} 
\end{equation}
and the subset of gradient terms chosen after
rescaling, $\ggtil^{\calH}$, has $\ell_{2}^{2}$
norm that's bigger by a factor of about $\tau$:
\begin{equation}
\label{eq:expanderSampConc:grad}
\tau \norm{\gghat^{\calG}}_{2}^2
\approx_{0.1}
\norm{\ggtil^{\calH}}_{2}^2
.
\end{equation}
\end{corollary}

\begin{proof}
The approximation of graphs follows from considering
vectors $\xx$ with $0$ in the last coordinate.

The approximation of $\ell_2^2$ norms of vectors follow
from considering the indicator vector with $1$ in the
last column and $0$ everywhere else.
\end{proof}

From this spectral approximation, we can
also conclude that $(V^{\calH}, E^{\calH})$
must be an expander, as captured by the next corollary.
\begin{corollary}
\label{cor:Hconductance}
Assuming Equation~\eqref{eq:expanderSampConc}, $H$ has conductance at least $0.8 \phi$.
\end{corollary}
\begin{proof}
Let $C_{\calG}(S)$ and $C_{\calH}(S)$ denote the number of
edges of $E^{\calG}$ and $E^{\calH}$ respectively
crossing a cut $S \subseteq V^{\calG} = V^{\calH}$.
  Condition~\eqref{eq:expanderSampConc:graph} implies that cuts are
  preserved between $(V^{\calG}, E^{\calG})$
  and $(V^{\calH}, E^{\calH})$: For all $S \subseteq V$
  $\tau C_{\calG}(S) \approx_{0.1} C_{\calH}(S) $,
  by computing the quadratic form in an indicator vector of $S$.
  
  The degree of every vertex is also preserved, to up a
  scaling of $\tau$ and a multiplicative error
  $1\pm0.1$, i.e.
  $\tau \deg_{\calH}(v) \approx_{0.1} \deg_{\calG}(v)  $.
  This follows from considering the quadratic form of 
  Condition~\eqref{eq:expanderSampConc:graph} in the indicator vector
  of vertex $v$.
  This implies for any $S$,
\[
\frac{C_{\calH}(S)}{\sum_{v \in S} \deg_{\calH}(v)}
    \approx_{0.2} \frac{C_{\calG}(S)}{\sum_{v \in S} \deg_{\calG}(v)},
\]
from which we conclude the conductance is
preserved up to a factor of $0.8$.
\end{proof}

We can also conclude from this that $\gghat^{\calH}$ is
well-spread.

\begin{corollary}
\label{cor:gghatapx} 
Assuming Equation~\eqref{eq:expanderSampConc}, 
the projection of $\gg^{\calH}$ onto the cycle-space of $H$,
$\gghat^{\calH}$ has $\ell_{1}$ and $\ell_2$ norms
that are close to $\tau$ times the corresponding
terms in $\calG$:
\begin{equation}
  \label{eq:gradprojapx2norm}
  \norm{\gghat^{\calH}}_{2}^2 \approx_{0.2} \tau
  \norm{\gghat^{\calG}}_{2}^2
\end{equation}
\begin{equation}
  \label{eq:gradprojapx1norm}
  \norm{\gghat^{\calH}}_{1} \approx_{O( \alpha \phi^{-6} \log^2{n} )} \tau
  \norm{\gghat^{\calG}}_{1} 
\end{equation}
And $\gghat^{\calH}$ is $O(\alpha \phi^{-6} \log^2{n})$-well spread,
i.e. (as $|F|$ is the number of entries of $\gghat^{\calH}$)
\begin{equation}
\label{eq:gradprojapxwellspread}
\norm{\gghat^{\calH}}_{\infty}^2
\leq
\frac{O(\alpha \phi^{-6} \log^2{n})}{\abs{F}} 
\norm{\gghat^{\calH}}_{2}^2 
.
\end{equation}
\end{corollary}
\begin{proof}
  We first show $ \norm{\gghat^{\calH}}_{2}^2 \approx_{0.5} \tau
  \norm{\gghat^{\calG}}_{2}^2$.
Note that $\norm{\gghat^{\calG}}_2^2 = \norm{{\BB^{\calG}} \vzero + \gghat^{\calG}}_2^2$
Now, consider
$\xx =
\begin{pmatrix}
\vzero \\
1
\end{pmatrix}
$.
 By Equation~\eqref{eq:expanderSampConc},
\begin{align*}
\tau \norm{\gghat^{\calG}}_2^2 
=
\tau 
\norm{
\left[ \BB^{\calG}, \gghat^{\calG}\right] \xx
}_2^2
&
\approx_{0.1}
\norm{
\left[\BB^{\calH},
\ggtil^{\calH}\right] \xx
}_2^2
\\
&
\geq 
\min_{\yy} \norm{
\left[ \BB^{\calH},
\ggtil^{\calH}\right] 
  \begin{pmatrix}
    \yy \\
    1
  \end{pmatrix}
}_2^2
=
\norm{\gghat^{\calH}}_2^2.
\end{align*}
Thus $\tau \norm{\gghat^{\calG}}_2^2 \geq 0.9 \norm{\gghat^{\calH}}_2^2.$

The definition of $\gghat^{\calG}$ ensures
$\norm{\gghat^{\calG}}_2^2 = \min_{\yy} \norm{{\BB^{\calG}} \yy +
  \gghat^{\calG}}_2^2$.
Letting $ \yy^{\calH} \in \arg\min_{\yy} \norm{{\BB^{\calH}} \yy +
  \ggtil^{\calH}}_2^2$, we get from the definition of $\gghat^{\calH}$, that 
\begin{align*}
\norm{\gghat^{\calH}}_2^2
=
\norm{
\left[ \BB^{\calH},
 \ggtil^{\calH}\right]
\begin{pmatrix}
    \yy^{\calH} \\
    1
\end{pmatrix}
}_2^2
&
\approx_{0.1}
\tau
\norm{
\left[\BB^{\calG}, \gghat^{\calG}\right] 
\begin{pmatrix}
    \yy^{\calH} \\
    1
\end{pmatrix}
}_2^2
\\
&
\geq 
\tau
\min_{\yy} \norm{
\left[
\BB^{\calG},
\gghat^{\calG}\right] 
\begin{pmatrix}
    \yy \\
    1
\end{pmatrix}
}_2^2
=
\tau \norm{\gghat^{\calG}}_2^2.
\end{align*}
Thus we also have $\tau \norm{\gghat^{\calH}}_2^2 \geq 0.9 \norm{\gghat^{\calG}}_2^2$,
allowing us to conclude that Equation~\eqref{eq:gradprojapx2norm} is
satisfied.

Next, observe that by Lemma~\ref{lem:electricalOblInfRoute}, since $H$
has conductance at least $0.8 \phi$,
\begin{align*}
\norm{\gghat^{\calH}}_\infty
= 
\norm{\left(
I-\BB^{\calH} \left( {\BB^{\calH}}^{\top} \BB^{\calH}
  \right)^{\dag} {\BB^{\calH}}^{\top}
\right)\ggtil^{\calH}}_\infty
&\leq 
\norm{\left(
I-\BB^{\calH} \left( {\BB^{\calH}}^{\top} \BB^{\calH}
  \right)^{\dag} {\BB^{\calH}}^{\top}
\right)}_{\infty} \norm{\ggtil^{\calH}}_\infty
\\
&\leq 
O(\phi^{-3} \log n) \norm{\gghat^{\calG}}_\infty,
\end{align*}
where in the last step we also used $\norm{\ggtil^{\calH}}_\infty \leq
\norm{\gghat^{\calG}}_\infty$, since the former vector consists of a
subset of the entries of the latter.
Furthermore, by combining the above inequality with the assumption that
$\gghat^{\calG}$ is $\alpha$-well-spread, and
Equation~\eqref{eq:gradprojapx2norm} holds, we get
\[
\norm{\gghat^{\calH}}_\infty^2
\leq 
O(\phi^{-6} \log^2 n) \frac{\alpha}{m} \norm{\gghat^{\calG}}_2^2
\leq
\frac{
O(\alpha \phi^{-6} \log^2 n) 
}{\tau m} \norm{\gghat^{\calH}}_2^2.
\]
As $\gghat^{\calH}$ has $\abs{F} \leq 2 \tau m$ entries,
this shows that it is $O(\alpha \phi^{-6} \log^2 n)$-well-spread,
which establishes Equation~\eqref{eq:gradprojapxwellspread}.

Next, to prove that Equation~\eqref{eq:gradprojapx1norm} holds, we
first observe that for any $\gamma$-well-spread vector on $\xx$ with
$t$ coordinates, since $\norm{\xx}_1 \norm{\xx}_\infty \geq \norm{ \xx
}_2^2 $, 
\[
\norm{\xx}_1 \geq
\frac{ \norm{ \xx }_2^2 }{ \norm{\xx}_\infty } \geq \frac{t}{\gamma}
\norm{\xx}_\infty \geq \frac{t^{1/2}}{\gamma}
\norm{\xx}_2
.
\]
So $\frac{t^{1/2}}{\gamma} \norm{\xx}_2 \leq \norm{\xx}_1 \leq t^{1/2}
\norm{\xx}_2 $.
As $\gghat^{\calG}$ and $\gghat^{\calH}$ are $\alpha$ and $O(\alpha \phi^{-6} \log^2
n)$-well-spread respectively,
we then get
\[
\Omega(1)
\frac{
\frac{
\abs{F}
}{
(\alpha \phi^{-6} \log^2 n)^2
}
\norm{\gghat^{\calH}}_2^2
}
{
m \norm{\gghat^{\calG}}_2^2
}
\leq 
\frac{ \norm{\gghat^{\calH}}_1^2 }
{ \norm{\gghat^{\calG}}_1^2}
\leq 
O(1)
\frac{
\abs{F}
\norm{\gghat^{\calH}}_2^2
}
{
\frac{
m
}{
\alpha^2
} \norm{\gghat^{\calG}}_2^2
}
.
\]
Combining this with 
$\norm{\gghat^{\calH}}_{2}^2 \approx_{0.2} \tau \norm{\gghat^{\calG}}_{2}^2$,
and $0.5 \tau \leq \frac{|F|}{m} \leq 2 \tau $, 
we get 
\[
\Omega(1)
\frac{
\tau^2
}{
(\alpha \phi^{-6} \log^2 n)^2
}
\leq 
\frac{ \norm{\gghat^{\calH}}_1^2 }
{ \norm{\gghat^{\calG}}_1^2}
\leq 
O(1)
\alpha^2 \tau^2
.
\]
From this we can directly conclude that
Equation~\eqref{eq:gradprojapx1norm} holds.
\end{proof}

Also, we can show that for any $\bb$ and $\theta$,
the optimum $\ell_2$ as well as $\ell_{p}$
energies are close to the per degree lower bounds.
We start with the lower bounds.

\begin{lemma}
\label{lem:ExpanderLower}
Consider any graph with degrees $\DD$,
uniform $r$ and $s$, gradient $\gg$
decomposable into
$
\gg = \gghat + \BB \ppsi
$
where $\gghat$ is the cycle-space projection of $\gg$,
and any $\theta$,
any flow $\ff$ such that:
\begin{enumerate}
\item $\ff$ has residues $\bb$: $\BB^{\top} \ff = \bb$,  and
\item $\ff$ has dot product $\theta + \bb^{\top} \ppsi$
with $\gg$, i.e.
$\gg^{\top} \ff = \theta + \bb^{\top} \ppsi$
\end{enumerate}
must satisfy
\[
\sum_{e} \rr_e \ff_e^2 + \ss_{e} \abs{\ff_e}^{p}
\geq
\Omega\left(
r \cdot \norm{\bb}_{\DD^{-1}}^{2}
+
r \cdot \frac{\theta^2}{\norm{\gghat}_2^2}
+
s \cdot \norm{\DD^{-1} \bb}_{\infty}^{p}
+
s \cdot \left(\frac{\theta}{\norm{\gghat}_1} \right)^{p}
\right).
\]
\end{lemma}

\begin{proof}
First, note that because $\BB^{\top} \ff = \bb$,
we have
\[
\gghat^{\top} \ff
=
\left( \gg - \BB \ppsi \right)^{\top} \ff
= \gg^{\top} \ff
- \ppsi^{\top} \left( \BB^{\top} \ff \right)
= \gg^{\top} \ff - \xx^{\top} \bb
= \theta.
\]
That is, the dot of $\ff$ against $\gghat$
must be $\theta$.

The total energy is a sum of the $\ell_2^2$
and $\ell_{p}^{p}$ terms.
First, we will give two different lower bounds on the $\ell_2^2$, and
can hence also conclude that the average of the two lower bounds is
another lower bound.
We then do the same for the $\ell_{p}^{p}$ terms and add the lower
bounds together for a lower bound on the overall objective.
We will do so separately, by matching the $\ell_{2}^2$
terms to the electrical energy, and the $\ell_{p}^{p}$
terms to the minimum congestion.
\begin{enumerate}
\item $\norm{\ff}_2^2 \geq \norm{\bb}_{\DD^{-1}}^2$:
here we use the fact that the minimum energy of the
electrical flow is given by
\[
\bb^{\top} \LL^{\dag} \bb,
\]
and that the graph Laplacian is dominated by twice
its diagonal
\[
\LL \preceq 2 \DD,
\]
to get
\[
\norm{\bb}_{\LL^{\dag}}^2
\geq
\norm{\bb}_{\frac{1}{2} \DD^{-1}}^2.
\]
\item $\norm{\ff}_2^2
\geq \frac{\theta^2}{\norm{\gghat}_2^2}$ is by rearranging
Cauchy-Schwarz inequality, which in its simplest form gives
\[
\norm{\ff}_2 \cdot \norm{\gghat}_2
\geq
\abs{\ff^{\top} \gghat} = \abs{\theta} 
.
\]
\item $\norm{\ff}_{p}^{p}
\geq \norm{\DD^{-1} \bb}_{\infty}^{p}$
is because if we have a residue of $\bb_{u}$
at some vertex, then some edge incident to $u$
must have flow at least
\[
\frac{\bb_{u}}{\dd_{u}}
\]
on it.
The $p$-th power of that lower bounds the overall
$p$-norm energy.
\item $\norm{\ff}_{p}^{p}
\geq (\frac{\abs{\theta}}{\norm{\gghat}_1})^{p}$ uses
a similar lower bound on $\norm{\ff}_{\infty}$,
except using Holder's inequality on $\ell_{\infty}$
and $\ell_{1}$ norms to obtain
\[
\norm{\ff}_{\infty} \norm{\gghat}_{1}
\geq
\abs{\theta},
\]
which rearranges to give
\[
\norm{\ff}_{\infty}
\geq
\frac{\abs{\theta}}{\norm{\gghat}_{1}}.
\]
\end{enumerate}
\end{proof}

Before proving upper bounds on the energy required to route a flow
in the graph, we state a lemma that upper bounds the energy required
to route the ``electrical'' component of the flow, i.e. the projection
of the flow orthogonal to the cycle space.

\begin{lemma}
\label{lem:electricalFlow2normUB}
  Consider a graph $G$ with degrees $\DD$,
conductance $\phi$, and edge-vertex incidence matrix $\BB$,
and any demand  $\bb \bot \vone$.
Define the electrical flow 
$\ff = \BB^{} \left( \BB^{ \top} \BB^{} \right)^{\dag} \bb$.
Then 
$\sum_{e} \ff_e^2 \leq 2\phi^{-2} \norm{\bb}_{\DD^{-1}}^2$.
\end{lemma}

The proof relies on first Cheeger's Inequality
(e.g. see \cite{SpielmanCheegerLect18}):
\begin{theorem}(Cheeger's Inequality)
\label{thm:cheeger}
  Consider a graph $G$ with degrees $\DD$,
conductance $\phi$, 
and adjancency matrix $\AA$.
Then
\[
0.5 \phi^2
\leq
\min_{\xx \bot \DD \vone}
\frac{\xx^{\top} (\DD - \AA ) \xx}
{\xx^{\top} \DD \xx }
\leq 
2\phi
\]
\end{theorem}

We will also need the following helpful fact.
\begin{fact}
\label{fac:pinvproduct}
  Suppose $\MM = \XX\AA\XX^\top$ where $\AA$ is symmetric and $\XX$ is
  non-singular, and that  $\PP$ is the projection orthogonal to the
  kernel of $\MM$, i.e. $\PP = \MM^\dag \MM = \MM\MM^\dag.$
 Then $\MM^\dag = \PP \XX^{-\top}\AA^\dag \XX^{-1}\PP$.
\end{fact}
\begin{proof}[Proof of Lemma~\ref{lem:electricalFlow2normUB}]
Note $\BB^{ \top} \BB^{} = \DD - \AA =\LL$, where $\AA$ is the adjacency
matrix of the graph, and $\LL$ is its Laplacian.
  Also 
$\sum_{e} \ff_e^2 = (\BB^{} \left( \BB^{ \top} \BB^{} \right)^{\dag}
  \bb)^{\top} \BB^{} \left( \BB^{ \top} \BB^{} \right)^{\dag} \bb
  = \bb^{\top} \left( \BB^{ \top} \BB^{} \right)^{\dag} \bb
  =\bb^{\top} \LL^{\dag} \bb
.
$
By Theorem~\ref{thm:cheeger}, we get that for $\xx \bot \DD \vone$
\[
\xx^\top \LL \xx \geq 0.5 \phi^2 \xx^\top \DD \xx. 
\]
Substituting $\yy = \DD^{1/2} \xx$ changes the constraint to
$\DD^{-1/2} \yy \bot \DD \vone$ i.e. $ \yy \bot \DD^{1/2} \vone$.
The inequality now states
\[
\yy^\top \DD^{-1/2} \LL \DD^{-1/2}  \yy \geq 0.5 \phi^2 \yy^\top \yy. 
\]
If we let $\QQ$ denote the projection orthogonal to $\DD^{1/2} \vone$, we
can summarize the inequality and orthogonality constraint in one
condition using the  Loewner order as 
\[
\QQ \DD^{-1/2} \LL \DD^{-1/2}  \QQ \succeq 0.5 \phi^2 \QQ. 
\]
Note that the null space of $\DD^{-1/2} \LL \DD^{-1/2}$ is spanned by
$\DD^{1/2} \vone$, as $\vone$ spans the null space of $\LL$. 
So in fact $\QQ \DD^{-1/2} \LL \DD^{-1/2}  \QQ = \DD^{-1/2} \LL
\DD^{-1/2}$, and we can conclude
\[
\DD^{-1/2} \LL \DD^{-1/2} \succeq 0.5 \phi^2 \QQ. 
\]
From this we conclude that, 
\[
(\DD^{-1/2} \LL \DD^{-1/2})^\dag
\preceq
 2\phi^{-2}\QQ^\dag 
\]
as $\AA \succeq \BB$ implies $\AA^\dag
\preceq \BB^\dag$ when $\AA$ and $\BB$ have the same null space.
Hence by Fact~\ref{fac:pinvproduct} and $\QQ = \QQ^\dag$, we then get 
\[
\QQ \DD^{1/2} \LL^\dag \DD^{1/2} \QQ \preceq 2\phi^{-2}\QQ
.
\]
This we can rewrite as for all $ \yy \bot \DD^{1/2} \vone$.
\[
\yy^\top \DD^{1/2} \LL^\dag \DD^{1/2} \yy \leq 2\phi^{-2}  \yy^\top \yy
.
\]
Substituting $\zz = \DD^{1/2} \yy$ changes the constraint to
$\DD^{-1/2} \zz \bot \DD^{1/2} \vone$ i.e. $ \zz \bot \vone$.
Thus we have that for all $ \zz \bot \vone$.
\[
\zz^\top\LL^\dag \zz \preceq 2\phi^{-2}  \zz^\top \DD^{-1}\zz
.
\]
Taking $\zz = \bb$, we then get $\bb^{\top} \left( \BB^{ \top}
  \BB^{} \right)^{\dag} \bb \leq 2\phi^{-2} \norm{\bb}_{\DD^{-1}}^2$.
\end{proof}

\begin{lemma}
\label{lem:electricalFlowInfnormUB}
  Consider a graph $G$ on $n$ vertices with degrees $\DD$,
conductance $\phi$, and edge-vertex incidence matrix $\BB$,
and any demand  $\bb \bot \vone$.
Define the electrical flow 
$\ff = \BB^{} \left( \BB^{ \top} \BB^{} \right)^{\dag} \bb$.
Then 
$\norm{\ff}_{\infty}  \leq O( \phi^{-3} \log(n)) \norm{\DD^{-1}\bb}_{\infty}$.
\end{lemma}

\begin{proof}
We first note that if $\ff^*$ is the optimal routing of $\bb$ in $G$,
then
\[
\norm{ \DD^{-1} \bb}_{\infty} \leq \norm{\ff^*}_{\infty} \leq \phi^{-1} \norm{ \DD^{-1} \bb}_{\infty}, 
\] 
as per Example 1.4 of \cite{Sherman13}.
Secondly, we note that by Lemma~\ref{lem:electricalOblInfRoute}, the
electrical flow
 $\ff^{\calE} 
= \BB^{} \left( {\BB^{}}^{\top} \BB^{}
  \right)^{\dag} \bb
=
\BB^{} \left( {\BB^{}}^{\top} \BB^{}
  \right)^{\dag} {\BB^{}}^{\top} \ff^*
$
satisfies
\[
\norm{\ff^{\calE}}_{\infty}
=
\norm{\BB^{} \left( {\BB^{}}^{\top} \BB^{}
  \right)^{\dag} {\BB^{}}^{\top} \ff^*}_{\infty}
\leq
\norm{\BB^{} \left( {\BB^{}}^{\top} \BB^{}
  \right)^{\dag} {\BB^{}}^{\top}}_{\infty\to\infty}
\norm{\ff^*}_{\infty} 
\leq
O(\phi^{-3} \log n ) \norm{ \DD^{-1} \bb}_{\infty}.
\]
\end{proof}

\begin{lemma}
\label{lem:ExpanderElectricalUpper}
On an expander $G$ with degrees $\DD$,
conductance $\phi$,
and gradient $\gg$ whose projection into
the cycle space of $G$, $\gghat$ is $\alpha$-well-spread,
for any demand $\bb \bot \vone$ and dot $\theta$ with $\gghat$,
the flow given by
\begin{equation}
\ff = \BB^{} \left( \BB^{ \top} \BB^{} \right)^{\dag} \bb
+ \frac{\theta}{\norm{\gghat}_2^2} \gghat
\label{eq:ElectricalRoute}
\end{equation}
 satisfies
 \begin{align*}
\sum_{e} r \ff_e^2
+
\sum_{e} s \abs{\ff_e}^p
\leq &
\\
O_{p} \left( 
r \cdot \phi^{-2} \norm{\bb}_{\DD^{-1}}^{2}
\right.
&
\left.
+
r \cdot\left( \frac{\abs{\theta}}{\norm{\gghat}_2} \right)^2
+
s \cdot m \cdot \left( \phi^{-3} \log n 
    \norm{\DD^{-1} \bb}_{\infty} \right)^{p}
+
s \cdot m \cdot
\left(\frac{\abs{\theta} \alpha^{1/2}}{\norm{\gghat}_1} \right)^{p}
\right).
 \end{align*}
\end{lemma}
\begin{proof}

We first bound the quadratic term $\sum_{e} r \ff_e^2$.
Let us write $\ff^\calE = \BB^{} \left( \BB^{ \top} \BB^{} \right)^{\dag}
\bb$ and $\ff^\calC= \frac{\theta}{\norm{\gghat}_2^2} \gghat$,
and note (in fact, appealing to
 orthogonality would save an additional factor of 2)
\[
\sum_{e} \ff_e^2
= \sum_{e} (\ff^\calE_e + \ff^\calC_e )^2
\leq 
\sum_{e} 2(\ff^\calE_e)^2 +2 (\ff^\calC_e )^2.
\]
Then we observe by Lemma~\ref{lem:electricalFlow2normUB} that 
$
\sum_{e} (\ff^\calE_e)^2
\leq \phi^{-2} \norm{\bb}_{\DD^{-1}}^2.
$
Furthermore, 
\[
\sum_{e} (\ff^\calC_e)^2 = \ff^{\calC\top}\ff^\calC =
\left( \frac{\abs{\theta}}{\norm{\gghat}_2} \right)^2
.
\]
Combining these equations gives 
\[
\sum_{e} r \ff_e^2
\leq
2r \cdot \phi^{-2} \norm{\bb}_{\DD^{-1}}^{2} +2r \cdot \left(
  \frac{\abs{\theta}}{\norm{\gghat}_2} \right)^2
.
\]
We then bound the $p$-th power term,
\[
\sum_{e} s \abs{\ff_e}^{p}
=
\sum_{e} s \abs{\ff^\calE_e+ \ff^\calC_e}^{p}
\leq
\sum_{e} s 2^p (\abs{\ff^\calE_e}^p+ \abs{\ff^\calC_e}^{p})
\leq
m s 2^p \cdot ( \norm{\ff^\calE}_{\infty}^{p} +
\norm{\ff^\calC}_{\infty}^{p} ) 
\]
Now by Lemma~\ref{lem:electricalFlowInfnormUB}, we have 
$\norm{\ff^\calE}_{\infty} 
=  \norm{\BB^{} \left( \BB^{ \top} \BB^{} \right)^{\dag} \bb}_{\infty}  
\leq 
O(\phi^{-3} \log n) \norm{\DD^{-1} \bb}_{\infty}  
$
and by the $\alpha$-well-spreadness of $\gghat$
\[
\norm{\ff^\calC}_{\infty}
=
\frac{\abs{\theta}}{\norm{\gghat}_2^2}  \norm{\gghat}_{\infty}
\leq
\frac{\abs{\theta}}{\norm{\gghat}_2^2}
\left( \frac{\alpha}{m}\norm{\gghat}_2^2 \right)^{1/2}
\leq
\frac{\alpha^{1/2} \abs{\theta}}{m^{1/2} \norm{\gghat}_2}
\leq 
\frac{\alpha^{1/2} \abs{\theta}}{\norm{\gghat}_1}
,
\]
where in the last step we used $\norm{\gghat}_1 \leq m^{1/2} \norm{\gghat}_2$.


\end{proof}

\begin{proof}[Proof of Theorem~\ref{thm:sampAndFixGrad}.]
 Refer to the pseudo-code in Algorithm~\ref{alg:SampleAndFixGradient}.
We first collect the facts that we have established about the sampling procedure.
In Line~\ref{alg:SampleAndFixGradient:samp}, $E^{\calH}$ is formed from $E^{\calG}$ by
sampling each edge independently with probability $\tau$. 
It follows that the expected number of edges
in $E^{\calH}$ is $\tau m$, and since $\tau > \log{n} / m$,
a standard scalar Chernoff bound shows that
$\abs{E^{\calH}} \leq 2 \tau m$ with high probability.
The parameters $r^{\calH} = \tau \cdot r^{\calG}$ and $s^{\calH} = \tau^p \cdot s^{\calG}$ are
set in Line~\ref{alg:SampleAndFixGradient:scaling}.
By Lemma~\ref{lem:L2Operator}, with high probability the sampling
in Line~\ref{alg:SampleAndFixGradient:samp} guarantees
Equation~\eqref{eq:expanderSampConc}.
Note also that
\begin{itemize}
\item Equation~\eqref{eq:expanderSampConc:graph} implies 
$\tau \DD^{\calH} \approx_{0.1} \DD^{\calG}$, by considering the quadratic form in
each of the standard basis vectors.
\item By Corollary~\ref{cor:gghatapx}, $\gghat^{\calH}$ is
  $O(\alpha \phi^{-6} \log^2{n})$-well-spread, and  $\norm{\gghat^{\calH}}_{2}^2
\approx
\tau \norm{\gghat^{\calG}}_{2}^2$
and 
$\norm{\gghat^{\calH}}_{1}
 \approx_{O(\alpha \phi^{-6} \log^2{n})}
\tau
\norm{\gghat^{\calG}}_{1}
.
$
\item As $G$ has conductance at least $\phi$, by 
Corollary~\ref{cor:Hconductance}, $H$ has conductance at least $0.8 \phi$.
\end{itemize}
We can now establish $\calG \preceq_{\kappa} \calH$.
 Suppose $\ff^{\calG}$ is a flow in $G$ with $\BB^{\calG} \ff^{\calG} = \bb$ and
$\gghat^{\calG\top} \ff^{\calG} = \theta$.
Then $\gg^{\calG\top} \ff^{\calG} = \theta + \ppsi^{\top} \bb$.
By Lemma~\ref{lem:ExpanderLower}, we then get that 
\begin{align*}
\sum_{e} r^{\calG} (\ff^{\calG}_e)^2
+
\sum_{e} s^{\calG} \abs{\ff^{\calG}_e}^p
\\
\geq
\Omega\left(
r^{\calG} \cdot \norm{\bb}_{(\DD^{\calG})^{-1}}^{2}
+
r^{\calG} \cdot \frac{\theta^2}{\norm{\gghat^{\calG}}_2^2}
+
s^{\calG} \cdot \norm{(\DD^{\calG})^{-1} \bb}_{\infty}^{p}
+
s^{\calG} \cdot \left(\frac{\theta}{\norm{\gghat^{\calG}}_1} \right)^{p}
\right).
\end{align*}
Applying our flow map from $\calG$ to $\calH$ 
\begin{align*}
\ff^{\calH} = \map{\calG}{\calH}(\ff^{\calG}) 
&=\BB^{\calH} \left(
     \BB^{H \top} \BB^{\calH} \right)^{\dag} \BB^{\calG\top} \ff^{\calG} 
+ \frac{1}{\norm{\gghat^{\calH}}_2^2} \gghat^{\calH} \gghat^{\calG\top} \ff^{\calG} 
\\
&=\BB^{\calH} \left(
     \BB^{H \top} \BB^{\calH} \right)^{\dag} \bb
+ \frac{1}{\norm{\gghat^{\calH}}_2^2} \gghat^{\calH} \theta
.
\end{align*}
We note that by construction, we can readily verify  $\BB^{\calH} \ff^{\calH} =
\bb$, $\gghat^{H\top} \ff^{\calH} = \theta$, and $\gg^{H\top} \ff^{\calH} =
\theta+\ppsi^{\top} \bb$.
So by applying Lemma~\ref{lem:ExpanderElectricalUpper} to $\frac{1}{\kappa}
\ff^{\calH}$, we get
\begin{align*}
\sum_{e} r^{\calH} & \left(\frac{\abs{\ff^{\calH}_e}}{\kappa}\right)^2
+
\sum_{e} s^{\calH} \left(\frac{\abs{\ff^{\calH}_e}}{\kappa}\right)^p
\\ &
\leq
O_{p} \left( 
r^{\calH} \norm{\bb}_{(\DD^{\calH})^{-1}}^{2}
\left(\frac{\phi^{-1} }{\kappa} \right)^{2} 
\vphantom{\vrule height 20pt}
+
r^{\calH} \cdot\left( \frac{\abs{\theta}}{\norm{\gghat^{\calH}}_2} \right)^2
  \kappa^{-2}
 \right.
\\
&
\left.
+
s^{\calH}  \cdot m \cdot \left( \phi^{-2} 
    \norm{(\DD^{\calH})^{-1} \bb}_{\infty} \right)^{p}\kappa^{-p}
+
s^{\calH}  \cdot m \cdot
\left(\frac{\abs{\theta} (\alpha \phi^{-6} \log^2{n})^{1/2}}{\norm{\gghat^{\calH}}_1} \right)^{p}\kappa^{-p}
\right)
\\
&
\leq
O_{p} \left( 
(\tau r^{\calG}) \cdot \frac{1}{\tau} \norm{\bb}_{(\DD^{\calG})^{-1}}^{2} \left(\frac{\phi^{-1} }{\kappa} \right)^{2} 
\vphantom{\vrule height 20pt} 
+
(\tau r^{\calG}) \cdot \frac{1}{\tau}  \cdot\left( \frac{\abs{\theta}}{\norm{\gghat^{\calG}}_2} \right)^2 \kappa^{-2}
+
\right.
\\
&
\left.
(\tau^p s^{\calG}) \cdot \tau^{-p} \cdot 
    \norm{(\DD^{\calH})^{-1} \bb}_{\infty}^{p}
\left(\frac{m^{1/p}\phi^{-2} }{\kappa} \right)^{p} 
+
(\tau^p s^{\calG}) \cdot \tau^{-p} \cdot
\left(\frac{\abs{\theta}}{\norm{\gghat^{\calG}}_1} \right)^{p}
\left(\frac{m^{1/p}(\alpha \phi^{-6} \log^2{n})^{3/2} }{\kappa} \right)^{p} 
\right)
.
\end{align*}
Our goal is to ensure 
\[
\gg^{H\top}  (\frac{1}{\kappa} \ff^{\calH})
-
\left(
\sum_{e} r^{\calH}  \left(\frac{\abs{\ff^{\calH}_e}}{\kappa}\right)^2
+
\sum_{e} s^{\calH} \left(\frac{\abs{\ff^{\calH}_e}}{\kappa}\right)^p
\right)
\leq 
\frac{1}{\kappa} 
\left(
\gg^{\calG\top} \ff^{\calG}
-
\left(
\sum_{e} r^{\calG} (\ff^{\calG}_e)^2
+
\sum_{e} s^{\calG} \abs{\ff^{\calG}_e}^p
\right)
\right)
.
\]
Because the linear terms cancel out, we can use the upper and lower
bounds established above to say that this inequality holds provided
the following conditions are satisfied (for a $C_p$ which is a constant
greater than 1 that depends on $p$):
\begin{itemize}
\item $C_p \left(\frac{\phi^{-1} }{\kappa} \right)^{2} \leq 1/\kappa$.
\item $C_p \kappa^{-2} \leq 1/\kappa$.
\item $C_p \left(\frac{m^{1/p}\phi^{-2} }{\kappa} \right)^{p} 
 \leq 1/\kappa$.
\item $C_p \left(\frac{m^{1/p}(\alpha \phi^{-6} \log^2{n})^{3/2}
    }{\kappa} \right)^{p} 
 \leq 1/\kappa$.
\end{itemize}
Recalling that $\alpha$ is a constant, it follows
that there exists a constant $C'_p$ (depending on p), s.t. all of the
above conditions are satisfied, provided 
\[
\kappa \geq C'_p \max\left( m^{1/(p-1)}\phi^{-2p/(p-1)},
  m^{1/(p-1)}\phi^{-9p/(p-1)} (\log{n})^{3p/(p-1)} ,  \phi^{-2} \right)
\]
And this in turn is implied by the stronger condition, 
$\kappa \geq C'_p ( m^{1/(p-1)} \phi^{-9} \log^{3}n )$, 
which is hence sufficient to ensure $\calG \preceq_{\kappa} \calH$.

We can then show $\calH \preceq_{\kappa} \calG$ with a very similar
calculation. We include it for completeness.
 Suppose $\ff^{\calH}$ is a flow in $H$ with $\BB^{\calH} \ff^{\calH} = \bb$ and
$\gghat^{H\top} \ff^{\calH} = \theta$.
Then $\gg^{H\top} \ff^{\calH} = \theta + \ppsi^{\top} \bb$.

By Lemma~\ref{lem:ExpanderLower}, we then get that 
\begin{align*}
\sum_{e} r^{\calH} (\ff^{\calH}_e)^2
+
\sum_{e} s^{\calH} \abs{\ff^{\calH}_e}^p
  \geq
\\
\Omega\left(
r^{\calH} \cdot \norm{\bb}_{(\DD^{\calH})^{-1}}^{2}
+
r^{\calH} \cdot \frac{\theta^2}{\norm{\gghat^{\calH}}_2^2}
+
s^{\calH} \cdot \norm{(\DD^{\calH})^{-1} \bb}_{\infty}^{p}
+
s^{\calH} \cdot \left(\frac{\theta}{\norm{\gghat^{\calH}}_1} \right)^{p}
\right).
\geq
\\
\Omega\left(
r^{\calH} \cdot \norm{\bb}_{(\DD^{\calH})^{-1}}^{2}
+
r^{\calH} \cdot \frac{\theta^2}{\norm{\gghat^{\calH}}_2^2}
+
s^{\calH} \cdot \norm{(\DD^{\calH})^{-1} \bb}_{\infty}^{p}
+
s^{\calH} \cdot \left(\frac{\theta}{\norm{\gghat^{\calH}}_1} \right)^{p}
\right).
\end{align*}
Applying our flow map from $\calH$ to $\calG$ 
\begin{align*}
\ff^{\calG} = \map{\calH}{\calG}(\ff^{\calH}) 
&=\BB^{\calG} \left(
     \BB^{\calG\top} \BB^{\calG} \right)^{\dag} \BB^{H\top} \ff^{\calH} 
+ \frac{1}{\norm{\gghat^{\calG}}_2^2} \gghat^{\calG} \gghat^{H\top} \ff^{\calH} 
\\
&=\BB^{\calG} \left(
     \BB^{\calG\top} \BB^{\calG} \right)^{\dag} \bb
+ \frac{1}{\norm{\gghat^{\calG}}_2^2} \gghat^{\calG} \theta
.
\end{align*}
Again, by construction, we have $\BB^{\calG} \ff^{\calG} =
\bb$, $\gghat^{\calG\top} \ff^{\calG} = \theta$, and $\gg^{\calG\top} \ff^{\calG} =
\theta+\ppsi^{\top} \bb$.
So by applying Lemma~\ref{lem:ExpanderElectricalUpper} to $\frac{1}{\kappa}
\ff^{\calG}$, we get
\begin{align*}
\sum_{e} r^{\calG} & \left(\frac{\abs{\ff^{\calG}_e}}{\kappa}\right)^2
+
\sum_{e} s^{\calG} \left(\frac{\abs{\ff^{\calG}_e}}{\kappa}\right)^p
\\ &
\leq
O_{p} \left( 
r^{\calG} \norm{\bb}_{(\DD^{\calG})^{-1}}^{2}
\left(\frac{\phi^{-1} }{\kappa} \right)^{2} 
\vphantom{\vrule height 20pt}
+
r^{\calG} \cdot\left( \frac{\abs{\theta}}{\norm{\gghat^{\calG}}_2} \right)^2
  \kappa^{-2}
 \right.
\\
&
\left.
+
s^{\calG}  \cdot m \cdot \left( \phi^{-2} 
    \norm{(\DD^{\calG})^{-1} \bb}_{\infty} \right)^{p}\kappa^{-p}
+
s^{\calG}  \cdot m \cdot
\left(\frac{\abs{\theta} (\alpha \phi^{-6} \log^2{n})^{1/2}}{\norm{\gghat^{\calG}}_1} \right)^{p}\kappa^{-p}
\right)
\\
&
\leq
O_{p} \left( 
(\tau^{-1} r^{\calH}) \cdot \tau \norm{\bb}_{(\DD^{\calH})^{-1}}^{2} \left(\frac{\phi^{-1} }{\kappa} \right)^{2} 
\vphantom{\vrule height 20pt} 
+
(\tau^{-1} r^{\calH}) \cdot \tau  \cdot\left( \frac{\abs{\theta}}{\norm{\gghat^{\calH}}_2} \right)^2 \kappa^{-2}
+
\right.
\\
&
\left.
(\tau^{-p} s^{\calH}) \cdot \tau^{p} \cdot 
    \norm{(\DD^{\calH})^{-1} \bb}_{\infty}^{p}
\left(\frac{m^{1/p}\phi^{-2} }{\kappa} \right)^{p} 
+
(\tau^{-p} s^{\calH}) \cdot \tau^{p} \cdot
\left(\frac{\abs{\theta}}{\norm{\gghat^{\calH}}_1} \right)^{p}
\left(\frac{m^{1/p}(\alpha \phi^{-6} \log^2{n})^{3/2} }{\kappa} \right)^{p} 
\right)
.
\end{align*}
Now, we want to guarantee
\[
\gg^{\calG\top}  (\frac{1}{\kappa} \ff^{\calG})
-
\left(
\sum_{e} r^{\calG}  \left(\frac{\abs{\ff^{\calG}_e}}{\kappa}\right)^2
+
\sum_{e} s^{\calG} \left(\frac{\abs{\ff^{\calG}_e}}{\kappa}\right)^p
\right)
\leq 
\frac{1}{\kappa} 
\left(
\gg^{H\top} \ff^{\calH}
-
\left(
\sum_{e} r^{\calH} (\ff^{\calH}_e)^2
+
\sum_{e} s^{\calH} \abs{\ff^{\calH}_e}^p
\right)
\right)
.
\]
Again the linear terms agree, and termwise verification shows that
$\kappa \geq C'_p ( m^{1/(p-1)} \log^{3}(n) \phi^{-9} )$ is
sufficient to give $\calH \preceq_{\kappa} \calG$.

\end{proof}



\section{Using Approximate Projections}
\label{sec:ApproxProj}

Finally, we need to account for the errors in computing the cycle
projections $\gghat$ of the gradients $\gg$.  This error arise due to
the use of iterative methods in Laplacian solvers used to evaluate
$(\BB^{\top} \BB)^{\dag}$.  As we only perform such projections on
expanders, we can in fact use iterative methods.  However, a
dependence of $\log(1 / \epsilon)$ in the error $\epsilon$ still
remain.

We first formalize the exact form of this error.
Kelner et al.~\cite{KelnerOSZ13} showed that a Laplacian solver
can converge in error proportional to that of the electrical
flow.
That is, for a slightly higher overhead of
$O(\log(n / \epsilon))$,
we can obtain a vector $\ggtil$ such that
\[
\norm{\gghat - \ggtil}_2
\leq
\epsilon \norm{\gghat}_2
\leq
\epsilon \norm{\gg}_2.
\]
This was also generalized to a black-box reduction between
solvers for vertex solutions and flows
subsequently~\cite{CohenKMPPRX14}.
As a result, we will work this guarantee with errors
relative to $\gghat$.

For the partitioning stage, this error occurs in two places:
for computing the norm of the projection,
and for identifying edges with high contributions
(aka. non-uniform) for removal.

For the former, a constant factor error in the norm of
$\gghat$ will only lead to a constant factor increase in:
\begin{enumerate}
\item The uniformity of the true projected gradient,
\item The factor of decrease in the norm of the projected
gradient from one step to next.
\end{enumerate}
For both of these, such constant factor slowdowns can be
absorbed by an increase in the thresholds, which in turn
result in a higher uniformity parameter in decompositions
returned.
As this uniformity parameter only affects the number of
edges sampled in Theorem~\ref{thm:sampAndFixGrad}, they
only accumulate to a larger overhead in the $m^{O(\nfrac{1}{\sqrt{p}})}$ term in the overall running time.

The other invocation of projections is in the sparsification
of expanders in Algorithm~\ref{alg:SampleAndFixGradient}.
Here the decomposition of $\gg^{\calG}$
into a circulation and potential flows is necessary for
the construction of the gradient of the sampled graph, $\calH$.

While an approximate energy minimizing circulation $\ggtil$
will not have $\gg - \ggtil$ being a potential flow,
we can instead perturb $\gg$ slightly in this instance.
Specifically, we can also compute a set of approximate
potentials $\ppsitil$ so that
\[
\norm{\gg - \left( \ggtil + \BB \ppsitil \right)}_2
\leq
\epsilon \norm{\gg}_2.
\]
That is, we can perturb the initial $\gg$ based on the
result of this solve so that
we have an exact decomposition of it into a circulation
and a potential flow.
The error of this perturbation is then incorporated in
the same manner as terminating when $\norm{\gghat}_2$
is too small in Case~\ref{case:Tiny} of Theorem~\ref{thm:Decompose}.
Specifically, the additive error of this goes into
the additive trailing terms of the guarantees of the
ultra-sparsifier shown in Theorem~\ref{thm:ultrasparsify}.

Finally, the projection of the sampled gradient $\ggtil^{\calH}$
into $\gghat^{\calH}$ also carries such an error term.
By picking $\epsilon$ to be in the $1 / poly(n)$ range,
we ensure that both the $\ell_2$ and $\ell_{1}$ norms
of $\gghat^{\calH}$ is close to their true terms.
This in turn leads to constant factor errors
in the lower and upper bounds on objectives give in 
Lemmas~\ref{lem:ExpanderLower} and~\ref{lem:ExpanderElectricalUpper},
and thus a constant factor increase in the overall
approximation factors.

Therefore, it suffices to set $\epsilon$ in these approximate
projection algorithms to be within $poly(n)$ factors of the
$\delta$ by which \textsc{UltraSparsify} is invoked by the
overall recursive preconditioning scheme.
The choice of parameters in Theorem~\ref{thm:RecPrecon}
then gives that it suffices to have
$\log(1 / \epsilon) \leq \Otil(1)$ in all projection steps.
In other words, all the projections can be performed
in time nearly-liner in the sizes of the graphs.
\section{$\ell_p$-norm Semi-Supervised Learning on Graphs.}
\label{sec:semisupervised}

In this appendix, we briefly describe how to convert
Problem~\eqref{eq:semisupervised}, 
into a form that can be solved using our algorithm for smoothed $p$-norm flows
as stated in Theorem~\ref{cor:smoothedpflows}.

Recall that formally, given a graph $G=(V,E)$ and a labelled subset of the nodes
$T \subset V$ with labels $\ss_{T} \in \rea^T$, we can write the
problem as 
\begin{equation*}
\min_{\substack{ \xx \in \Re^{V} \\ \text{ s.t. }\xx_{T} = \ss_{T}}}
\sum_{u \sim v} \abs{\xx_u - \xx_{v}}^{p}.
\end{equation*}
%
Taking a Lagrange dual now results in the problem
\[
\max_{\ff: \left( \BB^{\top} \ff \right)_{V \setminus T} = 0}
\gg^{\top} \ff
- \sum_{u \sim v} \abs{\ff_{uv}}^{q}
.
\]
where $q = \frac{1}{1-1/p}$,
and the gradient $\gg$ is given by $\gg = \BB_{:, T} \ss_{T}$.
We cannot directly solve this formulation, since the 
net incoming flow at vertices in $T$ is unknown.
However, notice that the flow is preserved at all other vertices, so
summed across all of $T$, the net flow must be zero.
Thus if we merge all the vertices in $T$ into one vertex, while
turning edges in $T \times T$ into self-loops, the problem is now a
circulation.
Note that the optimal flow on each self-loop can be computed exactly.
Now the resulting problem can be solved to high accuracy using
Theorem~\ref{cor:smoothedpflows}.
Meanwhile, mapping the flow back to the original flow, it can be shown
that the optimal flow arises as a simple non-linear function of some
voltages $\xx$: $\ff_e = (\BB\xx)_e^q$.
This means that if we have $\ff$ to high enough accuracy, we can get
an almost optimal set of voltages, e.g. by looking at flow along edges
of a tree to compute a set of voltages $\xx$ that are a
$(1+1/\poly(m))$
multiplicative accuracy 
solution to Problem~\eqref{eq:semisupervised}.
Since we call the algorithm of Theorem~\ref{cor:smoothedpflows} using
to solve a $\frac{p}{p-1}$-flow problem, where $p<2$, the running time will be on
the order of $2^{O((\frac{p}{p-1})^{\nfrac{3}{2}})}
m^{1+O(\sqrt{\frac{p-1}{p}})}$.
This in turn can be further simplified as
 $2^{O((\frac{1}{p-1})^{\nfrac{3}{2}})}
 m^{1+O(\sqrt{p-1})},$ since $p<2$.
 For $p = 1+\frac{1}{\sqrt{\log n}}$, this is time is bounded by $m^{1+o(1)}$.





\end{document}